\documentclass[letterpaper,nonanonymous]{article} %
\usepackage{aaai25}  %
\usepackage{times}  %
\usepackage{helvet}  %
\usepackage{courier}  %
\usepackage[hyphens]{url}  %
\usepackage{graphicx} %
\usepackage{natbib}  %
\usepackage{caption} %

\usepackage[linesnumbered,ruled,vlined]{algorithm2e}

\usepackage{subcaption}
\usepackage{comment}
\usepackage[dvipsnames]{xcolor}
\usepackage{amsmath,amsfonts,amssymb,amsthm}
\usepackage{comment}
\usepackage{tabularx}
\usepackage{booktabs}
\usepackage[capitalize, nameinlink]{cleveref}

\newtheorem{theorem}{Theorem}[section]

\definecolor{mygreen}{rgb}{0.1,0.4,0.2}

\newif\ifshowcomments
\showcommentstrue   %

\ifshowcomments
    \newcommand{\new}[1]{\textcolor{black}{#1}}
    \newcommand{\adri}[1]{\textcolor{olive}{[adri: #1]}}
    \newcommand{\nuria}[1]{\textcolor{magenta}{[nuria: #1]}}
    \newcommand{\georgina}[1]{\textcolor{purple}{[georgina: #1]}}
\else
    \newcommand{\new}[1]{\textcolor{brown}{#1}}
    \newcommand{\adri}[1]{}
    \newcommand{\nuria}[1]{}
    \newcommand{\georgina}[1]{}
\fi
\newcommand{\ERG}{\texttt{ERG-Link}}
\DeclareMathOperator{\CT}{\mathsf{CT}}
\DeclareMathOperator{\Rd}{\mathcal{R}_{\text{diam}}}
\DeclareMathOperator{\Rg}{\mathsf{R_{tot}}}
\DeclareMathOperator{\Bt}{\mathsf{B_R}}
\DeclareMathOperator{\diag}{diag}
\DeclareMathOperator{\Tr}{Tr}
\DeclareMathOperator*{\argmin}{arg\,min}

\title{Structural Group Unfairness:\\ Measurement and Mitigation by means of the Effective Resistance}
\author {
    Adrian Arnaiz-Rodriguez\textsuperscript{\rm 1},
    Georgina Curto\textsuperscript{\rm 2},
    Nuria Oliver\textsuperscript{\rm 1}
}
\affiliations {
    \textsuperscript{\rm 1}ELLIS Alicante, Alicante, Spain\\
    \textsuperscript{\rm 2}University of Notre Dame, Notre Dame, IN, USA\\
    \{adrian, nuria\}@ellisalicante.org, gcurtore@nd.edu
}

\begin{document}

\maketitle

\begin{abstract}
Social networks contribute to the distribution of social capital, defined as the relationships, norms of trust and reciprocity within a community or society that facilitate cooperation and collective action. Therefore, better positioned members in a social network benefit from faster access to diverse information and higher influence on information dissemination. 
A variety of methods have been proposed in the literature to measure social capital at an individual level. However, there is a lack of methods to quantify social capital at a group level, which is particularly important when the groups are defined on the grounds of protected attributes.
To fill this gap, we propose to measure the social capital of a group of nodes by means of the effective resistance and emphasize the importance of considering the entire network topology. 
Grounded in spectral graph theory, we introduce three effective resistance-based measures of group social capital, namely \textit{group isolation}, \textit{group diameter} and \textit{group control}, where the groups are defined according to the value of a protected attribute. We denote the social capital disparity among different groups in a network as \emph{structural group unfairness}, and propose to mitigate it by means of a budgeted edge augmentation heuristic that systematically increases the social capital of the most disadvantaged group.
In experiments on real-world networks, we uncover significant levels of structural group unfairness when using gender as the protected attribute, with females being the most disadvantaged group in comparison to males. We also illustrate how our proposed edge augmentation approach is able to not only effectively mitigate the structural group unfairness but also increase the social capital of all groups in the network. 
\end{abstract}

\maketitle

\section{Introduction}

Online social networks play an important role in defining and sustaining the social fabric of human communities. They allow individuals to connect, interact and share information with one another over the internet. They have opened up new opportunities for personal and professional networking, entertainment and learning. However, the formation of social networks ---whether through organic growth or recommendations--- can create imbalances in network positions which condition the access to resources and information~\citep{burt2004}. These network inequalities have an impact on the social capital of its members, which exists in the relations among individuals~\citep{coleman1988}. Better positioned network members benefit from faster access to diverse information, higher influence on information dissemination and more control of the information flow~\citep{granovetter1983strength,Burt2000,gundougdu2019bridging}. In practical terms, this means that individuals with a strategic position in the network will have more influence over others, and better access to information and opportunities regarding jobs, health, education or finance. %

Furthermore, link recommendation algorithms that pervade social media platforms tend to connect  similar users, contributing to the homophily and clustering of the network~\citep{su2016effect}. These \emph{filter bubbles} limit the access to diverse individuals~\citep{granovetter1983strength}, exacerbate the isolation and polarization of groups, reduce the opportunities of innovation and aggravate the perpetuation of societal stereotypes~\citep{garimella2018political}. In sum, the topology of the network can lead to a vicious cycle where those who are disadvantaged accumulate fewer opportunities to improve their social capital~\citep{fish2019gaps}. 

A variety of graph intervention methods have been proposed in the literature to mitigate disparities in social capital at an individual level~\citep{bashardoust2022reducing}. However, there is a lack of methods that consider such disparities at a group level, which is particularly relevant when the groups correspond to socially vulnerable groups, \textit{i.e.}, those defined on the grounds of sex, race, color, language, religion, political or other opinion, national or social origin, association with a national minority, property, birth or other~\citep{UnitedNations1966}.
Focusing on the group level also supports the development of inclusive solutions at scale that benefit entire communities, promoting equity, diversity and the inclusion of disadvantaged groups. We denote the disparity in social capital among different groups in the network as \emph{structural group unfairness}.

We consider a setting where each node in the network is a source of unique information and, therefore, access to all nodes is equally important. In this context, information flow is an integral component of the social capital and a distance metric that quantifies the total information flow in the graph, considering high-order relations that expand beyond the immediate neighbors, is of utmost importance. 
We propose using the effective resistance to measure the overall information flow between pairs of nodes \new{(and thus the social capital)}, since it is a theoretically grounded continuous graph diffusion metric that considers both \emph{local} and \emph{global} properties of the network's topology~\citep{chung1997spectral}. In Section~\ref{sec:groupmetrics}, we introduce three measures of group social capital ---\emph{group isolation}, \emph{group diameter} and \emph{group control}--- based on the effective resistance where the groups are defined according to the value of a protected attribute of interest. Using these measures of social capital, we define three measures of structural group unfairness in Section~\ref{sec:disparitymetrics}, and frame the challenge of mitigating structural group unfairness as a budgeted edge augmentation task in Section~\ref{sec:alg}. %
This section also presents the Effective Resistance Group Link (\ERG{}{}) algorithm, a greedy edge augmentation algorithm that iteratively adds edges to the graph to increase the social capital of the most disadvantaged group. %
In experiments on real-world networks, described in Section~\ref{sec:experiments}, we uncover significant levels of structural group unfairness when using gender 
as the protected attribute, with females being the most disadvantaged group in comparison to males. We also illustrate how our approach is able to not only mitigate \new{disparities in group social capital}, but also increase the social capital of all the groups in the network. 

\section{Related Work}

\paragraph{Social Capital}
Social capital is as a multidimensional construct that has been extensively studied in sociology, political science, economics, and more recently, computational social science~~\citep{putnam2015}. It is defined as the networks, relationships, and norms of trust and reciprocity within a community or society that facilitate cooperation and collective action~\citep{coleman1988}. In simple terms, social capital is the value derived from connections between people. It can be measured and analyzed both at an individual and collective levels~\citep{borgatti1998network} and it has been characterized according to different criteria. Some authors propose three main dimensions of social capital, namely: structural, emphasizing the relationships among individuals, organizations and communities; cognitive, focusing on the shared values, norms and beliefs that bind members of a group or community; and relational, highlighting the intensity and quality of relationships, including reciprocity, trust and obligations among individuals~\citep{nahapiet1998}. Others have proposed the distinction between bonding, bridging, and linking social capital~\citep{szreter2004}. Bonding social capital captures the aspects of ``inward looking'' communities that reinforce exclusive identities and homogeneous groups~\citep{coleman1988}; bridging social capital refers to ``outward looking'' networks across different groups that do not necessarily share similar identities~\citep{burt2004,granovetter1983strength}; and linking social capital characterizes the trusting relationships and norms of respect across power or authority gradients~\citep{woolcock2001place}. The three forms are important for the well-being of individuals and communities: bonding social capital contributes to social cohesion and support; bridging social capital to mutual understanding, solidarity and respect; and linking social capital to mobilize political resources and power.  

\paragraph{Computational Models of Social Capital}

Network analysis offers a robust computational framework to examine and quantify social capital~\citep{coleman1988}. We consider a setting where all the nodes in the network may be sources of relevant information. As a consequence, access to all nodes ---not just the sources or seeds of information--- is equally important. In this context, \emph{information flow} is an integral component of the social capital, and a variety of methods have been proposed to characterize it, mainly through two concepts: centrality and criticality~\citep{borgatti2005}. 

Centrality measures the relative importance or prominence of a node in the network, quantifying its ability to reach the rest of nodes. Different approaches have been proposed in the literature to measure centrality, including the degree centrality, closeness, local clustering, the assortativity coefficient~\citep{newman2003mixing}, Katz centrality~\citep{katz1953new} and PageRank~\citep{page1998pagerank}.
Criticality reflects the node's level of influence or vulnerability within the network~\citep{tizghadam2009graph}. Nodes with high criticality are essential, such that their failure or disruption can have significant consequences, cascading effects or system-wide impact. Measures of criticality include effective size~\citep{burt2004}, redundancy~\citep{borgatti1998network} and shortest path betweenness~\citep{jackson2019human}.

However, previously proposed methods are insufficient to accurately quantify the overall information flow in the network for several reasons. First, they model the distance between nodes as the shortest path distance (geodesic distance~\citep{newman2003mixing}) which overlooks alternative routes and indirect connections that may exist between distant nodes, thereby underestimating the potential pathways for information diffusion, influence propagation~\citep{stephenson1989rethinking} or resource exchange. This myopic view can lead to oversimplified representations of network dynamics, ignoring the interplay between weak ties, bridge nodes and overlapping communities that facilitate connectivity and communication across disparate components in the network. Second, most of the proposed approaches only consider first-order ---direct and local--- relationships between nodes, relying on small neighborhoods of the graph. As a result, they ignore global structural information~\citep{karhadkar2023fosr}, such as the properties of the network topology and long-range interactions between nodes, which can lead to inaccurate insights on how information flows globally~\citep{rampavsek2022recipe}. Third, popular approaches to model information flow in a network, such as the Independent Cascade model~\citep{kempe2003maximizing} %
assume homogeneous, deterministic and instantaneous interactions between neighboring nodes which might lead to inaccurate predictions, biased estimations and misrepresentations of actual diffusion patterns observed in complex networks~\citep{tauch2015measuring}. 

Conversely, graph diffusion metrics, such as the \emph{effective resistance}~\citep{stephenson1989rethinking,klein1993resistance}, offer a principled approach to quantifying distances and interactions between nodes within a network, addressing the above limitations. The effective resistance accurately captures not only short-range but also long-range relationships between nodes~\citep{chung1997spectral} because it considers alternative pathways, including the network dynamics, and quantifies connectivity between distant nodes. Therefore, it constitutes a natural information distance metric between nodes in a graph~\citep{bozzo2013resistance,stephenson1989rethinking}. 
Previous work has theoretically formulated measures of node centrality and criticality based on effective resistances~\citep{brandes2005centrality,tizghadam2010betres,bozzo2013resistance,newman2005measure}, yet we are not aware of any work that has modeled the social capital of a group of nodes by means of the effective resistance. From a practical perspective, the concept of effective resistance has been used to measure polarization in social networks~\citep{hohmann2023polarizationCT} and to rank user-items relations in recommender systems~\citep{fouss2007random}.
In this paper, we propose quantifying the social capital of a group of nodes in the network by means of three measures derived from the effective resistance: the group isolation, group diameter and group control, explained in Section~\ref{sec:groupmetrics}. 

\paragraph{Fairness in Graphs} 
Networks are used for a variety of purposes, including decision-making on nodes, link prediction, node embedding learning, clustering and community detection, ranking, and influence maximization. Fairness has been extensively analyzed in these scenarios~\citep[and references therein]{saxena2022fairsna,dong2023fairness}. 

Regarding social capital, social status plays a role in defining the structure of a network~\citep{ball2013} and a node's position in a network is a form of social capital~\citep{Burt2000}. %
Thus, there are structural advantages in information flow depending on the position that a node occupies in the network. Prior work has studied fairness from the perspective of disparities in access to information by differently positioned nodes in the graph, particularly in the case of influence maximization, \textit{i.e.}, when a single piece of information is spread in the network~\citep{wang2022information}.
However, there is a scarcity of studies that model group fairness considering that all the nodes in the network are sources of information and therefore access to all the nodes is equally important~\citep{bashardoust2022reducing}. 
In this context and to the best of our knowledge, no previous research has considered fairness in graphs from a group perspective, when the groups are defined according to protected attributes ---such as gender, ethnicity, religion or socio-economic status. In this paper, we fill this gap by defining, measuring and mitigating \emph{structural group unfairness}, understood as disparities in social capital between different groups and where social capital is measured by information flow.

 \paragraph{Network Interventions to Mitigate Unfairness in Graphs} Network interventions draw upon social network theory and structural analysis to understand and address the underlying mechanisms of unfairness within social networks. Interventions to mitigate structural unfairness in a network may entail redesigning network structures~\citep{santos2021link,hohmann2023polarizationCT} or altering (adding and/or removing) edges~\citep{tong2012gelling} to eliminate discriminatory barriers, reduce homophily, and foster diversity within the networks~\citep{granovetter1983strength}. These interventions aim to enhance connectivity, promote inclusivity, and facilitate equitable access to resources, opportunities, and support networks~\citep{borgatti2005}.

When aiming to improve the social capital in a network, edge augmentation (\textit{i.e.}, adding edges) is considered to be the natural intervention to mitigate disparities~\citep{bashardoust2022reducing}. Several edge augmentation strategies have been proposed in the literature, such as connecting similar nodes to improve bonding social capital~\citep{zhang2011differentiated}, linking nodes with the highest product of eigenvector centralities~\citep{tong2012gelling} or creating edges between the most disadvantaged nodes and the central node~\citep{bashardoust2022reducing}. However, these strategies are defined for individual notions of social capital and they do not consider long-range interactions between nodes. 

\paragraph{Contributions} Given previous work, the main contributions of our work are: \\
(1) We propose three effective resistance-based measures of group social capital in social networks ---namely \emph{group isolation}, \emph{group diameter} and \emph{group control}--- that consider short- and long-range interactions between nodes; \\
(2) We define \emph{structural group unfairness} as a disparity in the values of such measures by different groups in the graph, where the groups are defined according to the values of a protected attribute. This approach is particularly relevant from a social perspective when the disadvantaged group in the network corresponds to a vulnerable social group; \\ 
(3) We propose \ERG{}{}, an effective resistance-based greedy edge augmentation algorithm that iteratively adds edges to the network to maximize the social capital of the most disadvantaged \emph{group}. We approach this objective by adding weak ties~\citep{granovetter1983strength}
between the disadvantaged group and the rest of the graph; \\
(4) In experiments on real-world networks, we uncover significant levels of group structural unfairness when using gender as a protected attribute, with females being the most disadvantaged group in comparison to males. We also illustrate how our approach is the most effective in reducing structural group unfairness when compared to the baselines.

\section{Measuring Structural Group Unfairness}

\new{First, we define how to measure information flow between nodes in a social network, which forms the theoretical foundation for the proposed measures of group social capital. Next, we introduce three metrics to quantify a group's social capital within a graph and define structural group unfairness as the disparity in these metrics across different groups. Finally, we propose a greedy graph intervention (edge augmentation) algorithm designed to mitigate structural group unfairness, \emph{i.e.}, disparities in group social capital.}

\subsection{Preliminaries} 

\subsubsection{Effective Resistance and Social Capital} \label{subsec:effres}
We focus on the structural dimension of social capital, which emphasizes the relationships among individuals or communities~\citep{nahapiet1998}, and propose to measure it as the information flow of a node in the network.
Such a measure is captured by the \emph{effective resistance}~\citep{doyle1984rayleighs, klein1993resistance} %
of the node.
Given nodes $u$ and $v$ in graph $G=\{\mathcal{V,E}\}$, where $\mathcal{V}$ is the set of nodes, ${\mathcal{E}=\{(u,v) \in \mathcal{V} \times \mathcal{V} : A_{uv}=1\}}$ is the set of edges and $\mathbf{A}$ is the graph's adjacency matrix, the effective resistance $R_{uv}$ between nodes $u$ and $v$ is a distance metric given by:
\begin{align}\label{eq:effres}
&R_{uv}= (\mathbf{e}_u-\mathbf{e}_v)\mathbf{L}^\dagger(\mathbf{e}_u-\mathbf{e}_v)^\top,
\end{align} 

\noindent where $\mathbf{e}_u$ is the unit vector with a unit value at $u$-th index and zero elsewhere; $\mathbf{L}^\dagger= \sum_{i>1}  \frac{1}{\lambda_i}\phi_i \phi_i^\top$ is the pseudo-inverse of the graph's Laplacian $\mathbf{L} = \mathbf{D}-\mathbf{A}=\mathbf{\Phi \Lambda \Phi}^\top$, with $\mathbf{D}$ the graph's degree matrix, $D_{u,u} = \sum_{j\in V} A_{u,v}$ and 0 elsewhere; and $\lambda_i$ the $i$-th smallest eigenvalue of $\mathbf{L}$ corresponding to the $\phi_i$ eigenvector. The complete matrix of all pairwise effective resistances in a graph, $\mathbf{R}$ is given by $\mathbf{R}=\mathbf{1} \diag(\mathbf{L}^\dagger)^\top + \diag(\mathbf{L}^\dagger)\mathbf{1}^\top- 2\mathbf{L}^\dagger$.

The effective resistance is a distance \emph{metric} since it satisfies the symmetry, non-negativity and triangle inequality conditions~\citep{ellens2011effective}. In addition, $R_{uv}$ is proportional to the commute times between $u$ and $v$, \textit{i.e.}, the expected number of steps in a random walk starting at $v$ to reach node $u$ and come back: $R_{uv} \propto \mathbb{E}_u[v] + \mathbb{E}_v[u]$, where $\mathbb{E}_u[v]$, $\mathbb{E}_v[u]$ are the expected number of steps that a random walker takes to go from $u$ to $v$ and from $v$ to $u$, respectively~\citep{chandra1989electrical, tetali1991effres}. A high value of $R_{uv}$ means that $u$ and $v$ generally struggle to visit each other in a random walk, \textit{i.e.}, nodes with high effective resistance between them are unlikely to exchange information. $R_{uv}$ can be expressed as 
\begin{equation*}
    R_{uv} = \sum_{i=0}^{\infty}\left(\frac{1}{d_u}(\mathbf{A}^i)_{uu} + \frac{1}{d_v}(\mathbf{A}^i)_{vv} - \frac{1}{\sqrt{d_u d_v}} 2(\mathbf{A}^i)_{uu}\right),
\end{equation*}
being $\mathbf{A}^k$ the matrix that defines the number paths of length $k$ between $u$ and $v$~\citep{black2023understanding}. Hence, it is able to capture both short- and long-range interactions between nodes in the graph.

The effective resistance has been characterized as the \emph{information distance} in a network~\citep{stephenson1989rethinking,bozzo2013resistance} as it quantifies the amount of effort (distance) required to transmit information between the nodes. The total effective resistance $\Rg$ of a graph~\citep{ellens2011effective} -- defined as the sum of all $R_{uv}$ ($\Rg=\mathbf{1R}\mathbf{1}^\top$) -- is therefore inversely proportional to the expected ease of information flow in the graph. 

The \textit{total effective resistance of node} $u$, $\Rg(u)$, is given by $\Rg(u) =\sum_{v\in \mathcal{V}} R_{uv}$, \textit{i.e.}, the sum of all the effective resistances between node $u$ and the rest of nodes in the network. The smaller the total effective resistance of a node, the larger its information flow. In other words, the effective resistance allows to identify which nodes in a graph have limited information flow (\textit{i.e.}, high effective resistance) and thus low social capital. Equivalent terms to denote the effective resistance in the literature include the current-flow closeness centrality~\citep{brandes2005centrality} and the information centrality~\citep{stephenson1989rethinking}. %

From a computational perspective, calculating $R_{uv}$ does not require hyper-parameter tuning and can be efficiently calculated, mitigating two significant drawbacks of other diffusion or learnable graph distances~\citep{perozzi2014deepwalk}. An overview of the theoretical properties of $R_{uv}$ are provided in Appendix~\ref{app:extrabackgroundCT}. 

\subsection{Effective Resistance-based Measures of Group Social Capital}
\label{sec:groupmetrics}
Based on the definition of effective resistance above, we propose three metrics that characterize the social capital of a group of nodes in a graph. %
In the following, a graph $G=\{\mathcal{V,E}\}$ is composed of a set of nodes $\mathcal{V}$ and edges $\mathcal{E}$; and $S_i$ is a group of nodes defined as a subset of $\mathcal{V}$, \textit{i.e.}, $S_i \subseteq \mathcal{V}$ with $|S_i|$ nodes.

\paragraph{\textbf{1. Group Isolation}} %
The isolation of a group $S_i$, $\Rg(S_i)$, is given by the average of the total effective resistances of all the nodes in group. $\Rg(S_i)$ is proportional to the expected information distance when sampling one node from group $S_i$ and another node at random. It can be interpreted as a proxy for the marginalization of a group from the perspective of information flow, such that the lower the $\Rg(S_i)$, the \new{more the information flows and thus the} less isolated the group $S_i$ is in the network. Therefore, reducing this measure for group $S_i$ would yield an increase in its social capital.
It is given by:
\begin{equation}\label{eq:totaleffrgroup}
\Rg(S_i) = \mathbb{E}_{u\sim S_i}[R_{tot}(u)]  = |\mathcal{V}|\mathbb{E}_{u\sim S_i,v\sim \mathcal{V}}\left[R_{uv}\right]
\end{equation}

\noindent where $\Rg(u)=\sum_{v\in \mathcal{V}} R_{uv}$ is the total effective resistance of node $u$.
Computing the expectation enables comparing groups of different sizes.  

Note that adding links between nodes with the highest $R_{uv}$ ---irrespective of which group they belong to--- has been found to reduce the total effective resistance of a graph~\citep{black2023understanding} and hence the isolation of all the graph's nodes. 

\paragraph{\textbf{2. Group Diameter}} %
The group diameter, $\Rd(S_i)$, measures the average of the maximum distance between any node in group $S_i$ and any node in the graph. A larger group diameter suggests that the nodes in group $S_i$ are distant from the rest of the graph, indicating potential challenges in information exchange with the nodes outside of $S_i$, and hence it can be interpreted as another measure of social capital. Low diameter is important to ensure diverse information dissemination (\emph{e.g.}, job announcements reaching a diverse pool of candidates) which depends on the entire network and not just on the average distance as \textit{group isolation}. This measure is based on $\Rd(G)$, which is the maximum effective resistance of the graph~\citep{chandra1989electrical} and closely related to the cover time. We define \textit{Group Diameter} as:
\begin{equation}\label{eq:diamreffgroup}
\Rd(S_i) = \mathbb{E}_{u \sim S_i} [\Rd(u)] =  \mathbb{E}_{u \sim S_i}[\max_{v \in \mathcal{V}}  R_{uv}]
\end{equation}
where $\Rd(u) = \max_{v\in \mathcal{V}} R_{uv}$ is the diameter of node $u$, \textit{i.e.}, the maximum $R_{uv}$ from $u$ to any other node in the graph.
$\Rd(S_i)$ gives an indication of the information flow gap between the group $S_i$ and the rest of the network~\citep{fish2019gaps}. Therefore, the larger the $\Rd(S_i)$, the lower the social capital of group $S_i$. 

\paragraph{\textbf{3. Group Control}} 
The aforementioned concepts measure the amount of information flow through a group of nodes in the graph. Another relevant variable to assess is the \textit{criticality} of a node for the diffusion of information in the graph, which in the literature has been measured as betweenness~\citep{freeman1977betweenness}, redundancy~\citep{borgatti1998network} or effective size~\citep{burt2004}.
Nodes with high levels of control serve as important connectors in the network, facilitating the flow of information and enabling communication between otherwise disconnected groups of nodes~\citep{burt2004}.

The control of a node can be computed by restricting the summation of a node's total effective resistance to its direct neighbors. Thus, it is expressed as $\Bt(u) = \sum_{v \in \mathcal{N}(u)} R_{uv}$,
where $\mathcal{N}(u) = \{v: (u,v)\in \mathcal{E}\}$ are the neighbors of $u$. The larger the $\Bt(u)$, the more control a node has in the network's information flow and hence the larger its social capital. The node control of a node is bounded by $1\leq\Bt(u)\leq d_u$, being $d_u$ the number of neighbors of node $u$ (see Theorem~\ref{app:thm:nodecontrolbounds}). $\Bt(u)$ is theoretically related to the current-flow betweenness~\citep{newman2005measure}, the node's information bottleneck~\citep{arnaiz2022diffwire}, and the curvature of the node~\citep{devriendt2022discrete}.

We define the group control or group betweenness $\Bt(S_i)$ as the average of the controls of all the nodes in $S_i$, \emph{i.e.}:
\begin{equation}\label{eq:groupcontrol}
    \Bt(S_i) = \mathbb{E}_{u \sim S_i} [\Bt(u)]= \frac{1}{|S_i|}\sum_{u \in S_i} \Bt(u),
\end{equation}
and is bounded by $1\leq \Bt(S_i) \leq \text{vol}(S_i)/|S_i|$, being $\text{vol}(S_i)/|S_i|$ the average degree of the group $S_i$ (see Theorem~\ref{app:thm:groupcontrolbounds}).
Note that the sum of all $R_{uv}$ for all nodes in a graph is constant at $|\mathcal{V}|-1$ and it is independent of the number of edges~\citep{klein1993resistance}. If an edge is added, removed or modified in the graph, all $R_{uv}$ are updated accordingly such that their sum remains constant.
The sum and the average control of all nodes in the graph are also constant with values $\sum_{u\in\mathcal{V}} \Bt(u) = 2|V|-2$ and $\mathbb{E}_{u\sim \mathcal{V}}[\Bt(u)] = 2-\frac{2}{|\mathcal{V}|}$, respectively, independently of the number of edges (see Appendix~\ref {app:sec:groupcontrol} for more details). Consequently, the control of a node or group of nodes is distributed among the nodes in the network and cannot be optimized for every node/group in the graph by adding more edges: if a node or group of nodes increase their control over the information flow in the graph, they must do so at the cost of reducing the control of other nodes. %
\new{Figure \ref{fig:metrics-example} illustrates the three proposed measures of individual and group social capital.}

\begin{figure}[h]
    \centering
    \includegraphics[width=\linewidth]{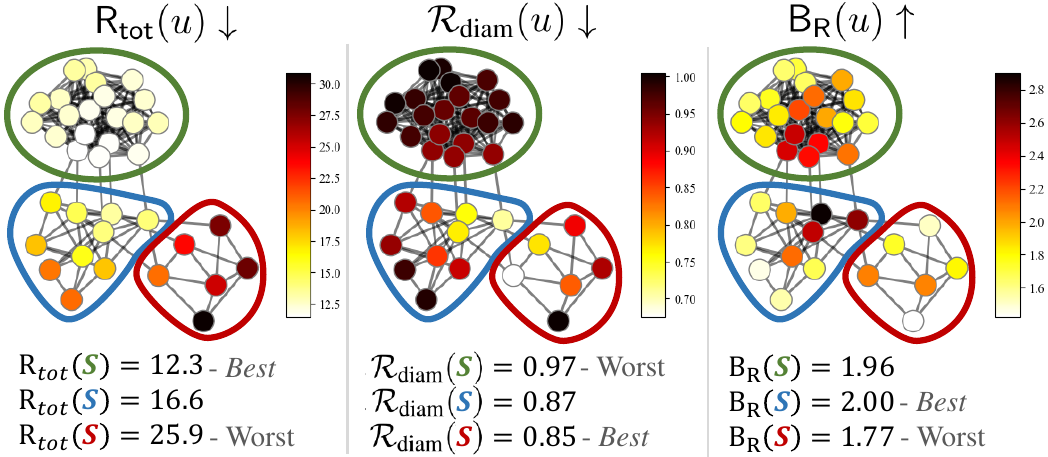}
    \caption{Illustration of the three proposed group social capital metrics on the same graph. The color of the nodes corresponds to $\Rg(u)$, $\Rd(u)$ and $\Bt(u)$, respectively. The nodes are grouped according to three different values of the protected attribute $S$ indicated as green, blue and red.}
    \label{fig:metrics-example}
\end{figure}

\subsection{Structural Group Unfairness}
\label{sec:disparitymetrics}
 
To study disparities in the distribution of social capital in the network, we define groups of nodes in the network $S_i$ according to the values of a sensitive attribute $i \in SA=\{sa_1, sa_2, \dots, sa_{|SA|}\}$, which is a categorical variable with $|SA|$ possible values referring to a socially relevant concept, such as sex, age, gender, religion or race. We denote the value of the sensitive attribute of a node $v$ as $SA(v)$. For instance, if $SA$ is sex with three possible values, $SA$=\{male, female, non-binary\}, the groups $S_{\text{male}}$, $S_{\text{female}}$ and $S_{\text{non-binary}}$ are the set of nodes whose sex is labeled as male, female and non-binary, respectively.

We define the \emph{structural group unfairness} in a network as the disparity in information flow between the nodes belonging to different groups in the network. Since we have defined the groups in terms of protected attributes, the structural group unfairness is socially relevant as it informs about potential disparities in information flow (and hence social capital) between a vulnerable group and the rest of the network. We present here three metrics to characterize the structural group unfairness, namely \textit{isolation disparity}, \textit{diameter disparity} and \textit{control disparity}. 

\paragraph{\textbf{1. Group Isolation Disparity}} 
Ideally, every group in the network should have the same levels of information flow and hence the same ---and low--- levels of group isolation, namely: 
\begin{equation}\label{eq:isodisp}
    \Rg(S_i) = \Rg(S_j), \:\forall\: i,j \in SA.
\end{equation}
Deviations from equality lead to isolation disparity $\Delta\Rg$, which is defined as 
the maximum over all groups in the graph of the differences in group isolation: $\Delta\Rg= \max_{i,j \in SA} |\Rg(S_i)-\Rg(S_j)|$.

Reducing the isolation disparity contributes to increasing the social capital of the most disadvantaged group and equalizes the information flow between the groups in the network.

\paragraph{\textbf{2. Group Diameter Disparity}} 
Ideally, every social group in the network should have the same --- and low --- group diameter: 
\begin{equation}
    \Rd(S_i) = \Rd(S_j), \:\forall\: i,j \in SA.
\end{equation}
Any deviations from equality lead to diameter disparity, $\Delta\Rd$, defined as the maximum over all groups in the graph of the differences in group diameter: $\Delta\Rd = \max_{i,j \in SA} |\Rd(S_i)-\Rd(S_j)|$. 

Achieving equal diameter entails equalizing the worst-case scenario in information flowing to the entire network from the perspective of any group of nodes in the graph. By promoting equal group diameter, we generate a fairer information-sharing environment.

\paragraph{\textbf{3. Group Control Disparity}} 
By striving for equalized control in all groups in the network, no particular group would dominate or be marginalized from the perspective of their control of information flow in the network: 
\begin{equation}
    \Bt(S_i) = \Bt(S_j) = 2-\frac{2}{|\mathcal{V}|}, \:\forall\: i,j \in SA.
\end{equation}
Then, control disparity, $\Delta\Bt$, is defined as the maximum over all groups in the graph of the differences in group control: $\Delta\Bt= \max_{i,j \in SA} |\Bt(S_i)-\Bt(S_j)|$.

Control is a bounded resource to be distributed among the groups of nodes in the graph with an expected value of $2-\frac{2}{|\mathcal{V}|}$. Hence, reducing the control disparity entails a redistribution of the control in all the groups in the graph converging to $\Bt(S_i) = 2-\frac{2}{|\mathcal{V}|}, \:\forall\: i \in SA$, and hence leading to a more equitable allocation of the control that different groups play regarding the information flow in the network. 

\subsection{Structural Group Unfairness Mitigation}
\label{sec:alg} 

\paragraph{Edge Augmentation} Edge augmentation has been proposed as the natural intervention to mitigate information flow disparities in a network where all the nodes are sources of unique pieces of information~\citep{bashardoust2022reducing}.
According to Rayleigh's monotonicity principle \citep{doyle1984rayleighs}, %
adding edges to a graph always improves information flow, which is the aim of our intervention. Socially, edge addition enhances the nodes' social capital since information reaches a larger audience. Moreover, note that edge deletion is not a suitable intervention as it could lead to a disruption of social dynamics by breaking existing connections between individuals, which is undesirable \citep{jackson2019human}.

\begin{figure}[ht]
    \centering
\includegraphics[width=\linewidth]{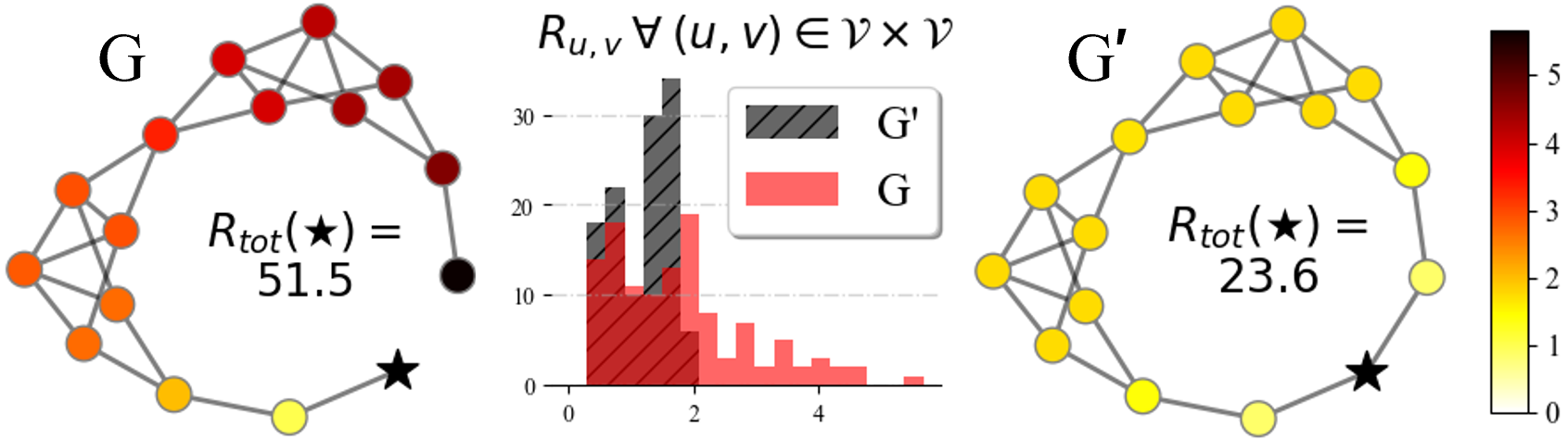}
    \caption{Illustration of the impact of adding one edge on the information flow of $G$. 
    Note how all the effective resistances between the star node and the rest of nodes in the network ($R_{\bigstar,v}$) decrease in $G'$ even if there is no change in the geodesic distance between them.}
    \label{fig:ruvexample}
\end{figure}

Regarding which structural group unfairness measure to optimize, we argue that we should primarily focus on improving the isolation disparity ($\Delta R_{tot}$) of the most isolated group in the graph. 
Note that mitigating isolation will also yield an improvement in the diameter and control disparities, as illustrated in our experiments. 
The reduction of $\Delta R_{tot}$ entails creating edges between distant nodes, \textit{i.e.}, fostering the creation of weak ties. Granoveter's work~\citep{granovetter1983strength} provides evidence that information spreads more effectively through weak ties than through strong ties because weak ties give peripheral nodes more visibility in the network, which leads to a decrease in group isolation and diameter. Adding weak ties reduces discontinuities in the information flow, increases redundancies in the paths between nodes and improves the control of peripheral nodes while reducing the control of dominant ones~\citep{burt2004}.

Previous work has proposed connecting peripheral isolated nodes (with high isolation and low control) to salient nodes (with low isolation and high control)~\citep{tong2012gelling}. However, these solutions lead to a \textit{rich-get-richer} phenomenon that benefits the best connected nodes and potentially increases disparities in information access and control~\citep{burt1999leaders}. Therefore, we advocate creating edges between the most distant nodes in the network without necessarily connecting them to a central node. Note that adding edges between the nodes with maximum $R_{uv}$ theoretically leads to a minimization of $\Rg$ and $\Rd$ for all nodes in the graph while balancing $\Bt$. Hence, the choice of $R_{tot}(S)$. %

\paragraph{Problem Definition} 
We consider a budgeted edge augmentation intervention: given a maximum number $B$ of allowed new connections to be created in the graph, we %
aim to identify the $B$ new edges $\mathcal{E}'$ to be added to the graph $G$ 
that would maximally reduce the group isolation disparity of the most disadvantaged group in the graph. This leads to a new graph $G'$ with lower levels of structural group unfairness: 

\begin{align}
    G' = \argmin_{G'=\mathcal{(V,E')}}& \quad \mathbb{E}_{u,v\thicksim V\times V}\left[R_{uv}\right]\\
    \textrm{s.t.}&\quad |\mathcal{E}'\setminus \mathcal{E}| = B \text{ and } 
     \mathcal{E} \subset \mathcal{E}' \nonumber
\end{align}

\paragraph{Algorithm} To tackle the problem above, we introduce \texttt{ERG-Link}, a greedy algorithm that adds edges between the nodes with the largest effective resistance between them, where at least one of the nodes belongs to the most isolated group as per Sec.~\ref{sec:groupmetrics}, and groups in the graph are defined on the grounds of a protected attribute. Note that this strategy also reduces the isolation (total effective resistance) of the entire graph~\citep{ghosh2008minimizing}.

Algorithm~\ref{alg:erp} outlines the main steps of \ERG{}{}. Given a graph $G=(V,E)$, a protected attribute $S$ and a total budget $B$ of new edges to add, the group isolation $R_{tot}(S)$ is computed for each group according to $S$. The most disadvantaged group $S_d$ is identified as the group with the largest $R_{tot}(S)$. Then, $R_{uv}$ $\forall (u,v) \in \mathcal{V} \times \mathcal{V}$ is computed, and a ranking of all potential new edges in the graph is created from the highest to the lowest values of effective resistance. In each iteration, \ERG{}{} adds the new edge to the graph that yields the largest improvement in the information flow of $S_d$, \textit{i.e.}, the edge that connects the two nodes with largest effective resistance between them where at least one of the nodes belongs to the disadvantaged group $S_d$. See~\citet{ghosh2008minimizing,black2023understanding} for a proof that such an edge is the one that maximally improves the information flow in the graph. 

\SetKwRepeat{Repeat}{Repeat}{Until}{}
\begin{algorithm}
  \caption{\ERG}\label{alg:erp}
  \KwData{Graph $G=(\mathcal{V},\mathcal{E})$, %
  a protected attribute $SA$, 
  budget $B$ of total number of edges to add}
  \KwResult{New Graph $G'=(\mathcal{V}', \mathcal{E}')$ with $B$ new edges}

  \BlankLine
  $\mathbf{L} = \mathbf{D}-\mathbf{A}$\;
  $S_d = \operatorname*{argmax}_{S_i \forall i \in SA} \Rg(S_i)$ %
  
  \Repeat{$|\mathcal{E'}\setminus\mathcal{E}|=B$}{
  $\mathbf{L}^\dagger = \sum_{i>0} \frac{1}{\lambda_i} \phi_i \phi_i^\top = \left(\mathbf{L}+\frac{\mathbf{1}\mathbf{1}^\top}{n}\right)^{-1} - \frac{\mathbf{1}\mathbf{1}^\top}{n}$%
  
  $\mathbf{R}=\mathbf{1} \diag(\mathbf{L}^\dagger)^\top + \diag(\mathbf{L}^\dagger)\mathbf{1}^\top- 2\mathbf{L}^\dagger$ %
  
  $C=\{(u,v) \mid u \in S_d \text{ or } v \in S_d, (u,v) \notin \mathcal{E}'\}$ %
  
  $\mathcal{E}' = \mathcal{E}'\cup \arg \max_{(u,v)\in C} R_{uv}$ %
  
  $\mathbf{L} = \mathbf{L} + (\mathbf{e}_u-\mathbf{e}_v)(\mathbf{e}_u-\mathbf{e}_v)^\top$%
}
  \Return $G'$\;
\end{algorithm}

\ERG{} leverages \emph{Rayleigh's monotonicity principle}~\citep{doyle1984rayleighs,ellens2011effective}, according to which the total effective resistance of a graph can only decrease when new edges are added to it, as illustrated in Figure~\ref{fig:ruvexample}. Therefore, creating an edge between nodes with maximum $R_{uv}$ where one of the nodes belongs to the disadvantaged group not only improves the information flow between the two nodes (increasing their social capital) but it also improves the information flow of the entire graph.

Note that the addition of each new edge changes all the pairwise information distances between nodes in the graph, requiring the re-computation of all distances (effective resistances) in each iteration. Therefore, it is not feasible to perform this type of edge augmentation by means of Independent Cascade distance estimation~\citep{bashardoust2022reducing, kempe2003maximizing}, random-walk embeddings~\citep{perozzi2014deepwalk} or GNNs~\citep{wu2022graph} since these methods require training expensive neural networks or running complex simulations for the estimation of the distances in each iteration. In addition, they only capture short-range interactions between nodes. 
Conversely, the effective resistance captures both short and long-range interactions between nodes in the graph and it is efficient to update. While it requires the computation of the Laplacian pseudo-inverse, Woodbury's formula~\citep{black2023understanding} can be used to avoid recomputing $\mathbf{L}^\dagger$ in line 3 of Alg.~\ref{alg:erp}, as reflected in Alg.~\ref{alg:erp-efficient}.

\section{Experiments}\label{sec:experiments}

\subsection{Datasets and Set-up} 

To empirically evaluate \ERG{}, we tackle the challenge of mitigating \new{group social capital disparities} in three real-world networks (school and online social networks), where the nodes are users and the edges correspond to connections between them, \textit{i.e.}, friendships. The three datasets are commonly used in the graph fairness literature, namely: 
\\
(1) The Facebook dataset~\citep{leskovec2012learning}, a dense graph of 1,034 Facebook users ($|\mathcal{V}|$) and 26,749 edges ($|\mathcal{E}|$). It corresponds to a large ego-network where nodes are connected if they are friends in the social network; \\
(2) The UNC28 dataset~\citep{red2011comparing}, consisting of a 2005 snapshot from the Facebook network of the university of North Carolina ($|\mathcal{V}|$=3985,$|\mathcal{E}|$=65287); \\
(3) The Google+ dataset~\citep{leskovec2012learning}, an ego-network of G+, the social network developed by Google, with 3,508 nodes ($|\mathcal{V}|$) and 253,930 edges ($|\mathcal{E}|$). 

Gender is the protected attribute in all networks with two possible values $SA=$\{male, female\}. We select the largest connected component for all the datasets. The original values of group social capital per gender are depicted in Table~\ref{tab:socialcapital}. As seen in the Table, our study unveils that the disadvantaged group according to the three defined measures  corresponds to females in the three datasets. %
\new{We provide a detailed breakdown of the statistics and characteristics of the datasets in Appendix~\ref{app:data-stats}.}

\begin{table}[ht]
\centering
\begin{small}
\begin{tabular}{lrrr}%
\toprule
$G$ & $R_{tot}\downarrow$ & $\mathcal{R}_{diam}\downarrow$ & $\mathsf{B_R}\uparrow$ \\
\midrule 
Facebook (female) & 221.4 & 2.29 & 1.93\\
Facebook (male) & \textbf{179.8} & \textbf{2.25} & \textbf{2.03} \\
\midrule 
UNC28 (female) & 608.6 & 2.11 & 1.99 \\
UNC28 (male)  & \textbf{586.3} & 2.11 & \textbf{2.00} \\
\midrule 
Google+ (female) & 564.1 & 1.31 & 1.81 \\
Google+ (male) & \textbf{287.7} & \textbf{1.24} & \textbf{2.32} \\
\bottomrule
\end{tabular}
\end{small}
\caption{\new{Group social capital metrics in the original graphs for each of the groups.}} %
\label{tab:socialcapital}\end{table}

\begin{table*}[th]
\centering
\begin{subtable}[ht]{0.32\textwidth}
\caption{Facebook ($B$=50)}
\label{tab:fb50fair}
    \centering
    \resizebox{\linewidth}{!}{
    \begin{tabular}{lrrr}
            \toprule
             & $\Delta\Rg$ & $\Delta\Rd$ & $\Delta \Bt$ \\
            $G$ (original) & 41.62  & 0.042 & 0.107 \\
            \midrule
            Random  & 38.7  & 0.039 & 0.108 \\
            SDRF & 41.6 &  0.042 & 0.106 \\
            FOSR & 34.5 & 0.027 & 0.109 \\
            DW    & 36.3 & 0.031 & 0.104  \\
            Cos   & 28.7 & 0.029 & 0.120  \\
            ERG   & \textbf{10.3} & \textbf{0.009} & \textbf{0.098}  \\
            \midrule 
            S-ERG   & 41.6 & 0.042 & 0.107  \\
            \bottomrule
    \end{tabular}
    }
\end{subtable}
\begin{subtable}[ht]{0.32\textwidth}
\centering
    \caption{UNC28 ($B$=5000)}
    \label{tab:unc5000fair} 
    \resizebox{\linewidth}{!}{
        \begin{tabular}{lrrr}
            \toprule
             & $\Delta\Rg$ & $\Delta\Rd$ & $\Delta \Bt$ \\
            $G$ (original) & 22.4 & 0.006 & 0.009 \\
            \midrule
            Random & 19.8 & 0.005 & 0.014 \\
            SDRF & 22.2 & 0.006 & 0.007 \\     
            FOSR & 19.7 & 0.005 & 0.017 \\
            DW & 22.2 & 0.006 & 0.004 \\
            Cos & 19.1 & 0.005 & 0.102 \\
            ERG & \textbf{8.8} & \textbf{0.002} & \textbf{0.003} \\
            \midrule 
            S-ERG & 22.3 & 0.006 & 0.004 \\
            \bottomrule
            \end{tabular}
        }
\end{subtable}
\begin{subtable}[ht]{0.312\textwidth}
    \centering
    \caption{Google+ ($B$=5000)}
    \label{tab:google5000fair} 
    \resizebox{\linewidth}{!}{
        \begin{tabular}{lrrr}
        \toprule
         & $\Delta\Rg$ & $\Delta\Rd$ & $\Delta \Bt$ \\
        $G$ (original) & 276.4 & 0.078 & 0.51\\
        \midrule
        Random & 129.4  & 0.037 & 0.47  \\
        SDRF & 276.1 & 0.079 & 0.52 \\
        FOSR & 240.7 & 0.068 & 0.50 \\
        DW & 274.1  & 0.078 & 0.51  \\
        Cos & 86.8  & 0.025 & 0.47  \\
        ERG & \textbf{37.1} & \textbf{0.011} & \textbf{0.29}  \\
        \midrule 
        S-ERG & 276.4 & 0.079 & 0.52  \\
        \bottomrule
        \end{tabular}
            }
\end{subtable}
\caption{Structural group unfairness, \new{i.e. differences in group social capital between males and females}, before and after the graph interventions. Best result highlighted in bold.}
\label{tab:expfairmetrics}
\end{table*}

\subsection{Baselines}
We compare edge augmentation by means of \ERG{}{} with five baselines:
\\
(1) \textit{Random}, which adds edges at random to the graph. \\
(2) \textit{DW}, which adds edges with the lowest dot product similarity of \emph{DeepWalk} embeddings~\citep{perozzi2014deepwalk}. 
\\
(3) \textit{Cos}, a greedy algorithm that adds edges with the lowest cosine similarity of the rows of the adjacency matrix~\citep{saxena2022fairsna}. This is an example of a classic method based on neighborhood similarity.\\
(4) \textit{SDRF}, an edge augmentation method originally proposed to mitigate over-squashing \citep{topping2022understanding}. It identifies the edge with the minimum Ricci curvature and adds the edges that maximally improve the Ricci curvature of that edge. The Ricci curvature of an edge is related to $R_{uv}$ and $\Bt(u)$ as per Eq.~\ref{app:eq:linkcurvature}.\\
(5) \textit{FOSR}, a graph rewiring algorithm to increase the spectral gap ($\lambda_2$) and hence avoid over-squashing \citep{karhadkar2023fosr}. 

Note that \textit{DW} and \textit{Cos} correspond to an algorithm similar to Alg.~\ref{alg:erp} (lines 2 and 6 remain the same) with one difference: instead of using the effective resistances to quantify the distances between nodes, they consider DW or cosine distances, respectively.
As previously noted, we do not include any Independent Cascade distance estimation, random-walk embedding or GNN-based method as baselines because they require training a neural network or a simulation to estimate the pairwise distances in the graph, which is computationally unfeasible in our task as it would entail retraining the neural network every time a new edge is added~\citep{wu2022graph}. 

\subsection{Experimental Methodology}
We set a budget $B$ of 5,000 new edges to be added to the UNC28 and Google+ datasets, which corresponds to approximately 0.05\% of the number of all potential edges in the graph. We also run experiments with a maximum of 50 new edges for the Facebook dataset to evaluate the performance of the algorithms with extremely low budgets.
We compute the social capital for each group (male and female) and the structural group unfairness on both the original and the augmented graphs (after all edges have been added) based on the defined measures. Moreover, we compute them at each step of edge addition to shed light on the evolution of the structural group unfairness as new edges are added.

\begin{figure*}[th]
\begin{subfigure}{0.21\textwidth}
    \centering
\includegraphics[width=\textwidth]{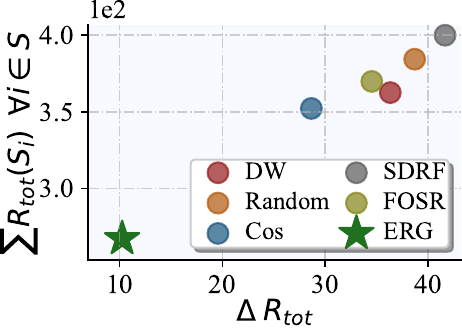}
\end{subfigure}
\hfill
\begin{subfigure}{0.21\textwidth}
    \centering
\includegraphics[width=\textwidth]{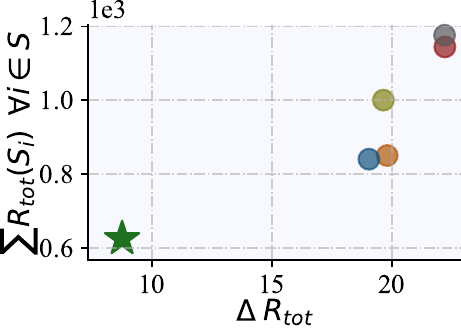}
\end{subfigure}
\hfill
\begin{subfigure}{0.21\textwidth}
    \centering
\includegraphics[width=\textwidth]{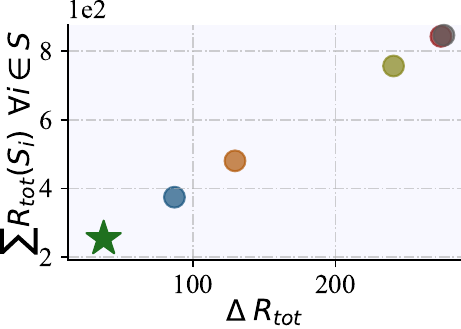}
\end{subfigure}
\\
\begin{subfigure}{0.21\textwidth}
    \centering
    \includegraphics[width=\textwidth]{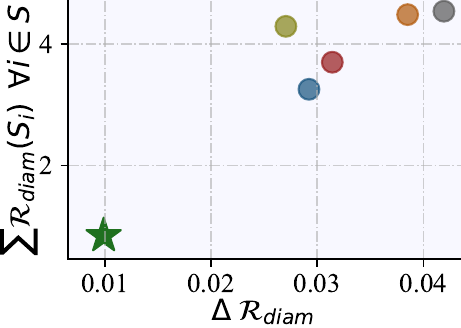}
\caption{Facebook}
\end{subfigure}
\hfill
\begin{subfigure}{0.21\textwidth}
    \centering
\includegraphics[width=\textwidth]{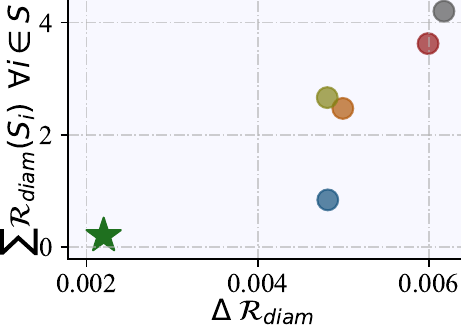}
\caption{UNC}
\end{subfigure}
\hfill
\begin{subfigure}{0.2\textwidth}
    \centering    \includegraphics[width=\textwidth]{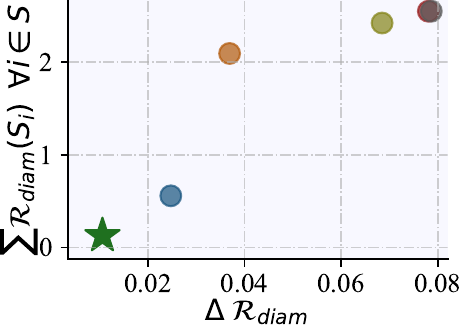}
    \caption{Google}
\end{subfigure}
\caption{Pareto front of the structural group unfairness (X-axis) vs the sum of the group's isolation of all the groups (Y-axis) using $\Rg$ and $\Rd$ (denoted by X-axis' label). Best results correspond to the bottom left corner of the graph.}
\label{fig:paretomain}
\end{figure*}

\begin{figure*}[th]
    \begin{subfigure}{0.27\textwidth}
    \centering
    \includegraphics[width=\linewidth]{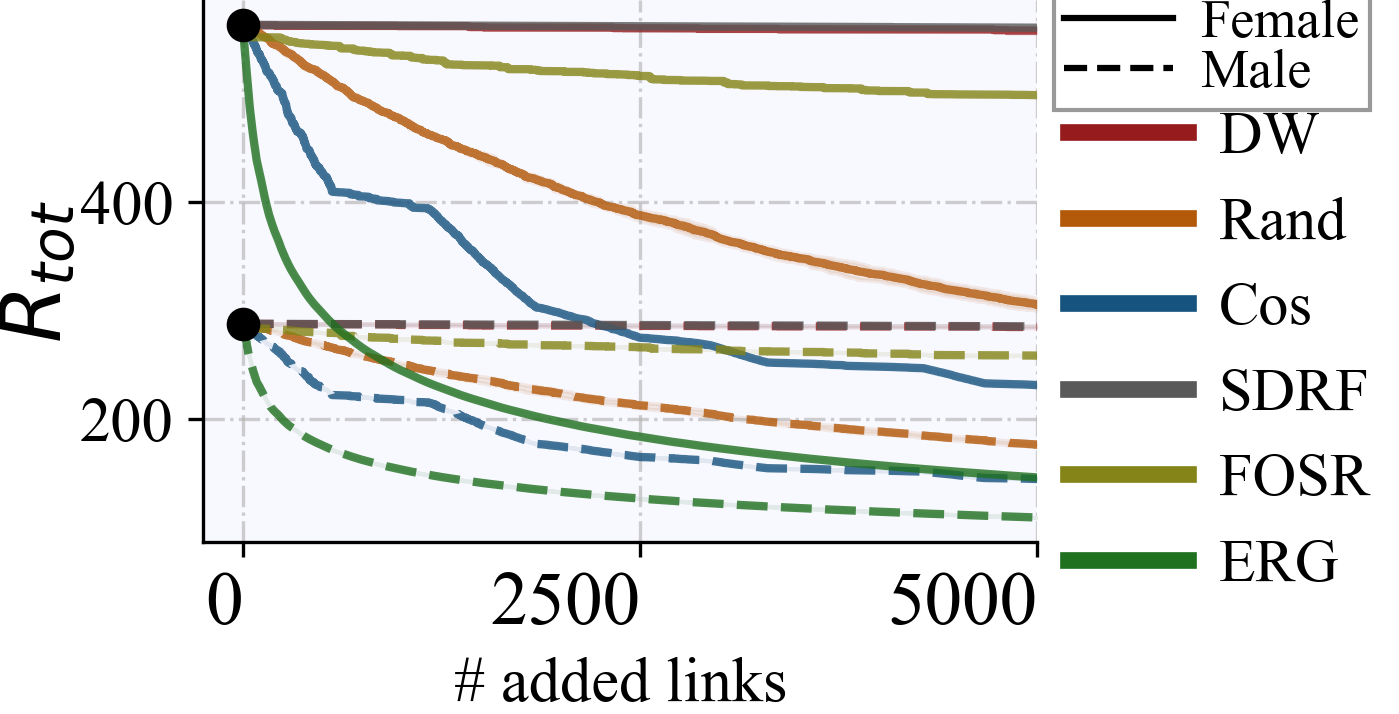}
    \end{subfigure}
    \hspace{.85in}
    \begin{subfigure}{0.2\textwidth}
    \centering
    \includegraphics[width=\linewidth]{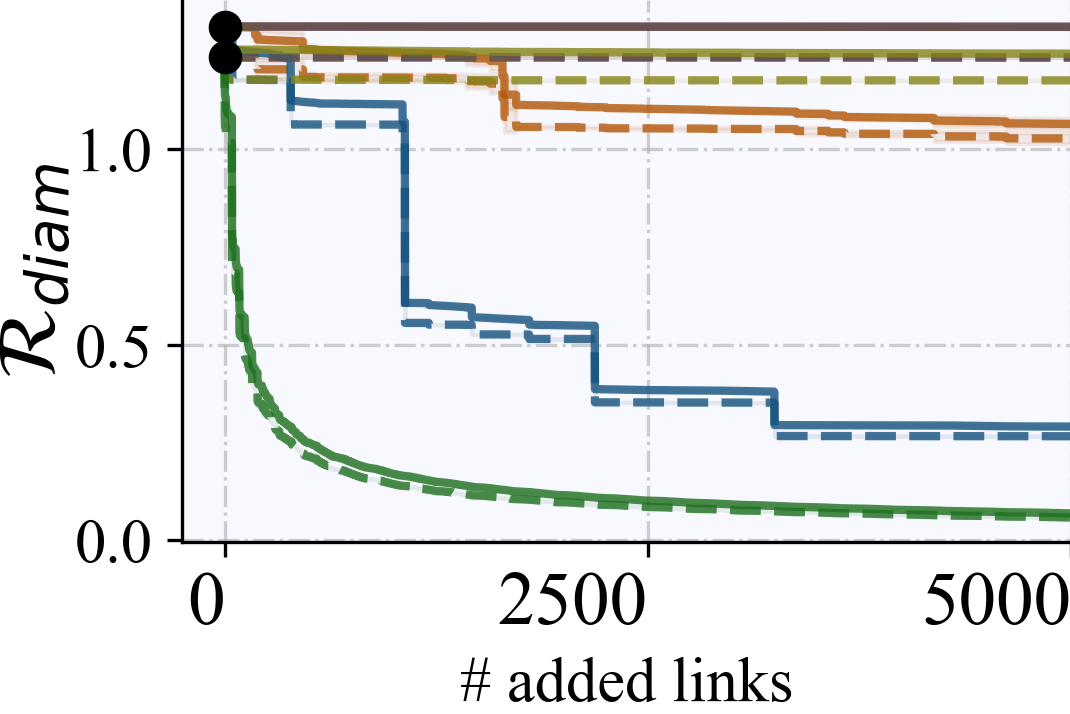}
    \end{subfigure}
    \hfill
    \begin{subfigure}{0.2\textwidth}
    \includegraphics[width=\linewidth]{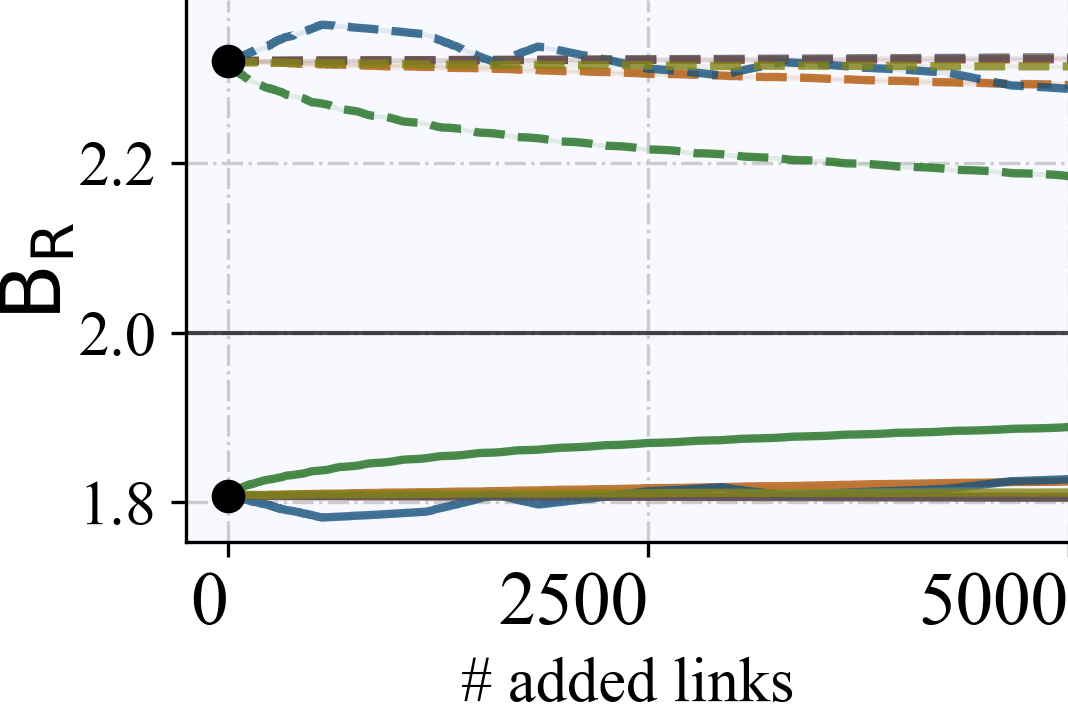}
    \end{subfigure}
    \\
    \begin{subfigure}{0.2\textwidth}
    \centering
    \includegraphics[width=\linewidth]{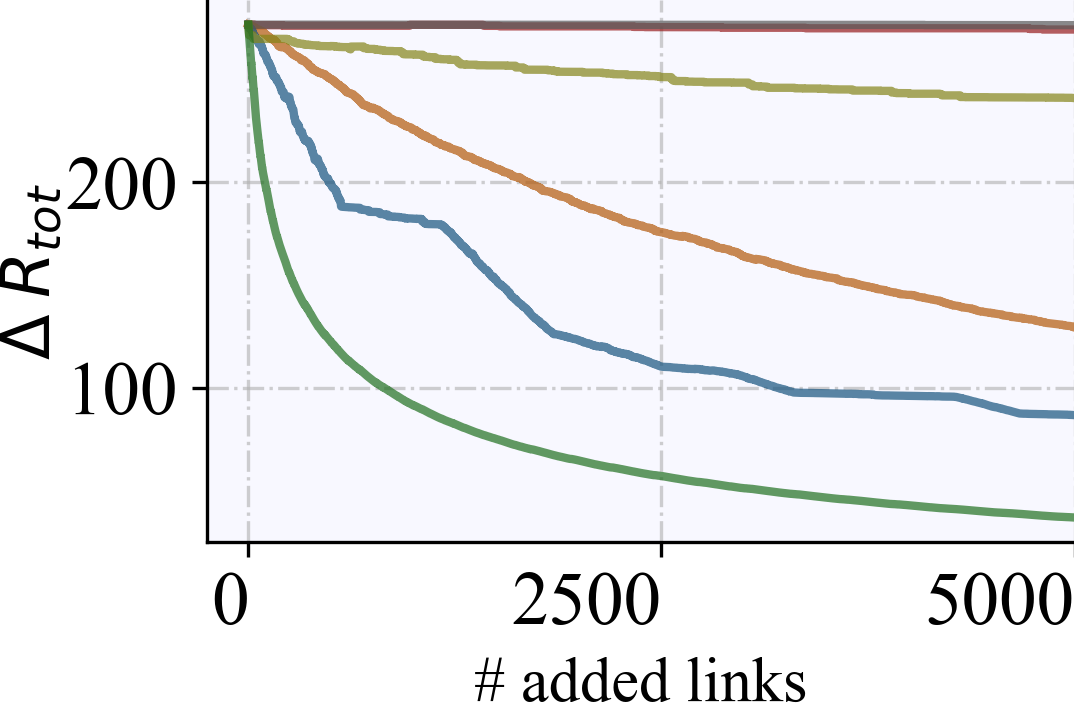}
    \end{subfigure}
    \hfill
    \begin{subfigure}{0.2\textwidth}
    \centering
    \includegraphics[width=\linewidth]{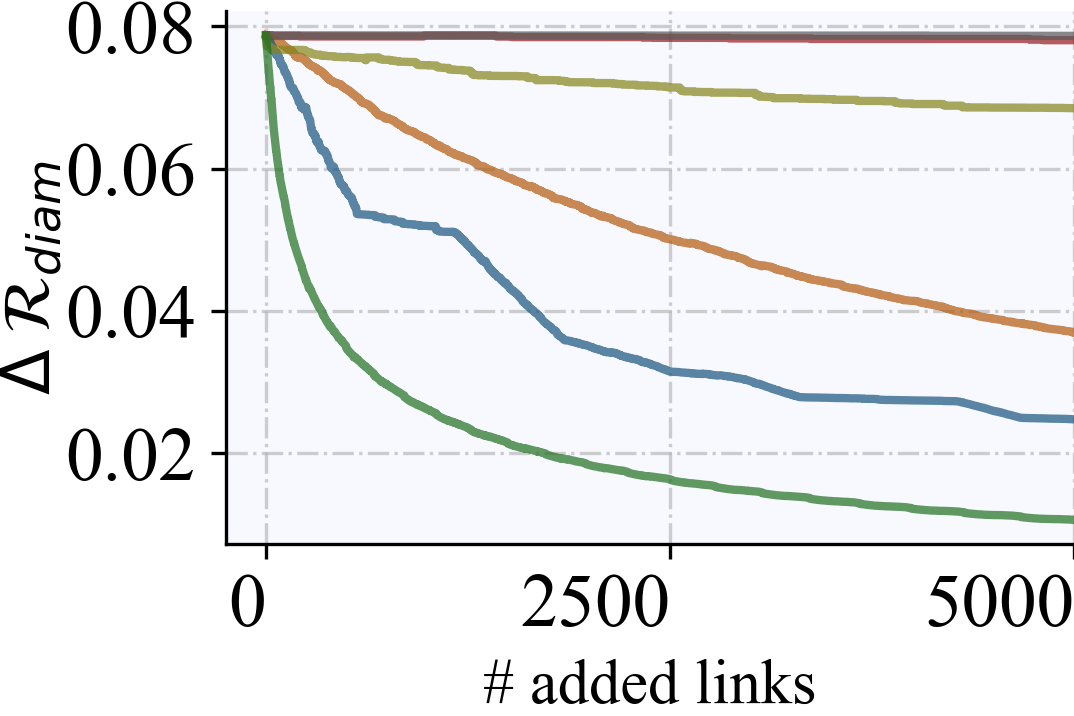}
    \end{subfigure}
    \hfill
    \begin{subfigure}{0.2\textwidth}
    \centering
    \includegraphics[width=\linewidth]{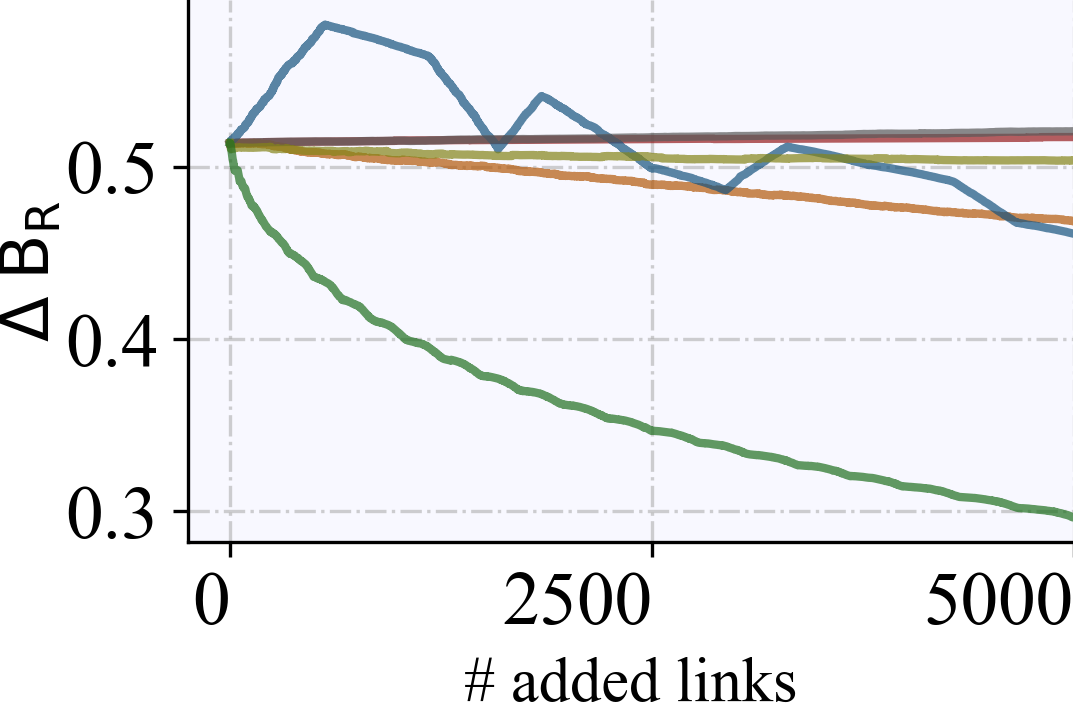}
    \end{subfigure}
    \caption{Evolution of the group metrics (top-row) and structural group unfairness metrics (bottom-row) as the number of added edges increases on the Google+ dataset with a total budget of 5,000 new edges.}
    \label{fig:googleresevolution}
\end{figure*}

\subsection{Structural Group Unfairness Mitigation}

Table~\ref{tab:expfairmetrics} depicts the three  structural group unfairness measures on the original graph $G$ and after adding 50 edges to the Facebook dataset and 5,000 edges to the UNC28 and Google+ datasets. 
The groups are defined based on gender (male, female) and the disadvantaged group are females according to the three structural group unfairness measures, as depicted in Table~\ref{tab:socialcapital} and top row of Table~\ref{tab:expfairmetrics}. The disparities in social capital between groups are particularly large in the Google+ network with a $\Delta\Rg$ of 276.4, meaning that females have significantly lower levels of information flow than males in this network. $\Delta\Bt$ also shows a difference of 0.51 on the control of the network, which is a large difference given that the values of $\Bt$ are in the range $[0, 2-2/|\mathcal{V}|]$. 

Regarding the results of the graph intervention algorithms, we observe how edge augmentation via \ERG{} outperforms all the baseline methods on the three datasets in terms of reducing disparities in group social capital. Interestingly, the larger the unfairness in the original graph, the larger the improvement after the intervention with \ERG{}. For example, the isolation disparity, $\Delta\Rg$,  %
improves by \textbf{75\%}, \textbf{61\%} and \textbf{87\%} after \ERG{}'s intervention in the Facebook, UNC28 and Google+ networks, respectively. A similar behavior is observed for the other structural group unfairness measures.

For illustration purposes, we also report the results of applying Algorithm~\ref{alg:erp} but where the added edges connect the nodes with the smallest ---instead of the largest--- $R_{uv}$, similar to how link recommendation algorithms work. We refer to this method as the ``Strong'' version and denote it with an ``S-''. %
As expected, strong edges do not improve the structural group unfairness since such methods are not designed to improve the information flow in the most isolated nodes in the graph, but to connect nodes that are already structurally close to each other (foster strong ties).

Note how the edge augmentation when using DW and Cos distances or via the curvature (SDRF) and $\lambda_2$ (FOSR) methods yields graphs with significant levels of structural group unfairness. Furthermore, in the case of using DW distances, the improvement in performance worsens as the graph gets larger (Google+). 
While using Cos distance for edge augmentation improves $\Delta\Rg$ and $\Delta\Rd$, it is unable to always improve $\Delta\Bt$ due to the inherently more intricate nature of control disparity optimization. Unlike $\Delta\Rg$ and $\Delta\Rd$, minimizing $\Delta\Bt$ requires the precise identification of network gaps which are difficult to detect using a cosine similarity distance. 

Over-squashing oriented methods (SDRF and FOSR) aim to improve the information flow by focusing only on the information bottlenecks of the network, which leads to a small effect on improving the overall information flow and reducing disparities. SDRF identifies where to add the edges partly based on the control $\Bt(u)$ (Eq.~\ref{app:eq:linkcurvature}) rather than based on the diameter $\Rd$ as \ERG{} does. FOSR's goal is based solely on increasing the spectral bottleneck,---i.e. $\lambda_2$---, while \ERG{} considers the entire spectrum, leading to a better characterization and improvement of the information flow.

For completeness, the social capital metrics for each group in the three data sets are shown in Tab.~\ref{tab:groupimprovement} in Appendix~\ref{app:exp:gorupmetrics}. Finally, edge augmentation with \ERG{} is not only more effective in mitigating the group structural unfairness in the graph, but also more computationally efficient than most of the baselines as shown in Appendix~\ref{app:computation}.

\subsection{Overall Social Capital Improvement} 

Figures~\ref{fig:paretomain} and~\ref{fig:googleresevolution} illustrate how \ERG{} is more effective than the baseline methods not only to reduce the group unfairness in social capital ($\Delta\Rg$ and $\Delta\Rd$), but also to improve the overall social capital measures of group isolation $R_{tot}(S_i)$ and group diameter $R_{diam}(S_i)$ for all the groups in the graph.
For each dataset, the structural group unfairness metric ($\Delta\Rg$ or $\Delta\Rd$) is shown on the x-axis of the graphs and the global improvement (sum of group isolation for all groups or sum of diameters for all groups) on the y-axis. 
In both axes, the lower the values, the better.
Edge augmentation via \ERG{} clearly outperforms any other graph intervention strategy, providing evidence that it not only reduces inequalities in social capital between groups, but also improves the social capital of all groups.

\subsection{Evolution of Structural Group Unfairness Throughout the Interventions}
\label{subsec:evolution}

Figure~\ref{fig:googleresevolution} illustrates the evolution of the structural group unfairness metrics for the Google+ dataset as new edges are added to the network with a total budget of 5,000 edges. As seen in the Figure, edge augmentation via \ERG{} quickly mitigates the group isolation and diameter disparities, even after the addition of a small number of edges.  
Furthermore, edge augmentation via \ERG{} exhibits a smoother and more consistent reduction in control disparity ($\Delta\Bt$), in contrast to the stair-step behavior observed when adding edges using the baseline methods.

Note that $\Bt$ is a finite resource to be allocated among the groups and cannot be globally maximized. Hence, the right-most graphs show how $\Delta\Bt$ is improved by decreasing $\Bt$ of the group with the highest initial $\Bt$ and increasing it otherwise (top-right). The goal is to converge both to the optimal bound of $2-2/|\mathcal{V}|$ as indicated by the black line.

Figure~\ref{app:evolutionfigs} depicts the evolution of the structural group unfairness and  group social capital metrics when the budget $B$ is very small (50 edges) and large (5,000 edges) to study the efficacy of the interventions. Figure~\ref{fig:violinfb5000} in Appendix~\ref{app:evolutionfigs} shows the distribution of the nodes' social capital metrics in each of the groups in the original graph and the resulting graphs after edge augmentation.
The results are consistent: for a fixed budget, edge augmentation via \ERG{} yields the best results both for the mitigation of the structural group unfairness and the increase in the overall social capital of all the groups in the graph. Additionally, for a fixed structural group unfairness mitigation goal, edge augmentation via \ERG{} achieves it with a significantly smaller budget of added edges than any of the baseline methods. Finally, Appendix~\ref{app:sec:rewiredgraphs} contains examples of synthetic graphs.

\new{\section{Discussion}}

\new{The proposed \ERG{} algorithm may be seen as an \textit{algorithmic reparative} social tool~\cite{Davis2021}, \emph{i.e.}, a computational approach designed to address and repair social inequalities, injustices or biases, in our case, in the context of social networks. Our method contributes to the use of social networks, seen as socio-technical systems~\cite{baxter2011}, for social good purposes. While an extensive body of work has been dedicated to the identification and mitigation of algorithmic bias \cite{danks2017,Barocas2016a}, %
this paper aims to trigger a deeper reflection and further discussion about the principles of algorithmic fairness and justice in relation social networks, their ultimate goals and their intricate business models of data trading among platform owners, data brokers, service and advertising companies~\cite{Srnicek2016}.} %
\new{There is not a universal interpretation of fairness~\citep{Curto2023}, yet the distributive paradigm of fairness~\citep{Rawls1971} dominates the political and philosophical discussions of social justice today. As a result, most discussions and implementations of fairness focus on how outcomes or benefits are distributed across different socially relevant groups.
Nevertheless, there are alternative notions of fairness that remain unexplored in social networks. For example, a prioritarian conception of fairness~\citep{Coeckelbergh2022} would give priority to the individuals belonging to a vulnerable social group. Moreover, prominent voices in social philosophy have long defended that fairness can be understood not only as the redistribution of individual social capital, but also as the social recognition of identity groups~\citep{Fraser2003}. In this paper, we focus on group fairness and combine both a prioritarian and a Rawlsian approach since the proposed intervention aims to mitigate the social capital gap for the most disadvantaged group while increasing the social capital of all groups. Given the importance of social networks in the definition of the social fabric, the ultimate goal of our work is to spur a reflection on the potential use of social networks as \textit{reparative} tools for social inequality~\citep{Davis2021}.} 

\new{While one could argue that defining groups in terms of protected attributes can be considered a form of discrimination, our approach does not aim to systematically increase the social capital of a pre-defined vulnerable social group, but to detect the social group or groups that is structurally disadvantaged in the network and implement a mechanism at scale that benefits an entire community. This type of mitigation of structural unfairness is particularly relevant since the potential disadvantages suffered by groups in social networks can add up to already existing social conditions that contribute to systemic injustice~\citep{Zalta2011}}.

\new{The implications of our work go beyond the mitigation of disparities in group social capital and information silos within social networks. The proposed metrics and graph intervention algorithm can also promote innovation and social mobility, since new connections between individuals across different groups are created while prioritizing the information flow gain to/from the most structurally disadvantaged group, and thus strengthening the influence of individuals that suffer from systemic injustice. From a practical perspective, we envision our proposal as a complement to existing graph interventions in a hybrid setup that combines the addition of edges connecting nodes with both small (strong ties) and large (weak ties) effective resistances. \ERG{} provides incentives to users in terms of increasing their social capital~\cite{bashardoust2022reducing} and opportunities for innovation by connecting distant nodes~\cite{Rajkumar2022}. Therefore the deployment of \ERG{} implies a more \textit{socially responsible} way to manage connections.} %
\new{Note that the aim of this proposal is not to present a universal solution to group social capital disparities in social networks, but a contribution to spur further reflections regarding the societal impact of algorithms in relation to their purpose and use.}

\new{Finally, \ERG{}'s objectives are well aligned with the Digital Services Act (DSA)\footnote{\url{https://www.eu-digital-services-act.com/Digital_Services_Act_Article_34.html}} in Europe, which aims to regulate online platforms and services, enhancing the transparency, accountability and safety of social platforms. The DSA imposes strict responsibilities to very large online platforms (those with more than 45 million EU users), given their significant societal impact. These platforms must conduct regular risk assessments, reduce systemic risks --such as the spread of harmful content or the manipulation of democratic processes-- and submit to independent audits. In particular, Article 34 of the DSA states that ``providers of very large online platforms need to assess their risks to society and adapt their algorithms if necessary''  and Article 38 demands that ``providers of very large online platforms and of very large online search engines that use recommender systems shall provide at least one option for each of their recommender systems which is not based on profiling''. The solution offered by \ERG{}, where new edges are added with the goal of mitigating disparities in social capital across groups while increasing everyone's social capital, could be of great interest to providers of very large platforms in the context of this regulation.}

\section{Conclusion and Future Work} 

\new{In this paper, we have presented a novel method, based on the effective resistance, to measure and mitigate group social capital disparities within a social network, where the groups are defined according to the values of a protected attribute. Grounded in spectral graph theory, we have introduced three measures of group social capital based on the effective resistance and we have proposed to mitigate disparities in group social capital by means of \ERG{}, a budgeted edge augmentation algorithm that systematically increases the social capital of the most disadvantaged group in the network and hence reduces the disparities in group social capital. In extensive experiments with three benchmark graph datasets, we illustrate how \ERG{} is the most effective method to decrease disparities in group social capital when compared to five baselines.}%

In future work, we plan to 
explore alternative non-greedy edge augmentation algorithms to mitigate structural group unfairness using our effective resistance-based measures and we would like to incorporate additional node features corresponding to the characteristics of the individuals in the network. We also maintain ongoing discussions with several community organizations to define a project with participants that arrive in another country as refugees, for whom information access and connection with the local population of the hosting country are \new{of paramount importance}.

\bibliography{ICSWM24}

\begin{thebibliography}{74}
\providecommand{\natexlab}[1]{#1}

\bibitem[{Arnaiz-Rodr\'{i}guez et~al.(2022)Arnaiz-Rodr\'{i}guez, Begga, Escolano, and Oliver}]{arnaiz2022diffwire}
Arnaiz-Rodr\'{i}guez, A.; Begga, A.; Escolano, F.; and Oliver, N. 2022.
\newblock {DiffWire: Inductive Graph Rewiring via the Lov{á}sz Bound}.
\newblock In \emph{Proceedings of the First Learning on Graphs Conference}, volume 198 of \emph{Proceedings of Machine Learning Research}, 15:1--15:27. PMLR.

\bibitem[{Ball and Newman(2013)}]{ball2013}
Ball, B.; and Newman, M.~E. 2013.
\newblock Friendship networks and social status.
\newblock \emph{Network Science}, 1(1): 16--30.

\bibitem[{Barab{\'a}si(2013)}]{barabasi2013network}
Barab{\'a}si, A.-L. 2013.
\newblock Network science.
\newblock \emph{Philosophical Transactions of the Royal Society A: Mathematical, Physical and Engineering Sciences}, 371(1987): 20120375.

\bibitem[{Barocas and Selbst(2016)}]{Barocas2016a}
Barocas, S.; and Selbst, A.~D. 2016.
\newblock {Big Data's Disparate Impact}.
\newblock \emph{SSRN Electronic Journal}.

\bibitem[{Bashardoust et~al.(2023)Bashardoust, Friedler, Scheidegger, Sullivan, and Venkatasubramanian}]{bashardoust2022reducing}
Bashardoust, A.; Friedler, S.~A.; Scheidegger, C.~E.; Sullivan, B.~D.; and Venkatasubramanian, S. 2023.
\newblock Reducing Access Disparities in Networks using Edge Augmentation.
\newblock In \emph{Proceedings of the 2023 FAccT}. ACM.

\bibitem[{Baxter and Sommerville(2011)}]{baxter2011}
Baxter, G.; and Sommerville, I. 2011.
\newblock Socio-technical systems: From design methods to systems engineering.
\newblock \emph{Interacting with computers}, 23(1): 4--17.

\bibitem[{Black et~al.(2023)Black, Wan, Nayyeri, and Wang}]{black2023understanding}
Black, M.; Wan, Z.; Nayyeri, A.; and Wang, Y. 2023.
\newblock Understanding Oversquashing in {GNN}s through the Lens of Effective Resistance.
\newblock In \emph{Proceedings of the 40th ICML}.

\bibitem[{Borgatti(2005)}]{borgatti2005}
Borgatti, S.~P. 2005.
\newblock Centrality and network flow.
\newblock \emph{Social networks}, 27(1): 55--71.

\bibitem[{Borgatti, Jones, and Everett(1998)}]{borgatti1998network}
Borgatti, S.~P.; Jones, C.; and Everett, M.~G. 1998.
\newblock Network measures of social capital.
\newblock \emph{Connections}, 21(2): 27--36.

\bibitem[{Bozzo and Franceschet(2013)}]{bozzo2013resistance}
Bozzo, E.; and Franceschet, M. 2013.
\newblock Resistance distance, closeness, and betweenness.
\newblock \emph{Social Networks}, 35(3): 460--469.

\bibitem[{Brandes and Fleischer(2005)}]{brandes2005centrality}
Brandes, U.; and Fleischer, D. 2005.
\newblock Centrality measures based on current flow.
\newblock In \emph{Annual symposium on theoretical aspects of computer science}, 533--544. Springer.

\bibitem[{Burt(1999)}]{burt1999leaders}
Burt, R.~S. 1999.
\newblock The social capital of opinion leaders.
\newblock \emph{The Annals of the American Academy of Political and Social Science}, 566(1): 37--54.

\bibitem[{Burt(2000)}]{Burt2000}
Burt, R.~S. 2000.
\newblock The network structure of social capital.
\newblock \emph{Research in organizational behavior}, 22: 345--423.

\bibitem[{Burt(2004)}]{burt2004}
Burt, R.~S. 2004.
\newblock Structural holes and good ideas.
\newblock \emph{American journal of sociology}, 110(2): 349--399.

\bibitem[{Chandra et~al.(1989)Chandra, Raghavan, Ruzzo, and Smolensky}]{chandra1989electrical}
Chandra, A.~K.; Raghavan, P.; Ruzzo, W.~L.; and Smolensky, R. 1989.
\newblock The electrical resistance of a graph captures its commute and cover times.
\newblock In \emph{Proceedings of the twenty-first annual ACM symposium on Theory of computing}, 574--586.

\bibitem[{Chung(1997)}]{chung1997spectral}
Chung, F.~R. 1997.
\newblock \emph{Spectral graph theory}, volume~92.
\newblock American Mathematical Soc.

\bibitem[{Coeckelbergh(2022)}]{Coeckelbergh2022}
Coeckelbergh, M. 2022.
\newblock \emph{{The Political Philosophy of AI}}.
\newblock Polity.

\bibitem[{Coifman et~al.(2005)Coifman, Lafon, Lee, Maggioni, Nadler, Warner, and Zucker}]{coifman2005geometric}
Coifman, R.~R.; Lafon, S.; Lee, A.~B.; Maggioni, M.; Nadler, B.; Warner, F.; and Zucker, S.~W. 2005.
\newblock Geometric diffusions as a tool for harmonic analysis and structure definition of data: Diffusion maps.
\newblock \emph{Proceedings of the NAS of the United States of America}, 102(21): 7426--7431.

\bibitem[{Coleman(1988)}]{coleman1988}
Coleman, J.~S. 1988.
\newblock Social capital in the creation of human capital.
\newblock \emph{American journal of sociology}, 94: S95--S120.

\bibitem[{Curto and Comim(2023)}]{Curto2023}
Curto, G.; and Comim, F. 2023.
\newblock {SAF: Stakeholder's Agreement on Fairness in the Practice of Machine Learning Developent}.
\newblock \emph{Science and Engineering Ethics}, 29.

\bibitem[{Danks and London(2017)}]{danks2017}
Danks, D.; and London, A.~J. 2017.
\newblock Algorithmic Bias in Autonomous Systems.
\newblock In \emph{Ijcai}, volume~17, 4691--4697.

\bibitem[{Davis, Williams, and Yang(2021)}]{Davis2021}
Davis, J.; Williams, A.; and Yang, M. 2021.
\newblock {Algorithmic reparation}.
\newblock \emph{Big Data and Society}, 8(2).

\bibitem[{Devriendt and Lambiotte(2022)}]{devriendt2022discrete}
Devriendt, K.; and Lambiotte, R. 2022.
\newblock Discrete curvature on graphs from the effective resistance.
\newblock \emph{Journal of Physics: Complexity}, 3(2): 025008.

\bibitem[{Di~Giovanni et~al.(2023)Di~Giovanni, Giusti, Barbero, Luise, Lio, and Bronstein}]{digiovanni2023over}
Di~Giovanni, F.; Giusti, L.; Barbero, F.; Luise, G.; Lio, P.; and Bronstein, M. 2023.
\newblock {On Over-Squashing in Message Passing Neural Networks: The Impact of Width, Depth, and Topology}.
\newblock In \emph{In Proceedings of the 40th ICML}.

\bibitem[{Dong et~al.(2023)Dong, Ma, Wang, Chen, and Li}]{dong2023fairness}
Dong, Y.; Ma, J.; Wang, S.; Chen, C.; and Li, J. 2023.
\newblock Fairness in graph mining: A survey.
\newblock \emph{IEEE Transactions on Knowledge and Data Engineering}.

\bibitem[{Doyle and Snell(1984)}]{doyle1984rayleighs}
Doyle, P.~G.; and Snell, J.~L. 1984.
\newblock \emph{Random walks and electric networks}, volume~22.
\newblock American Mathematical Soc.

\bibitem[{Ellens et~al.(2011)Ellens, Spieksma, Van~Mieghem, Jamakovic, and Kooij}]{ellens2011effective}
Ellens, W.; Spieksma, F.~M.; Van~Mieghem, P.; Jamakovic, A.; and Kooij, R.~E. 2011.
\newblock Effective graph resistance.
\newblock \emph{Linear algebra and its applications}, 435(10): 2491--2506.

\bibitem[{Fish et~al.(2019)Fish, Bashardoust, Boyd, Friedler, Scheidegger, and Venkatasubramanian}]{fish2019gaps}
Fish, B.; Bashardoust, A.; Boyd, D.; Friedler, S.; Scheidegger, C.; and Venkatasubramanian, S. 2019.
\newblock Gaps in information access in social networks?
\newblock In \emph{The World Wide Web Conference}, 480--490.

\bibitem[{Fouss et~al.(2007)Fouss, Pirotte, Renders, and Saerens}]{fouss2007random}
Fouss, F.; Pirotte, A.; Renders, J.-M.; and Saerens, M. 2007.
\newblock Random-walk computation of similarities between nodes of a graph with application to collaborative recommendation.
\newblock \emph{IEEE Transactions on knowledge and data engineering}, 19(3): 355--369.

\bibitem[{Fraser and Honneth(2003)}]{Fraser2003}
Fraser, N.; and Honneth, A. 2003.
\newblock \emph{{Redistribution or recognition? A political -philosophical exchange}}.
\newblock Verso Books.

\bibitem[{Freeman(1977)}]{freeman1977betweenness}
Freeman, L.~C. 1977.
\newblock A set of measures of centrality based on betweenness.
\newblock \emph{Sociometry}, 35--41.

\bibitem[{Garimella et~al.(2018)Garimella, De~Francisci~Morales, Gionis, and Mathioudakis}]{garimella2018political}
Garimella, K.; De~Francisci~Morales, G.; Gionis, A.; and Mathioudakis, M. 2018.
\newblock Political discourse on social media: Echo chambers, gatekeepers, and the price of bipartisanship.
\newblock In \emph{Proceedings of the 2018 WWW conference}, 913--922.

\bibitem[{Ghosh, Boyd, and Saberi(2008)}]{ghosh2008minimizing}
Ghosh, A.; Boyd, S.; and Saberi, A. 2008.
\newblock Minimizing effective resistance of a graph.
\newblock \emph{SIAM review}, 50(1): 37--66.

\bibitem[{Granovetter(1983)}]{granovetter1983strength}
Granovetter, M. 1983.
\newblock The strength of weak ties: A network theory revisited.
\newblock \emph{Sociological theory}, 201--233.

\bibitem[{G{\"u}ndo{\u{g}}du et~al.(2019)G{\"u}ndo{\u{g}}du, Panzarasa, Oliver, and Lepri}]{gundougdu2019bridging}
G{\"u}ndo{\u{g}}du, D.; Panzarasa, P.; Oliver, N.; and Lepri, B. 2019.
\newblock The bridging and bonding structures of place-centric networks: Evidence from a developing country.
\newblock \emph{PloS one}, 14(9): e0221148.

\bibitem[{Hohmann, Devriendt, and Coscia(2023)}]{hohmann2023polarizationCT}
Hohmann, M.; Devriendt, K.; and Coscia, M. 2023.
\newblock Quantifying ideological polarization on a network using generalized Euclidean distance.
\newblock \emph{Science Advances}, 9(9): eabq2044.

\bibitem[{Jackson(2019)}]{jackson2019human}
Jackson, M.~O. 2019.
\newblock \emph{The human network: How your social position determines your power, beliefs, and behaviors}.
\newblock Vintage.

\bibitem[{Karhadkar, Banerjee, and Mont{\'u}far(2023)}]{karhadkar2023fosr}
Karhadkar, K.; Banerjee, P.~K.; and Mont{\'u}far, G. 2023.
\newblock Fo{SR}: First-order spectral rewiring for addressing oversquashing in {GNN}s.
\newblock In \emph{The Eleventh International Conference on Learning Representations}.

\bibitem[{Katz(1953)}]{katz1953new}
Katz, L. 1953.
\newblock A new status index derived from sociometric analysis.
\newblock \emph{Psychometrika}, 18(1): 39--43.

\bibitem[{Kempe, Kleinberg, and Tardos(2003)}]{kempe2003maximizing}
Kempe, D.; Kleinberg, J.; and Tardos, {\'E}. 2003.
\newblock Maximizing the spread of influence through a social network.
\newblock In \emph{Proceedings of the ninth ACM SIGKDD}, 137--146.

\bibitem[{Klein and Randi{\'c}(1993)}]{klein1993resistance}
Klein, D.~J.; and Randi{\'c}, M. 1993.
\newblock Resistance distance.
\newblock \emph{Journal of mathematical chemistry}, 12: 81--95.

\bibitem[{Leskovec and Mcauley(2012)}]{leskovec2012learning}
Leskovec, J.; and Mcauley, J. 2012.
\newblock Learning to discover social circles in ego networks.
\newblock \emph{Advances in neural information processing systems}, 25.

\bibitem[{Lipman, Rustamov, and Funkhouser(2010)}]{lipman2010biharmonic}
Lipman, Y.; Rustamov, R.~M.; and Funkhouser, T.~A. 2010.
\newblock Biharmonic distance.
\newblock \emph{ACM Transactions on Graphics (TOG)}, 29(3): 1--11.

\bibitem[{Lov{\'a}sz(1993)}]{Lovasz1996}
Lov{\'a}sz, L. 1993.
\newblock Random walks on graphs.
\newblock \emph{Combinatorics, Paul erdos is eighty}, 2(1-46): 4.

\bibitem[{Nahapiet and Ghoshal(1998)}]{nahapiet1998}
Nahapiet, J.; and Ghoshal, S. 1998.
\newblock Social capital, intellectual capital, and the organizational advantage.
\newblock \emph{Academy of management review}, 23(2): 242--266.

\bibitem[{Newman(2003)}]{newman2003mixing}
Newman, M.~E. 2003.
\newblock Mixing patterns in networks.
\newblock \emph{Physical review E}, 67(2): 026126.

\bibitem[{Newman(2005)}]{newman2005measure}
Newman, M.~E. 2005.
\newblock A measure of betweenness centrality based on random walks.
\newblock \emph{Social networks}, 27(1): 39--54.

\bibitem[{Page et~al.(1998)Page, Brin, Motwani, and Winograd}]{page1998pagerank}
Page, L.; Brin, S.; Motwani, R.; and Winograd, T. 1998.
\newblock The pagerank citation ranking: Bring order to the web.
\newblock Technical report, technical report, Stanford University.

\bibitem[{Perozzi, Al-Rfou, and Skiena(2014)}]{perozzi2014deepwalk}
Perozzi, B.; Al-Rfou, R.; and Skiena, S. 2014.
\newblock Deepwalk: Online learning of social representations.
\newblock In \emph{Proceedings of the 20th ACM SIGKDD conference}, 701--710.

\bibitem[{Putnam(2015)}]{putnam2015}
Putnam, R.~D. 2015.
\newblock Bowling alone: America’s declining social capital.
\newblock In \emph{The city reader}, 188--196. Routledge.

\bibitem[{Qiu and Hancock(2007)}]{qiu2007clustering}
Qiu, H.; and Hancock, E.~R. 2007.
\newblock Clustering and embedding using commute times.
\newblock \emph{IEEE Transactions on Pattern Analysis and Machine Intelligence}, 29(11): 1873--1890.

\bibitem[{Rajkumar et~al.(2022)Rajkumar, Saint-Jacques, Bojinov, Brynjolfsson, and Aral}]{Rajkumar2022}
Rajkumar, K.; Saint-Jacques, G.; Bojinov, I.; Brynjolfsson, E.; and Aral, S. 2022.
\newblock {A causal test of the strength of weak ties}.
\newblock \emph{Science (American Association for the Advancement of Science)}, 377(6612): 1304--1310.

\bibitem[{Ramp{\'a}{\v{s}}ek et~al.(2022)Ramp{\'a}{\v{s}}ek, Galkin, Dwivedi, Luu, Wolf, and Beaini}]{rampavsek2022recipe}
Ramp{\'a}{\v{s}}ek, L.; Galkin, M.; Dwivedi, V.~P.; Luu, A.~T.; Wolf, G.; and Beaini, D. 2022.
\newblock Recipe for a general, powerful, scalable graph transformer.
\newblock \emph{Advances in Neural Information Processing Systems}, 35: 14501--14515.

\bibitem[{Rawls(1971)}]{Rawls1971}
Rawls, J. 1971.
\newblock \emph{{A Theory of Justice}}.
\newblock Oxford: Oxford University Press.

\bibitem[{Red et~al.(2011)Red, Kelsic, Mucha, and Porter}]{red2011comparing}
Red, V.; Kelsic, E.~D.; Mucha, P.~J.; and Porter, M.~A. 2011.
\newblock Comparing community structure to characteristics in online collegiate social networks.
\newblock \emph{SIAM review}, 53(3): 526--543.

\bibitem[{Santos, Lelkes, and Levin(2021)}]{santos2021link}
Santos, F.~P.; Lelkes, Y.; and Levin, S.~A. 2021.
\newblock Link recommendation algorithms and dynamics of polarization in online social networks.
\newblock \emph{Proceedings of the National Academy of Sciences}, 118(50): e2102141118.

\bibitem[{Saxena, Fletcher, and Pechenizkiy(2022)}]{saxena2022fairsna}
Saxena, A.; Fletcher, G.; and Pechenizkiy, M. 2022.
\newblock Fairsna: Algorithmic fairness in social network analysis.
\newblock \emph{arXiv preprint arXiv:2209.01678}.

\bibitem[{Srnicek(2016)}]{Srnicek2016}
Srnicek, N. 2016.
\newblock \emph{{Platform capitalism}}.
\newblock Wiley.

\bibitem[{Stephenson and Zelen(1989)}]{stephenson1989rethinking}
Stephenson, K.; and Zelen, M. 1989.
\newblock Rethinking centrality: Methods and examples.
\newblock \emph{Social networks}, 11(1): 1--37.

\bibitem[{Su, Sharma, and Goel(2016)}]{su2016effect}
Su, J.; Sharma, A.; and Goel, S. 2016.
\newblock The effect of recommendations on network structure.
\newblock In \emph{Proceedings of the International WWW Conference}, 1157--1167.

\bibitem[{Szreter and Woolcock(2004)}]{szreter2004}
Szreter, S.; and Woolcock, M. 2004.
\newblock Health by association? Social capital, social theory, and the political economy of public health.
\newblock \emph{International journal of epidemiology}, 33(4): 650--667.

\bibitem[{Tauch, Liu, and Pears(2015)}]{tauch2015measuring}
Tauch, S.; Liu, W.; and Pears, R. 2015.
\newblock Measuring cascade effects in interdependent networks by using effective graph resistance.
\newblock In \emph{2015 INFOCOM Workshops}, 683--688. IEEE.

\bibitem[{Tetali(1991)}]{tetali1991effres}
Tetali, P. 1991.
\newblock Random walks and the effective resistance of networks.
\newblock \emph{Journal of Theoretical Probability}, 4: 101--109.

\bibitem[{Tizghadam and Leon-Garcia(2009)}]{tizghadam2009graph}
Tizghadam, A.; and Leon-Garcia, A. 2009.
\newblock A graph theoretical approach to traffic engineering and network control problem.
\newblock In \emph{2009 21st International Teletraffic Congress}, 1--8.

\bibitem[{Tizghadam and Leon-Garcia(2010)}]{tizghadam2010betres}
Tizghadam, A.; and Leon-Garcia, A. 2010.
\newblock Betweenness centrality and resistance distance in communication networks.
\newblock \emph{IEEE Network}, 24(6): 10--16.

\bibitem[{Tong et~al.(2012)Tong, Prakash, Eliassi-Rad, Faloutsos, and Faloutsos}]{tong2012gelling}
Tong, H.; Prakash, B.~A.; Eliassi-Rad, T.; Faloutsos, M.; and Faloutsos, C. 2012.
\newblock Gelling, and melting, large graphs by edge manipulation.
\newblock In \emph{Proceedings of the 21st ACM CIKM}, 245--254.

\bibitem[{Topping et~al.(2022)Topping, Giovanni, Chamberlain, Dong, and Bronstein}]{topping2022understanding}
Topping, J.; Giovanni, F.~D.; Chamberlain, B.~P.; Dong, X.; and Bronstein, M.~M. 2022.
\newblock Understanding over-squashing and bottlenecks on graphs via curvature.
\newblock In \emph{International Conference on Learning Representations}.

\bibitem[{{United Nations}(1966)}]{UnitedNations1966}
{United Nations}. 1966.
\newblock {International Covenant on Civil and Political Rights}.

\bibitem[{Wang, Varol, and Eliassi-Rad(2022)}]{wang2022information}
Wang, X.; Varol, O.; and Eliassi-Rad, T. 2022.
\newblock Information access equality on generative models of complex networks.
\newblock \emph{Applied Network Science}, 7(1): 1--20.

\bibitem[{Woolcock et~al.(2001)}]{woolcock2001place}
Woolcock, M.; et~al. 2001.
\newblock The place of social capital in understanding social and economic outcomes.
\newblock \emph{Canadian journal of policy research}, 2(1): 11--17.

\bibitem[{Wu et~al.(2022)Wu, Sun, Zhang, Xie, and Cui}]{wu2022graph}
Wu, S.; Sun, F.; Zhang, W.; Xie, X.; and Cui, B. 2022.
\newblock Graph neural networks in recommender systems: a survey.
\newblock \emph{ACM Computing Surveys}, 55(5): 1--37.

\bibitem[{Zalta and Nodelman(2011)}]{Zalta2011}
Zalta, E.~N.; and Nodelman, U. 2011.
\newblock {Stanford Encyclopedia of Philosophy}.

\bibitem[{Zhang, Anderson, and Zhan(2011)}]{zhang2011differentiated}
Zhang, S.; Anderson, S.~G.; and Zhan, M. 2011.
\newblock The differentiated impact of bridging and bonding social capital on economic well-being: An individual level perspective.
\newblock \emph{J. Soc. \& Soc. Welfare}, 38: 119.

\bibitem[{Zhu et~al.(2020)Zhu, Yan, Zhao, Heimann, Akoglu, and Koutra}]{zhu2020homophily}
Zhu, J.; Yan, Y.; Zhao, L.; Heimann, M.; Akoglu, L.; and Koutra, D. 2020.
\newblock Beyond homophily in graph neural networks: Current limitations and effective designs.
\newblock \emph{Advances in neural information processing systems}, 33: 7793--7804.

\end{thebibliography}

\appendix

\section{Effective Resistance and information flow}
\label{app:extrabackgroundCT}

\subsection{Graph Diffusion Measures}

\paragraph{Discrete Information Propagation} Information propagation in networks has been widely studied~\citep{kempe2003maximizing}, prominently by means of graph diffusion and Random Walks methods. A Random Walk (RW) on a graph is a Markov chain that starts at a given node $i$, and moves randomly to another node from its neighborhood with probability $1/D_{i,i}$. The RW transition probability matrix is given by $\mathbf{P}=\mathbf{D}^{-1}\mathbf{A}$ and defines the discrete probability of a random walker to move from node $u$ to node $v$. $\mathbf{P}^k$ is the $k$-th power of the transition matrix $P$: the entry $(P^k)_{ij}$ denotes the probability of transitioning from node $i$ to node $j$ in exactly $k$ steps. The graph's diffusion matrix is defined as $\mathbf{T} = \sum_{k=0}^\infty \theta_k \mathbf{P}^k$, and it represents the cumulative effect of multiple steps of a random walk on the graph. Each entry $T_{ij}$ of $\mathbf{T}$ corresponds to the probability of transitioning from node $i$ to node $j$ over an infinite number of steps. 
$\theta_k$ is known as the teleport probability at step $k$ in the random walk. It quantifies the likelihood that, at each step, the random walker will teleport to a random node instead of following an edge. Thus, the sequence $\{\theta_k\}$ is a series of teleport probabilities over the steps. The resulting $T$ captures the cumulative probabilities of transitioning between nodes over an infinite number of steps in the random walk, such that  the probability of co-occurrence of two nodes on a random walk corresponds to the probability of information flowing between these two nodes.

However, this approach to assess information flow between nodes in a graph has several limitations. First, it requires considering all the potential paths in a graph, which might not be computationally feasible for large graphs. To overcome this issue, a value of $k$ is typically chosen, which limits the power of the method. Second, the teleport probabilities, $\theta_k$, need to be defined for each $k$-hop. Several methods have studied how to approximate it, such as Independent Cascade~\citep{kempe2003maximizing}, Katz~\citep{katz1953new}, SIR or PageRank~\citep{page1998pagerank}. Independent Cascade or SIR methods~\citep{kempe2003maximizing} are based on infection models, where they sample guided random walks and, thus, usually rely on expensive Monte Carlo simulations leading to a sub-optimal probability of transition, unable to consider the topology of the entire graph.

\paragraph{Graph Continuous Diffusion Metrics} Graph continuous diffusion metrics ---such as the Heat kernel distance,~\citep{coifman2005geometric}, effective resistance (or \emph{commute times} distance)~\citep{klein1993resistance, fouss2007random} or the bi-harmonic distance~\citep{lipman2010biharmonic}--- arise as a generalization of random walk metrics. Their mathematical foundations allow for a better characterization of the information flow and an intuitive interpretation of the diffusion processes in a network. %

Diffusion metrics define distances based on fine-grained nuances of the topology of the graph that are not captured by simple geodesic distances. When two nodes can be reached by many paths, they should be \emph{closer} than when they can be reached only by few paths of equal length. When two nodes can be reached by a set of edge-independent paths, they are \emph{closer} than when they are reached by redundant paths. Similarly, when two nodes are separated by a shorter path, they are \emph{closer} than when they are separated by a longer path~\citep{bozzo2013resistance}.

In addition, these metrics provide a node embedding, \emph{i.e.}, a numerical representation of each node in the graph that reflects its importance in the process of information diffusion. These embeddings capture the global structure of the network because they incorporate both the local and global geometry of the graph. 

The continuous diffusion metrics can be computed using the pseudo-inverse (or Green’s function) of the combinatorial graph Laplacian $\mathbf{L}=\mathbf{D}-\mathbf{A}$, or the normalized Laplacian $\mathcal{L} = \mathbf{D}^{-1/2}\mathbf{L}\mathbf{D}^{-1/2}$~\citep{fouss2007random}. The pseudo-inverse, denoted as $\mathbf{L}^+$ is computed using the spectral decomposition: $\mathbf{L}^+ = \sum_{i\ge 0} \lambda_i^{-1}\phi \phi^\top$ where $\mathbf{L} = \mathbf{\Phi \Lambda \Phi}^\top$, where $\lambda_i$ is the $i$-smallest eigenvalue of the Laplacian corresponding to the $\phi_i$ eigenvector.

In this paper, we use the effective resistance, which is a continuous diffusion metric.

\subsection{Effective Resistance and Commute Times}
\label{app:ct}

The Commute Time (CT)~\citep{Lovasz1996}, $\CT(u,v)$, is the expected number of steps that a random walker needs to go from node $u$ to $v$ and come back to $u$. The Effective Resistance, $R_{uv}$, is the Commute Time divided by the volume of the graph~\citep{klein1993resistance}. In addition to providing a distance for all pairs of nodes ---whether connected or not--- $R_{uv}$ may be viewed as an indicator of the criticality or importance of the edges in the flow of information throughout the network~\citep{tizghadam2009graph}. 

Intuitively, this distance captures how structurally similar and connected are two nodes in a graph. If two nodes are structurally similar to each other, then the effective resistance between them will be small. Conversely, if two nodes are weakly or not connected, then their effective resistance will be large. In addition, we can define a commute time embedding (CTE, $\mathbf{Z} = \sqrt{vol(G)}\mathbf{\Lambda^{-1/2}}\mathbf{\Phi}^\top$) of the nodes in the graph ---similar to the idea of the \emph{node's access signature} in~\citet{bashardoust2022reducing}---, where the Euclidean distance in such an embedding corresponds exactly to the commute times $\CT(u,v)=||Z_{u,:} - Z_{v,:}||^2=\mathbb{E}_u[v]+\mathbb{E}_v[u]=2|\mathcal{E}|R_{uv}$. This distance is upper bounded by the geodesic distance, with equality in the case of the graph being a tree. 

Note that the effective resistance does not rely on any parameter and it is an accurate metric to measure the graph's information flow~\citep{,chandra1989electrical,ghosh2008minimizing,arnaiz2022diffwire}, as explained next. Finally, note that $R_{uv}$ can be computed in a spectral manner:
\begin{equation}\label{eq:effresspec}
R_{uv} =\sum_{i>0} \frac{1}{\lambda_i} \left(\phi_i(u)-\phi_i(v)\right)^2
\end{equation}
where it become explicit that $R_{uv}$ depends on all the eigenvectors of the Laplacian, leading to a better characterization of the information flow.

\paragraph{Information Flow in a Graph}
The graph's information flow is given by the graph's conductance, which is measured leveraging the Cheeger Constant, $h_G$, of a graph~\citep{chung1997spectral}:
\begin{equation}\label{eq:cheeger}
    h_G=\min_{H\subseteq V} \frac{|\{e=(u,v): u\in S, v\in \bar{H}\}|}{\min (vol(H), vol(\bar{H}))}
\end{equation}
The larger $h_G$, the harder it is to disconnect the graph into separate communities. Therefore, to increase the information flow in the network, one could add edges to the original graph $G$, creating a new graph $G'$, such that $h_{G'} > h_G$.
In addition, by virtue of the Cheeger Inequality, $h_G$ is bounded by the smallest non-zero eigenvalue of $\mathbf{L}$ defined as $\lambda_2$:
\begin{equation}\label{cheegerineq}
    2h_G \leq \lambda_2 < \frac{h_G^2}{2}
\end{equation}
Finally, the $\CT$ is bounded by $\lambda_2$ as per the Lov\'{a}sz Bound~\citep{Lovasz1996}  
\begin{equation}\label{lovasz}
    \left| \frac{\CT(u,v)}{\text{vol}(G)}-\left(\frac{1}{d_u} + \frac{1}{d_v}\right)\right|\le \frac{1}{\lambda_2}\frac{2}{d_{min}}
\end{equation}

where vol($G$) is the volume of $G$, \emph{i.e.}, the sum of the degrees of the all nodes in the graph; $d_u$, $d_v$ are the degrees of nodes $u$ and $v$, respectively; and $d_{min}$ is the minimum degree in the graph.  

Therefore, a graph's information flow is bounded by $\lambda_2$ which is bounded by $h_G$. The intuition is that graphs with large $\lambda_2 \propto h_G$ have short $\CT$ distances and thus they have better information flow. Edge augmentation in a graph would lead to a new graph $G'$ where $h_{G'} > h_G$, with smaller $\CT$ distances and therefore better information flow. 

The effective resistance is also related to other ways of computing the information flow between two nodes in a graph, such as the Jacobian~\citep{topping2022understanding,digiovanni2023over,black2023understanding}.

\subsection{Measures Derived from Effective Resistance}
\label{app:relatedmetrics}

The group social capital measures that we propose in this paper (Sections~\ref{sec:groupmetrics} and~\ref{sec:disparitymetrics}) are grounded in previous metrics from the literature.

\paragraph{Total Effective Resistance}
$\Rg$~\citep{ellens2011effective} is the sum of all effective distances in the graph. A lower value of $\Rg$ indicates ease of signal propagation across the entire network and hence larger information flow. $\Rg$ is given by:
\begin{align}\label{eq:totaleffrg}
\Rg = \Rg(\mathcal{V}) &= \frac{1}{2} \mathbf{1}^\top \mathbf{R1}= \frac{1}{2}\sum_{(v,u) \in V} R_{uv}\\
&= n \sum_2^{n} \frac{1}{\lambda_n} = n \Tr(\mathbf{L}^\dagger)
\end{align}
The minimum $\Rg=|V|-1 = n-1$ is achieved in a fully connected graph. Conversely, the maximum $\Rg$ is achieved on a path graph (or linear graph) and $\Rg=\sum_i^{n-1}i = \frac{1}{2}(n(n-1))$. Therefore, $\Rg$ is ---for connected graphs--- in the range $[n-1, \frac{n(n-1)}{2}]$

Additionally, since the distance between $u$ and $v$ is the Euclidean distance in the embedding $\mathbf{Z}$,  %
$\Rg$ can be obtained as follows~\citep{ghosh2008minimizing}: 
\begin{equation}
    \Rg = \sum_{(v,u) \in V} ||Z_{u:} - Z_{v:}||^2 = n \sum_{u \in V} ||Z_{u:}||^2
\end{equation}

Similarly to $R_{uv}$, $\Rg$ is theoretically related to the connectivity of the graph defined by it smallest non-zero eigenvalue.~\citet{ellens2011effective} demonstrated the relation between $\Rg$ and $\lambda_2$:  
\begin{equation}\label{eq:diameffrg}
\frac{n}{\lambda_2} < \Rg \leq\frac{n(n-1)}{\lambda_2}.
\end{equation}

\paragraph{Resistance Diameter} The proposed group resistance diameter is based on the  resistance diameter of a graph $\Rd$, which is the maximum effective resistance on the graph~\citep{chandra1989electrical, qiu2007clustering}:
\begin{equation}
    \Rd = \max_{u,v \in V} R_{uv}
\end{equation}
$\Rd \propto \lambda_2$~\citep{chung1997spectral, chandra1989electrical}, since 
\begin{equation}\label{eq:appRdcheegerineq}
    \frac{1}{n\lambda_2} \leq \Rd \leq \frac{2}{\lambda_2}
\end{equation}
and specifically~\citep{qiu2007clustering, arnaiz2022diffwire}:
\begin{equation}\label{eq:appRdCheegerineq2}
    h_G\le \frac{\alpha^{\epsilon}}{\sqrt{\Rd\cdot \epsilon}} vol(S)^{\epsilon-1/2}, 
\end{equation}
By~\citep{qiu2007clustering} we know that
\begin{equation}
\Rd \leq \frac{1}{\lambda_2} \:\:\text{and}\:\: \Rd \leq \frac{1}{h_G^2}
\end{equation}

In addition, $\Rd$ is related to the \emph{cover time} of the graph, which is the expected time required for a random walk to visit every node at least once, \textit{i.e.}, the expected time for a piece of information to reach the entire network. $\Rd$ can be used to estimate the cover time of the graph, as per~\citep{chandra1989electrical}:
\begin{equation}\label{eq:appRdcover}
m\Rd \le \text{cover time} \le O(m\Rd \log n)
\end{equation}

\paragraph{Resistance Betweenness and Curvature}\label{app:sec:resistancebetweeness}
As the effective resistance is an information distance, this metric can be used to propose alternative betweenness or criticality metrics to the shortest path betweenness~\citep{newman2005measure}. In the literature, several effective resistance-based measures have been proposed to determine a node's criticality, such as: the current flow betweenness~\citep{tizghadam2009graph,tizghadam2010betres,bozzo2013resistance,brandes2005centrality,newman2005measure}, the resistance curvature of a node~\citep{topping2022understanding,devriendt2022discrete}, and the information bottleneck property of a node~\citep{arnaiz2022diffwire,karhadkar2023fosr,black2023understanding}. Note that the last two definitions are mathematically connected~\citep{devriendt2022discrete}.

In this work, we focus on the information bottleneck which has a close connection with the resistance curvature of a node. 
The resistance curvature of a node~\citep{devriendt2022discrete} is expressed as 
\begin{equation}
    p_u = 1 - \frac{1}{2} \sum_{v \in \mathcal{N}(u)} R_{uv}.
\end{equation}
Therefore, it fulfills the following equality:    
\begin{align}
p_u &= 1 - \frac{1}{2} \sum_{v \in \mathcal{N}(u)} R_{uv} =  1 - \frac{1}{2} \Bt(u)\nonumber \\
&\rightarrow \Bt(u)= -2(p_u-1).
\end{align}
Although the definition of a node's resistance curvature involves computing the sum of the effective resistances between the node and its neighboring nodes, the overall structure of the graph affects all $R_{uv}$'s and, consequently, the value of $p_i$ and $\Bt(u)$. The curvature of a node is bounded by $1-d_u/2\leq p_u \leq 1/2$.

In addition to the node resistance curvature, the link resistance curvature is defined in \citet{devriendt2022discrete} as
\begin{equation}
    \kappa_{uv} = \frac{2(p_u+p_j)}{R_{uv}}
\end{equation}
and writing it in terms of the proposed metrics in this work it will translate in:
\begin{align}\label{app:eq:linkcurvature}
    \kappa_{uv} &= \frac{2(p_u+p_j)}{R_{uv}} \nonumber\\
    &=\frac{2\left((1-\frac{1}{2}\Bt(u))+(1-\frac{1}{2}\Bt(v))\right)}{R_{uv}}\nonumber\\
    &=\frac{4-\Bt(u)-\Bt(v)}{R_{uv}}
\end{align}
Additionally, this is the value that the \textit{SDRF} algorithm~\citep{topping2022understanding} will use to identify the link to add links around. SDRF identifies the link with \textit{lowest} $\kappa_{uv}$ and adds an edge between a pair the neighbors of the endpoints of that edge that mostly improves the information bottleneck.

\section{Group Social Capital Metrics and Edge Augmentation Algorithm}

\subsection{Group Social Capital Metrics}

\subsubsection{Group Isolation and Isolation Disparity} Group Isolation is based on the previously explained notion of total effective resistance of a graph and its close connection with the current flow closeness centrality~\citep{newman2005measure}. 

We propose to define the isolation of a node as its total effective resistance $\Rg(u) = \sum_{v\in \mathcal{V}} R_{uv}$. This centrality metric based on all $R_{uv}$ considers all the eigenvector for its computation, since $R_{uv}$ can be defined using the whole spectrum of the graph. We proceed to derive $\Rg(u)$ in terms of the pseudo-inverse of the graph:
\begin{align}
    \Rg(i)=& \sum_{v}^N R_{uv} = \sum_{u}^N L^+_{uu} + L^+_{vv} - 2L^+_{uv} \nonumber\\
    =& N L^+_{uu} + \Tr(\mathbf{L}^+) - 2\sum_v^N L^+_{uv} \nonumber\\ 
    =& N L^+_{uu} + \Tr(\mathbf{L}^+)
\end{align}
using \citet[Theorem 4.2]{ellens2011effective} and \citet[Eq. 15] {bozzo2013resistance}. Here, $ L^+_{uu}$ indicates how close node $u$ is on expectancy to all the nodes, and a global component, $\Tr(\mathbf{L}^+)$, proportional to the average pairwise distance in the network, which is denotes of how large is the network overall, independently of the particular node $u$.

Based on the introduced $\Rg(u)$, we propose group isolation which is obtained as the expectation in $\Rg(u)$ for all the nodes in the group:
\begin{align}\label{app:eq:rtotsexp}
    \Rg(S_i) &= \mathbb{E}_{u\sim S_i}\left[\Rg(u)\right] = \left|\mathcal{V}\right|\mathbb{E}_{u\sim S_i}\left[\frac{1}{|\mathcal{V}|} \sum_{v \in \mathcal{V}} R_{uv}\right] \nonumber\\
    &=
    |\mathcal{V}| \mathbb{E}_{u\sim S_i}\left[\mathbb{E}_{v\sim \mathcal{V}}\left[R_{uv}\right]\right] \nonumber\\
    &= |\mathcal{V}|\mathbb{E}_{u\sim S_i,v\sim \mathcal{V}}\left[R_{uv}\right]
\end{align}

 $\Rg(S_i)$ is computed as the expectation of $\Rg(u)$ for all nodes in group $S_i$:
\begin{align}\label{app:eq:rtotnewexp}
    \Rg(S_i) &= |\mathcal{V}|\mathbb{E}_{u\sim S_i,v\sim \mathcal{V}}\left[R_{uv}\right]
    \\&=
    |\mathcal{V}| \frac{1}{|S_i|}\sum_{u\in S_i}\frac{1}{|\mathcal{V}|}\sum_{v\in \mathcal{V}} R_{uv}
    \nonumber\\
    &=
    \frac{1}{|S_i|}\sum_{u\in S_i}\sum_{v\in \mathcal{V}} R_{uv} =
    \frac{1}{|S_i|} \sum_{u\in S_i} \Rg(u) 
    \nonumber\\&= \mathbb{E}_{u\sim S_i}\left[\Rg(u)\right]\nonumber
\end{align}

Finally, the cumulative group isolation across all groups in the graph fulfills the following equality with the total effective resistance of the graph.
\begin{align}
 \frac{1}{2} \sum_{i \in SA} \left|S_i\right| \Rg(S_i)
&= \frac{1}{2} \sum_{i \in SA} \left|S_i\right| \frac{1}{|S_i|}\sum_{u\in S_i}\sum_{v\in \mathcal{V}} R_{uv} 
\nonumber\\
    &= \frac{1}{2} \sum_{v \in \mathcal{V}}\sum_{v\in \mathcal{V}} R_{uv} = \Rg = n \Tr(\mathbf{L}\dagger)
\end{align}
Therefore, since $\Rg$ is related to the information flow, $\Rg(S_i)$ measures the ease of information flow through group $S_i$ in the graph. %

The optimal --- yet extreme --- scenario of maximum information flow in a graph is such where all nodes in the graph are connected. Hence, $G$ will be a fully connected graph. In this scenario, the total effective resistance of the graph reaches its minimum $n-1$, where $n=|\mathcal{V}|$. The total effective resistance of all existing edges in the graph is $n-1$, and given that we have all connections, $\Rg$ also sums up to $n-1$. Therefore, every pair of nodes will be separated by $ R_{uv} = 2/n$. In this scenario, $\Rg(u)=(2/n)(n-1) = 2-2/n$ for all nodes in the graph. It leads to an group isolation of $\Rg(S_i) = 2-2/n$ for all the different groups in the graph and therefore a Isolation disparity $\Delta\Rg=0$.

Regarding group isolation disparity, Equation~\ref{eq:isodisp} can be redefined using Equation~\ref{app:eq:rtotnewexp}. It is equivalent to the equality for all groups of the mean $R_{uv}$ between of the nodes in the group and all nodes in the graph:
\begin{align}
    \Rg(S_i) &= \Rg(S_j), \forall i,j \in SA\\
    \mathbb{E}_{u\sim S_i,v\sim \mathcal{V}}\left[R_{uv}\right] &= \mathbb{E}_{u\sim S_j,v\sim \mathcal{V}}\left[R_{uv}\right], \forall i,j \in SA\times SA\nonumber
\end{align}

\subsubsection{Group Control and Control Disparity} \label{app:sec:groupcontrol}

We provide the proof for the bounds associated with the proposed control metrics. %

\paragraph{Bounds of $\Bt(u)$ and $\Bt(S_i)$} We show the proof of the bounds of the node and group control. 
\begin{theorem}\label{app:thm:nodecontrolbounds}
    The control of a node is bounded by $1\leq\Bt(u)\leq d_u$, being $d_u$ the degree of node $u$. Equality on the upper bound holds when all the edges are cut edges, \emph{i.e.,} edges that if removed the graph would become disconnected.
\end{theorem}
\begin{proof}
The resistance curvature of a node is known to be bounded by $1-d_u/2 \leq p_u \leq 1/2$~\citep{devriendt2022discrete}, and the relation between the curvature and group control is given by the equality $p_u = 1-\frac{1}{2}\Bt(u)$. Therefore, we obtain the bounds as
\begin{align}
1-\frac{d_u}{2}&\leq p_u \leq \frac{1}{2} \rightarrow\nonumber\\ \nonumber
1-\frac{d_u}{2}&\leq 1-\frac{1}{2}\Bt(u) \leq -\frac{1}{2} \rightarrow\\ \nonumber
-d_u&\leq -\Bt(u)\leq -1 \rightarrow \\\nonumber
1&\leq \Bt(u)\leq d_u
\end{align}
\end{proof}

\begin{theorem}\label{app:thm:groupcontrolbounds}
    The control of a group is bounded by $1\leq\Bt(S_i)\leq \frac{\text{vol}(S_i)}{|S_i|}$, being vol$(S_i)$ the sum of the degrees of node $u$. Thus, $\text{vol}(S_i)/|S_i|$ is the average degree of all the nodes in $S_i$. Equality on the upper bound holds when the subgraph with all nodes of $S_i$ and their neighbors is a tree graph, \textit{i.e.}, a connected acyclic undirected graph.
\end{theorem}

\begin{proof}
The control of a group is defined as $\Bt(S_i) = \mathbb{E}_{u\thicksim S_i}[\Bt(u)]$ and using Theorem~\ref{app:thm:nodecontrolbounds}, we derive the bounds of $\Bt(S_i)$ as follows:
\begin{align}
1&\leq \Bt(u)\leq d_u \rightarrow \nonumber\\
1&\leq \mathbb{E}_{u\thicksim S_i}[\Bt(u)]\leq \mathbb{E}_{u\thicksim S_i}[d_u] \nonumber\\
1&\leq\Bt(S_i)\leq \frac{\text{vol}(S_i)}{|S_i|}\nonumber\\
\end{align}
\end{proof}

\paragraph{Control as an Allocation Problem} %
Group control is ($\Bt(S_i)$) a limited resource to be distributed in the network. %
The sum of $R_{uv}$ for every edge always equals to $|\mathcal{V}|-1$~\citep{ellens2011effective}:
$$
\sum_{(u,v)\in\mathcal{E}} R_{uv} = |V|-1.
$$
Therefore, the sum of the node controls in the graph is defined as
$$
\sum_{u\in\mathcal{V}} \Bt(u) = \sum_{u\in\mathcal{V}} \sum_{v\in\mathcal{N}(u)} R_{uv} = 2\times \sum_{(u,v)\in\mathcal{E}} R_{uv} = 2|V|-2,
$$
and the expectation as:
$$
\mathbb{E}_{u\sim V}[\Bt(u)] = 2-\frac{2}{|V|}
$$
independently of the number of edges (density) of the graph. 

As a consequence of the definition of group control (Equation~\ref{eq:groupcontrol}), 
the weighted sum of group control for all groups in the graph and the weighted mean also remain constant at:
$$
\sum_{S_i \in V} |S_i|\times \Bt(S_i)= 2|V|-2
$$
and
$$
\frac{1}{|V|}\sum_{S_i \in V} |S_i|\times \Bt(S_i)= 2-\frac{2}{|V|}.
$$ 

\subsection{Group Isolation vs Group Information Share}

The group isolation $\Rg(S_i)$ is given by Equation \ref{eq:totaleffrgroup} and it is proportional to the expected information distance when sampling one node from group $S_i$ and another node at random. Thus, $\Rg(S_i)$ quantifies the information access of a group of nodes in the network. 

Conversely, the amount of information that a group of nodes $S_i$ \emph{shares} with the rest of the network is given by the expected information distance between a random node $v$ from the network and a node in group $S_i$. We define this metric as the \textit{Group Information Share} ($R_{\text{From-}S_i}(V)$).

Mathematically, both concepts are are equivalent as $\Rg(S_i)$ not only quantifies how isolated a group is (in terms of information access), but also how much it contributes to the information flow that reaches other nodes in the network. Therefore, the disparity in information access between groups and the disparity in information sharing between groups on a randomly selected node in the network are also equivalent. 

\paragraph{Mathematical Equivalence}

The Group Isolation is defined as the expected total effective resistance (or information distance, $R_{uv}$) between a randomly chosen node in group $S_i$ to all other nodes in the network:
\begin{equation}
\Rg(S_i) = \mathbb{E}_{u \sim S_i}\left[R_{\text{tot}}(u)\right] = |\mathcal{V}|\mathbb{E}_{u \sim S_i}[\mathbb{E}_{v \sim \mathcal{V}}[R_{uv}]],
\end{equation} where $R_{\text{tot}}(u)$ represents the total effective resistance from node $u$ to all other nodes.

Similarly, the Group Information Share measures the expected information distance from a randomly selected node in the network to a group $S_i$:
\begin{equation}
R_{\text{From-}S_i}(V) = \mathbb{E}_{v \sim V}[R_{\text{From-}S_i}(v)] = |\mathcal{V}|\mathbb{E}_{v \sim \mathcal{V}}[\mathbb{E}_{u \sim S_i}[R_{uv}]]
\end{equation}

Thus, the expected information distance when a random node is sampled from group $S_i$ is the same as when a random node is sampled from the network and connected to group $S_i$:
\begin{align}
|\mathcal{V}|\mathbb{E}_{u \sim S_i}[\mathbb{E}_{v \sim \mathcal{V}}[R_{uv}]] &= |\mathcal{V}|\mathbb{E}_{v \sim \mathcal{V}}[\mathbb{E}_{u \sim S_i}[R_{uv}]] \\
    \Rg(S_i) &= R_{\text{From-}S_i}(V)
\end{align}

\paragraph{Group Isolation Disparity}

In terms of disparity between groups, the Group Isolation Disparity is defined as the difference in isolation between two groups $S_i$ and $S_j$:
\begin{equation}
\Delta \Rg = \Rg(S_i) - \Rg(S_j)
\end{equation}

Similarly, the Group Information Share Disparity measures the disparity in information contribution between groups:
\begin{equation}
\Delta R_{\text{From-}S}(V) = R_{\text{From-}S_i}(V) - R_{\text{From-}S_j}(V)
\end{equation}

Since $\Rg(S_i) = R_{\text{From-}S_i}(V)$, it follows that:
\begin{equation}
\Delta \Rg = \Delta R_{\text{From-}S_i}(V)
\end{equation}

Thus, the disparities in Group Isolation and in Group Information Share are the same, highlighting the dual role of the proposed metric in capturing both access to information from the network and information sharing with the rest of the network.

\subsection{Efficient Version of ERG}
Algorithm~\ref{alg:erp-efficient} shows an efficient manner to update the pseudo-inverse of the Laplacian after adding one edge to the graph. Therefore, we avoid the computation of $\mathbf{L}^\dagger$ after every edge addition. $\mathbf{L}^\dagger$ is easily updated using the Woodbury's formula which is based on the values of $\mathbf{L}^\dagger$ (see~\citet{black2023understanding} for a proof).

\SetKwRepeat{Repeat}{Repeat}{Until}{}
\begin{algorithm}
  \caption{ERG-Link}\label{alg:erp-efficient}
  \KwData{Graph $G=(\mathcal{V},\mathcal{E})$, %
  a protected attribute $SA$, 
  budget $B$ of total number of edges to add}
  \KwResult{Rewired Graph $G'=(\mathcal{V}', \mathcal{E}')$}

  \BlankLine
  $\mathbf{L} = \mathbf{D}-\mathbf{A}$\;
  $S_d = \operatorname*{argmax}_{S_i \forall i \in SA} \Rg(S_i)$ \;
  $\mathbf{L}^\dagger = \sum_{i>0} \frac{1}{\lambda_i} \phi_i \phi_i^\top = \left(\mathbf{L}+\frac{\mathbf{1}\mathbf{1}^\top}{n}\right)^{-1} - \frac{\mathbf{1}\mathbf{1}^\top}{n}$ \;
  
  \Repeat{$|\mathcal{E'}\setminus\mathcal{E}|=B$}{
  $\mathbf{R}=\mathbf{1} \diag(\mathbf{L}^\dagger)^\top + \diag(\mathbf{L}^\dagger)\mathbf{1}^\top- 2\mathbf{L}^\dagger$ \;
  $C=\{(u,v) \mid u \in S_d \text{ or } v \in S_d, (u,v) \notin E\}$ \tcp*{Select edge candidates}
  $\mathcal{E}' = \mathcal{E}'\cup \arg \max_{(u,v)\in C} R_{uv}$ \;
  
  \BlankLine
  \tcp{Fast update of $\mathbf{L}$ and $\mathbf{L}^\dagger$}
  $\mathbf{L}= \mathbf{L}+(\mathbf{e}_u-\mathbf{e}_v)(\mathbf{e}_u-\mathbf{e}_v)^\top$\;
  $\mathbf{L}^\dagger = \mathbf{L}^\dagger - \frac{1}{1+R_{uv}} \times (\mathbf{L}^\dagger_{u,:}-\mathbf{L}^\dagger_{v,:})\otimes (\mathbf{L}^\dagger_{u,:}-\mathbf{L}^\dagger_{v,:})$ \tcp*{updated by Woodbury}

}

  \BlankLine
  \Return $G'$\;
\end{algorithm}

\section{\new{Dataset Statistics}}
\label{app:data-stats}

\begin{figure*}[t]
    \begin{subfigure}{0.33\textwidth}
    \centering
    \includegraphics[width=\linewidth]{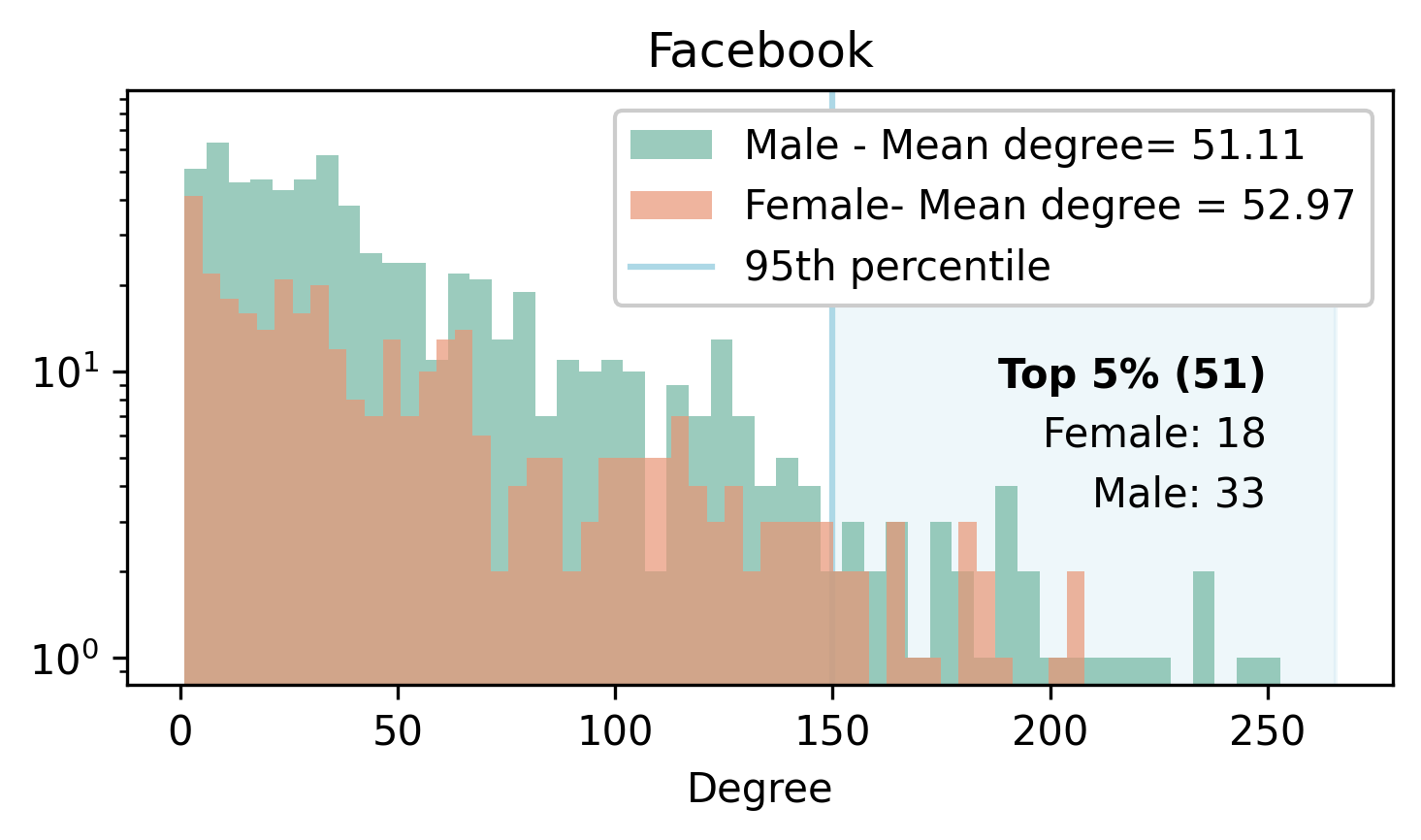}
    \end{subfigure}
    \hfill
    \begin{subfigure}{0.33\textwidth}
    \centering
    \includegraphics[width=\linewidth]{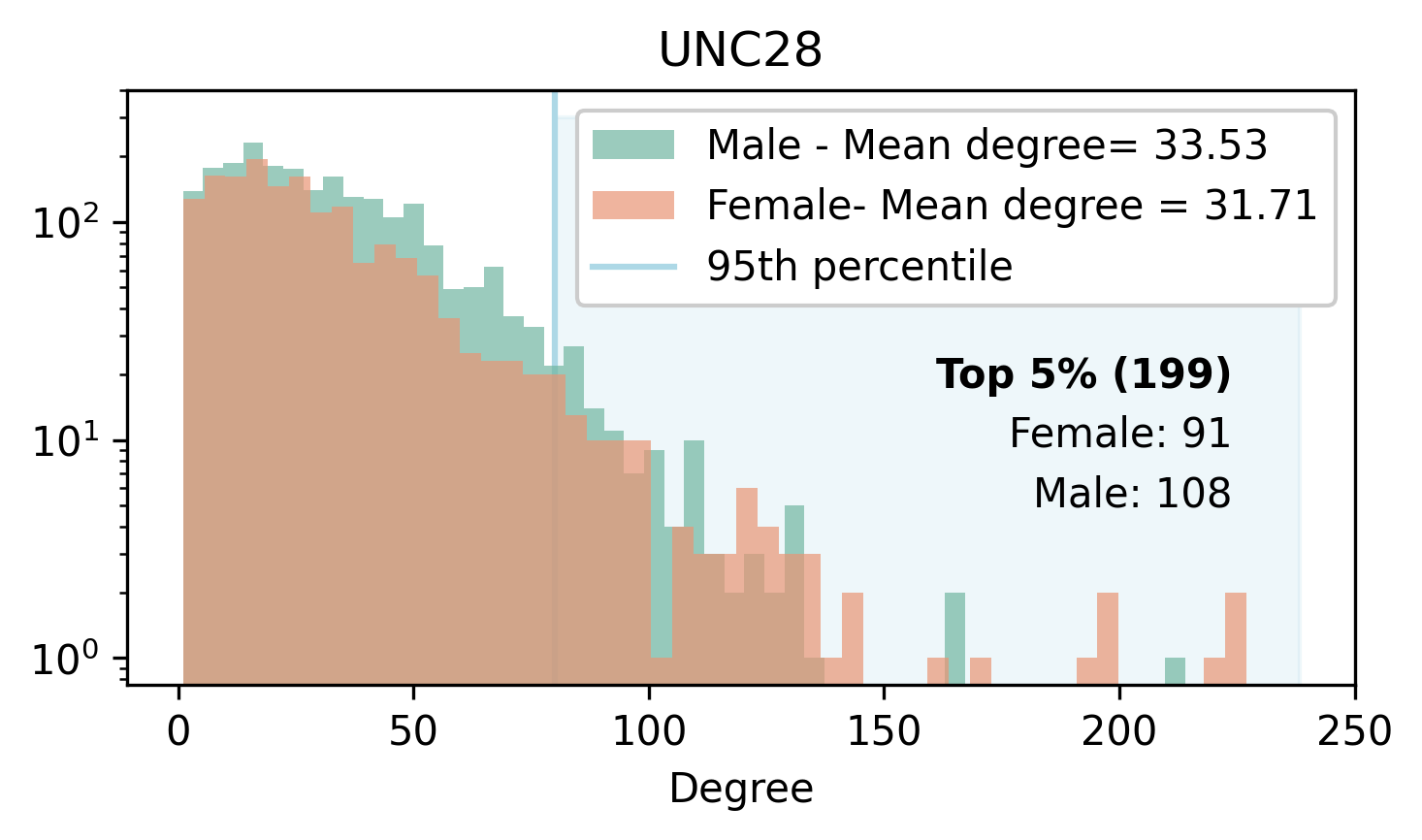}
    \end{subfigure}
    \hfill
    \begin{subfigure}{0.33\textwidth}
    \centering
    \includegraphics[width=\linewidth]{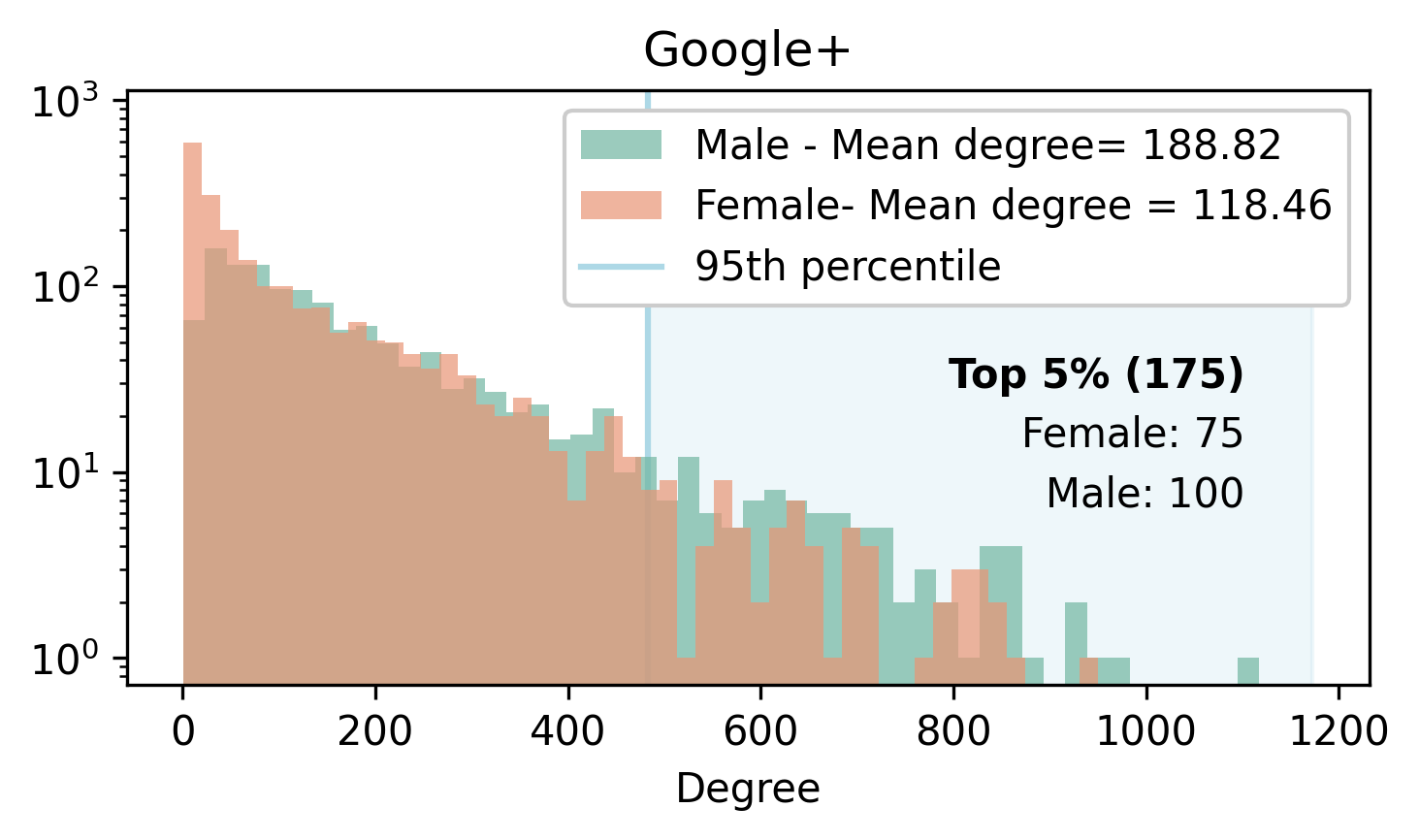}
    \end{subfigure}    
    \caption{Original degree distribution per dataset and gender. We also include the average degree per group and the composition of the top 5\% of nodes.}
    \label{fig:data-degree-dist}
\end{figure*}
\begin{figure*}[ht]
    \begin{subfigure}{0.33\textwidth}
    \centering
    \includegraphics[width=\linewidth]{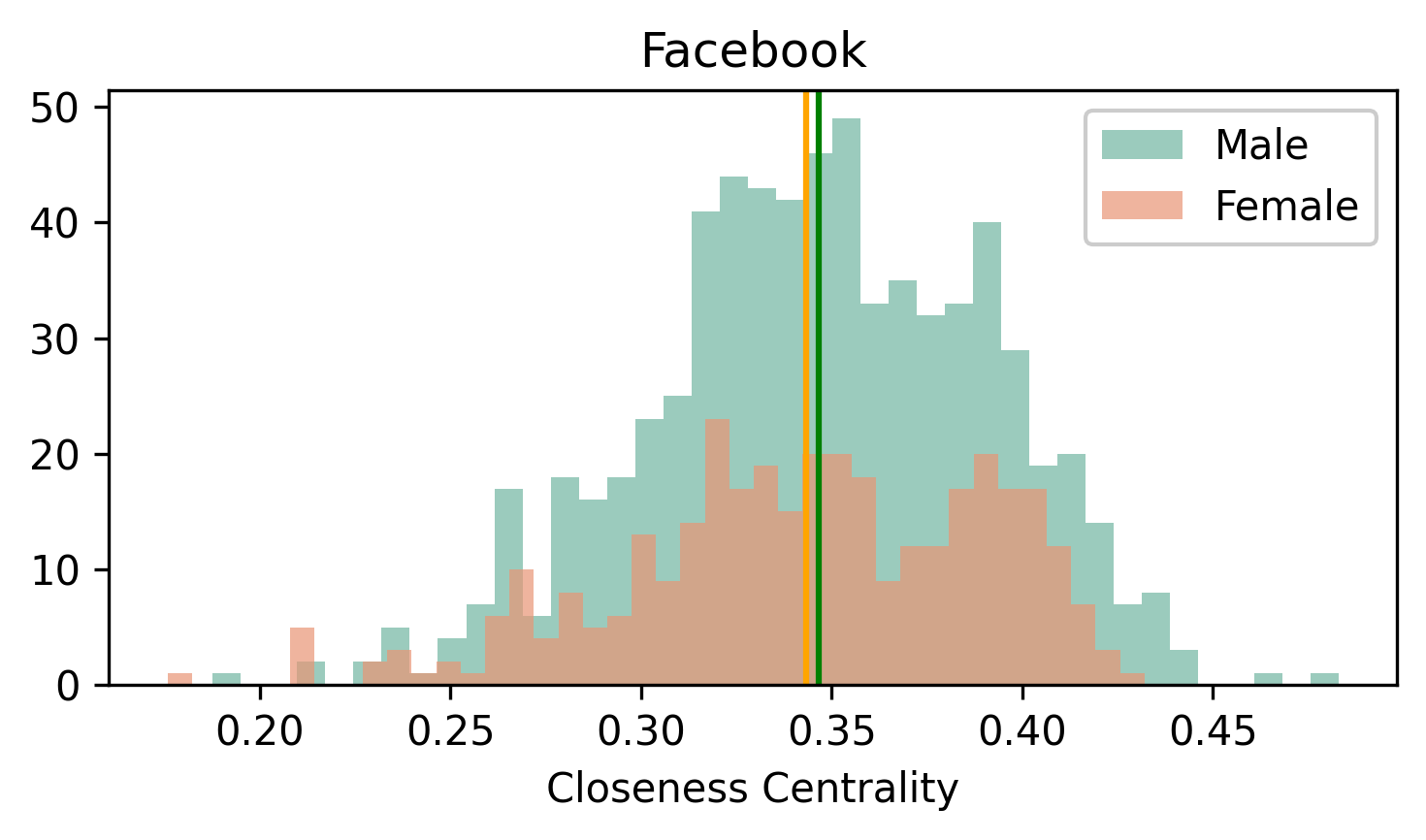}
    \end{subfigure}
    \hfill
    \begin{subfigure}{0.33\textwidth}
    \centering
    \includegraphics[width=\linewidth]{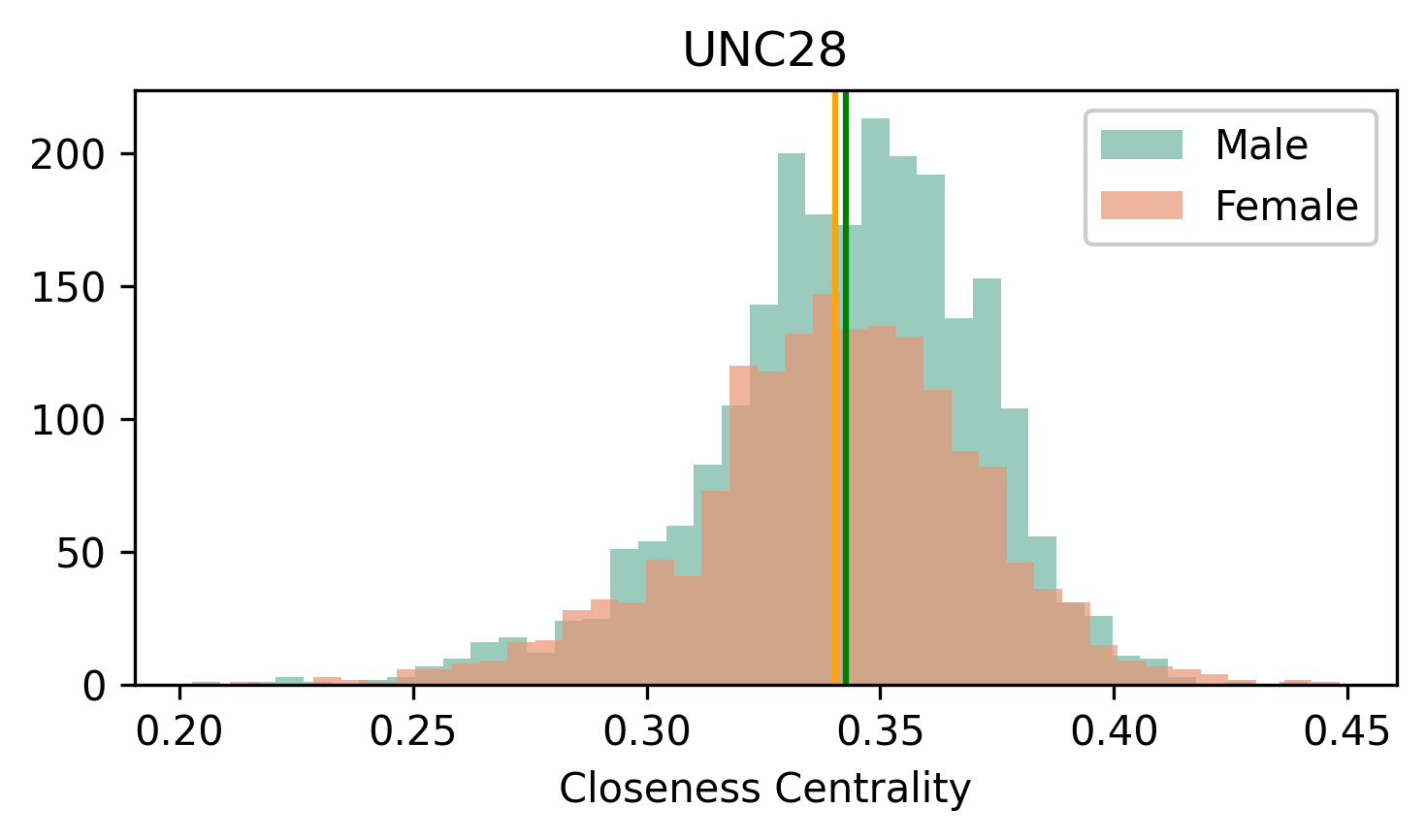}
    \end{subfigure}
    \hfill
    \begin{subfigure}{0.33\textwidth}
    \centering
    \includegraphics[width=\linewidth]{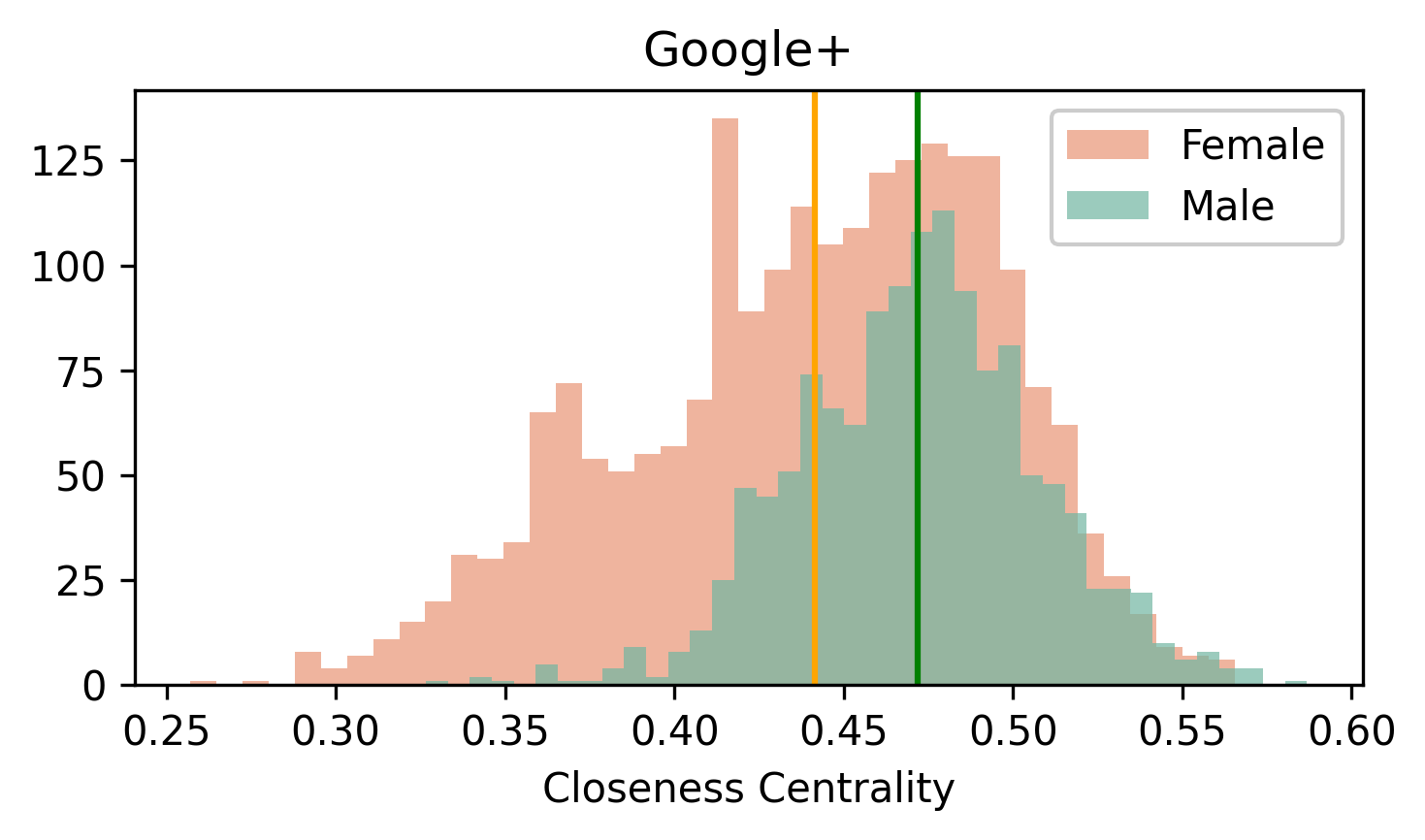}
    \end{subfigure}

    \begin{subfigure}{0.33\textwidth}
    \centering
    \includegraphics[width=\linewidth]{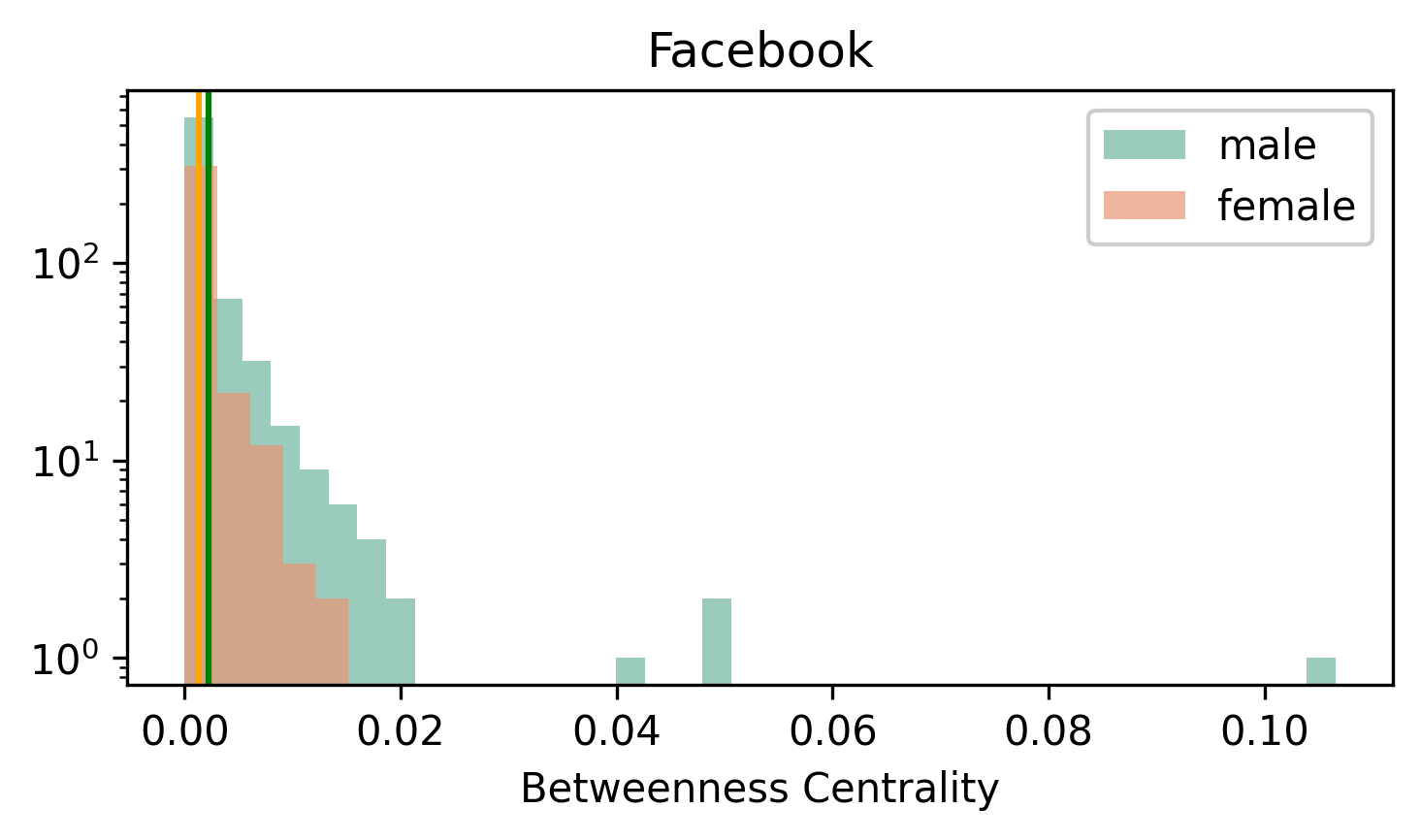}
    \end{subfigure}
    \hfill
    \begin{subfigure}{0.33\textwidth}
    \centering
    \includegraphics[width=\linewidth]{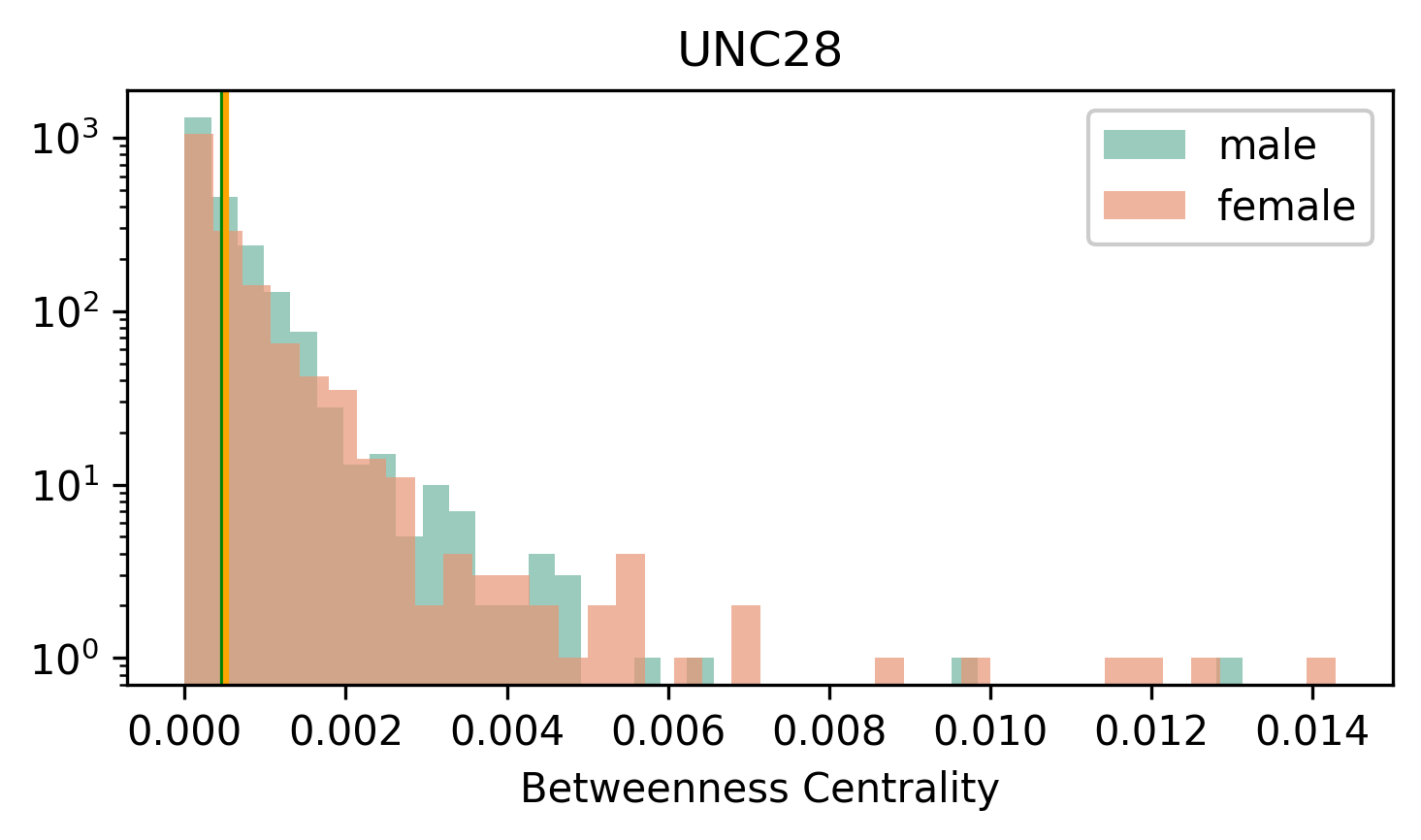}
    \end{subfigure}
    \hfill
    \begin{subfigure}{0.33\textwidth}
    \centering
    \includegraphics[width=\linewidth]{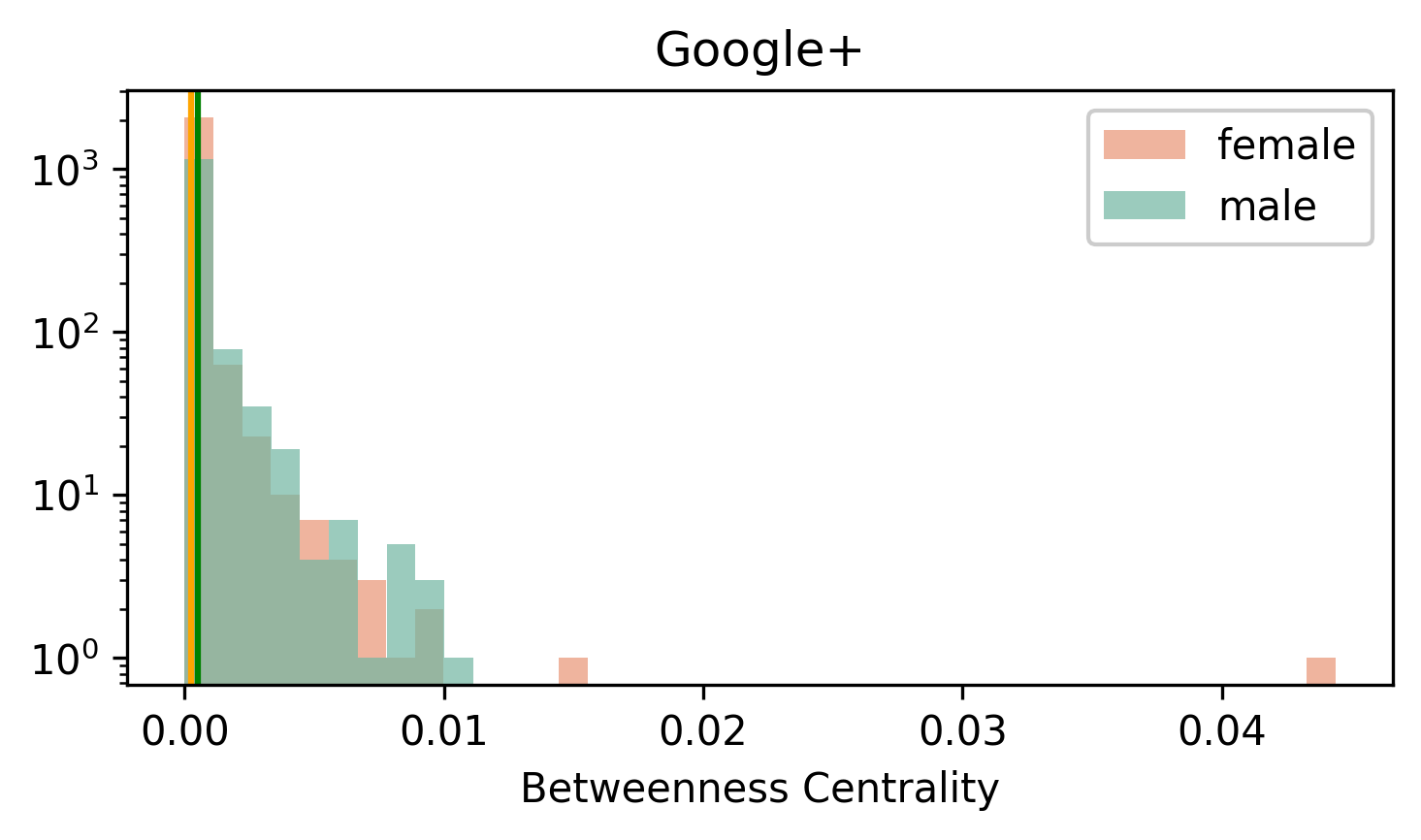}
    \end{subfigure}
    \caption{Per-group centrality metrics in the original datasets.}
    \label{fig:data-centrality-metrics}
\end{figure*}

 \begin{table}[ht]
    \centering
    \begin{tabular}{r|ccc}
    \toprule
    Dataset          & Facebook & UNC28 & Google+     \\
    \midrule
    $|V|$       & 1,034  & 3,985  & 3,508      \\
    $|\mathcal{E}|$ & 26,749 & 65,287 & 253,930  \\
    Density & 0.05 & 0.008 & 0.04  \\
    \midrule
    \# Males &  685 & 2307 & 1312 \\
    \# Females & 349 & 1678 & 2196 \\
    $h_{edge}$ & 0.58 & 0.55 & 0.51 \\
    \bottomrule
    \end{tabular}
    \caption{Dataset Statistics}
    \label{tab:data-stats}
\end{table}

\new{\paragraph{Number of male and female nodes} Table \ref{tab:data-stats} shows a breakdown of the number of nodes, edges, graph density and number of males and females in each of the studied datasets. The Facebook and UNC28 datasets have a majority of males whereas the Google+ dataset has a majority of females. Interestingly, the Google+ dataset has the highest levels of disparity in group social capital (see Table \ref{tab:expfairmetrics}) which means that, although there are fewer males in the network, they are better positioned than females with respect to their access to information and thus social capital.}

\new{\paragraph{Edge homophily} We also report the edge homophily~\citep{zhu2020homophily} for the three datasets. Edge homophily, denoted as $h_{edge} \in [0,1]$, represents the percentage of edges connecting nodes with the same protected attribute. The larger the group homophily, the more internally connected the groups are, with fewer connections between different groups. The obtained values of edge homophily suggest that the groups are well interconnected, and thus, there is no significant community polarization between males and females. In consequence, the differences between groups in access and control of information flow, as shown in Table~\ref{tab:expfairmetrics}, stem from disparities in roles and network positions, rather than from a polarized social structure.}

\new{\paragraph{Degree Distribution} Figure \ref{fig:data-degree-dist} shows the degree distribution for the three different datasets grouped by protected attribute. As seen in the Figure, the number of high degree nodes is larger for males than females in the Facebook and Google+ datasets. Note that although there are fewer males than females in the Google+ dataset, the majority of nodes with high degree correspond to males, which explains their larger group social capital.}

\new{\paragraph{Centrality Metrics} Figure \ref{fig:data-centrality-metrics} depicts the per-group closeness and betweenness centrality metrics~\citep{barabasi2013network} for the three datasets, providing an insight into the structural properties of the networks. However these classical metrics are based on geodesic distances and are therefore limited in their ability to capture the entire network topology. For example, the geodesic distance is a poor estimator of long-range connections. In contrast, the proposed metrics, which are built on random walks and consider the entire network, offer a more complete view of information flow within the network (see \ref{tab:expfairmetrics}). Nevertheless, even with classical metrics, we observe notable differences between groups, particularly in the Google+ dataset, where males exhibit larger closeness centrality despite being a minority.}

\section{Additional Experiments}

All the experiments were done using a workstation with 16GB RAM; and processor 11th Gen Intel(R) Core(TM) i7, 3.00GHz, 2995 Mhz, 4 Core(s), 8 Logical Processor(s).

\subsection{Group Social Capital Metrics}
\label{app:exp:gorupmetrics}

\begin{table*}[t]
\centering
\begin{subtable}[ht]{0.32\textwidth}
\centering
\caption{Facebook (50) Female}
\begin{small}
\begin{tabular}{lrrr}
\toprule
 & $R_{tot} \downarrow$ & $\mathcal{R}_{diam} \downarrow$ & $\mathsf{B_R} $ \\
$G$  & 221.4 & 2.29 & 1.927 \\
\midrule
Random & 211.5 & 2.26 & 1.927 \\
SDRF & 220.7 & 2.28 & 1.928 \\
FOSR & 202.1 & 2.15 & 1.926 \\
DW & 199.4 & 1.87 & 1.929 \\
Cos & 190.4 & 1.64 & 1.918 \\
ERG & \textbf{138.7} & \textbf{0.43} & \textbf{1.933} \\
\bottomrule
\end{tabular}
\end{small}
\end{subtable}
\hfill
\begin{subtable}[ht]{0.32\textwidth}
\centering
\caption{UNC28 (5,000) Female}
\begin{small}
\begin{tabular}{lrrr}
\toprule
 & $R_{tot} \downarrow$ & $\mathcal{R}_{diam} \downarrow$ & $\mathsf{B_R}$ \\
$G$  & 608.6 & 2.11 & 1.994 \\
\midrule
Random & 435.0 & 1.23 & 1.992 \\
SDRF & 599.1 & 2.11 & 1.996 \\
FOSR & 509.9 & 1.33 & 1.990 \\
DW & 583.2 & 1.81 & 1.997 \\
Cos & 429.6 & 0.42 & 1.940 \\
ERG & \textbf{316.8} & \textbf{0.10} & \textbf{2.001} \\
\bottomrule
\end{tabular}
\end{small}
\end{subtable}
\hfill
\begin{subtable}[ht]{0.32\textwidth}
\centering
\caption{Google+ (5,000) Female}
\begin{small}
\begin{tabular}{lrrr}
\toprule
 & $R_{tot} \downarrow$ & $\mathcal{R}_{diam} \downarrow$ & $\mathsf{B_R}$ \\
$G$ & 564.1 & 1.31 & 1.807 \\
\midrule
Random & 305.1 & 1.07 & 1.824 \\
SDRF & 561.1 & 1.31 & 1.805 \\
FOSR & 498.9 & 1.25 & 1.811 \\
DW & 558.7 & 1.31 & 1.806 \\
Cos & 230.9 & 0.29 & 1.827 \\
ERG & \textbf{145.5} & \textbf{0.07} & \textbf{1.889} \\
\bottomrule
\end{tabular}
\end{small}
\end{subtable}
\vskip 0.1in
\begin{subtable}[ht]{0.32\textwidth}
\centering
\caption{Facebook (50) Male}
\begin{small}
\begin{tabular}{lrrr}
\toprule
 & $R_{tot} \downarrow$ & $\mathcal{R}_{diam} \downarrow$ & $\mathsf{B_R}$ \\
$G$ & 179.8 & 2.25 & 2.034 \\
\midrule
Random & 172.8 & 2.22 & 2.035 \\
SDRF & 179.1 & 2.24 & 2.034 \\
FOSR & 167.6 & 2.13 & 2.035\\
DW & 163.1 & 1.83 & 2.033 \\
Cos & 161.7 & 1.61 & 2.039 \\
ERG & \textbf{128.5} & \textbf{0.41} & \textbf{2.031} \\
\bottomrule
\end{tabular}
\end{small}
\end{subtable}
\hfill
\begin{subtable}[ht]{0.32\textwidth}
\centering
\caption{UNC28 (5,000) Male}
\begin{small}
\begin{tabular}{lrrr}
\toprule
 & $R_{tot} \downarrow$ & $\mathcal{R}_{diam} \downarrow$ & $\mathsf{B_R}$ \\
$G$ & 586.3 & 2.11 & 2.004 \\
\midrule
Random & 415.2 & 1.23 & 2.005 \\
SDRF & 576.9 & 2.10 & 2.002 \\
FOSR & 490.2 & 1.33 & 2.007 \\
DW & 561.0 & 1.81 & 2.001 \\
Cos & 410.6 & 0.41 & 2.043 \\
ERG & \textbf{308.0} & \textbf{0.09} & \textbf{1.998} \\
\bottomrule
\end{tabular}
\end{small}
\end{subtable}
\hfill
\begin{subtable}[ht]{0.3\textwidth}
\centering
\caption{Google+ (5,000) Male}
\begin{small}
\begin{tabular}{lrrr}
\toprule
 & $R_{tot} \downarrow$ & $\mathcal{R}_{diam} \downarrow$ & $\mathsf{B_R}$ \\
$G$ & 287.7 & 1.24 & 2.321 \\
\midrule
Random & 175.7 & 1.03 & 2.293 \\
SDRF & 285.2 & 1.23  & 2.325 \\
FOSR & 258.1 & 1.17 & 2.315 \\
DW & 284.6 & 1.24 & 2.323 \\
Cos & 144.1 & 0.27 & 2.288 \\
ERG & \textbf{108.5} & \textbf{0.06} & \textbf{2.185}\\
\bottomrule
\end{tabular}
\end{small}
\end{subtable}
\caption{Group Social Capital after Edge Augmentation. The best values are highlighted in bold. Note how edge augmentation via \ERG{}{} is able to not only increase the social capital of the disadvantaged group (females) but also of the rest of the groups (males).}
\label{tab:groupimprovement}
\end{table*}

\begin{figure*}[ht]
    \centering
    \includegraphics[width=\linewidth]{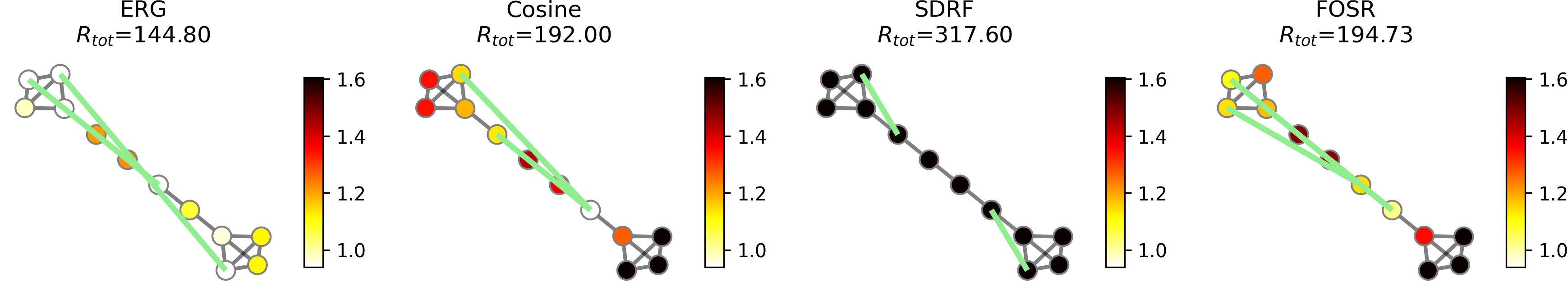}
    \caption{Example of the added links by the different algorithms on a Barbell synthetic graph. Nodes are colored by $\Rg(u)$. More examples in Fig~\ref{fig:rew2}.}
    \label{fig:rew}
\end{figure*}

The experiments presented in the main paper illustrate how edge augmentation via \ERG{} is able to (1) significantly mitigate the structural group unfairness (Table~\ref{tab:expfairmetrics}); and (2) increase the social capital for all the groups in the graph (Figure~\ref{fig:paretomain}). Table~\ref{tab:groupimprovement} depicts the group social capital for each group after each intervention. Note that the optimal value of $\Bt$ is to get every group's control equal to $2-2/|\mathcal{V}|$. %

\begin{figure*}[ht]
    \centering
     \begin{subfigure}{0.24\textwidth}
    \centering
    \includegraphics[width=\linewidth]{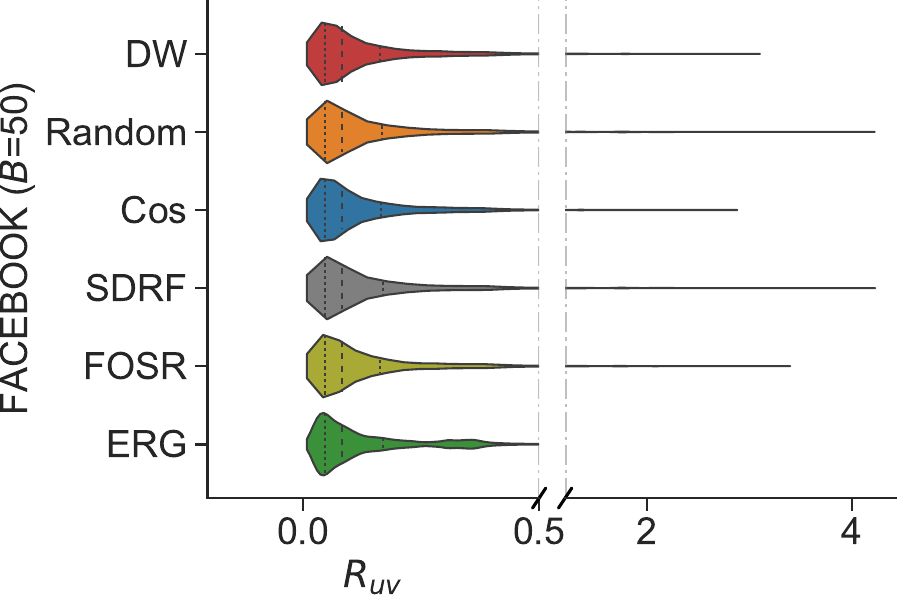}
    \end{subfigure}
    \hfill
    \centering
     \begin{subfigure}{0.24\textwidth}
    \centering
    \includegraphics[width=\linewidth]{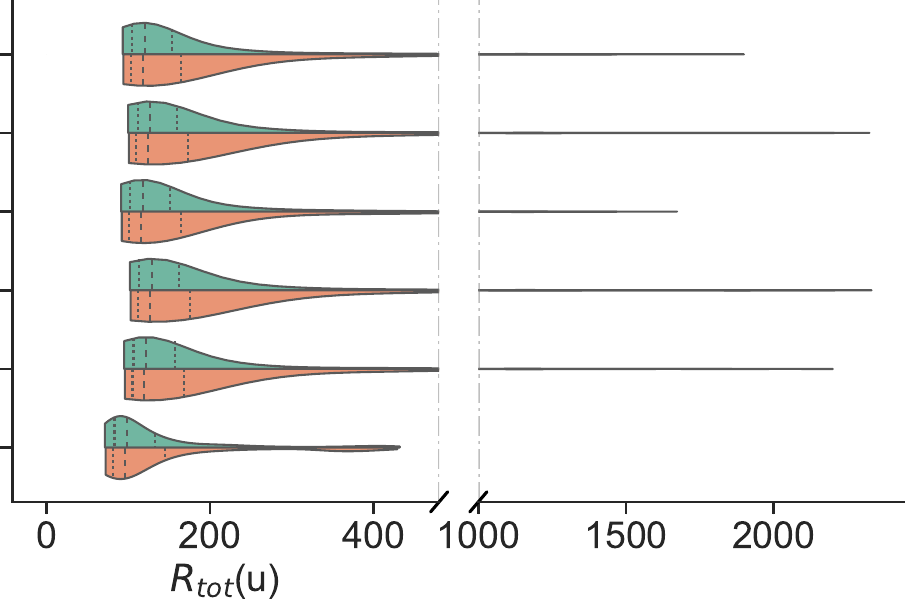}
    \end{subfigure}
    \hfill
    \centering
     \begin{subfigure}{0.24\textwidth}
    \centering
    \includegraphics[width=\linewidth]{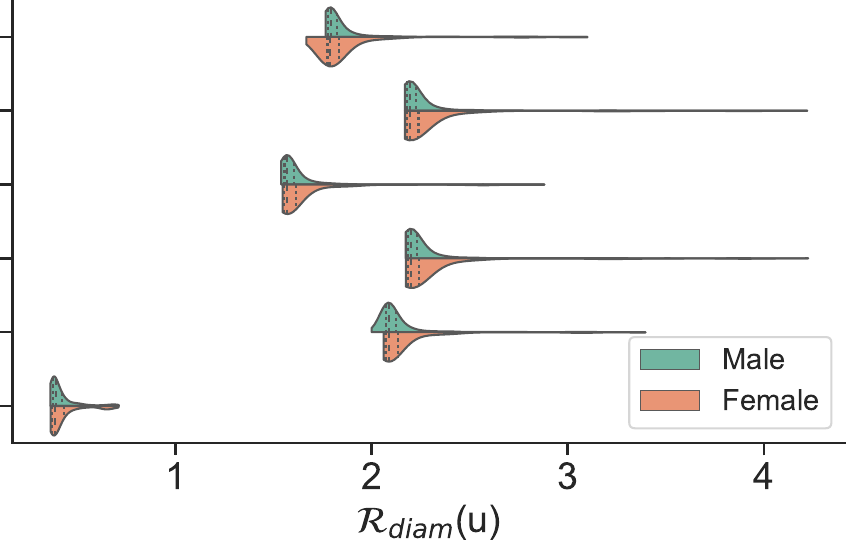}
    \end{subfigure}
    \hfill
    \centering
     \begin{subfigure}{0.24\textwidth}
    \centering
    \includegraphics[width=\linewidth]{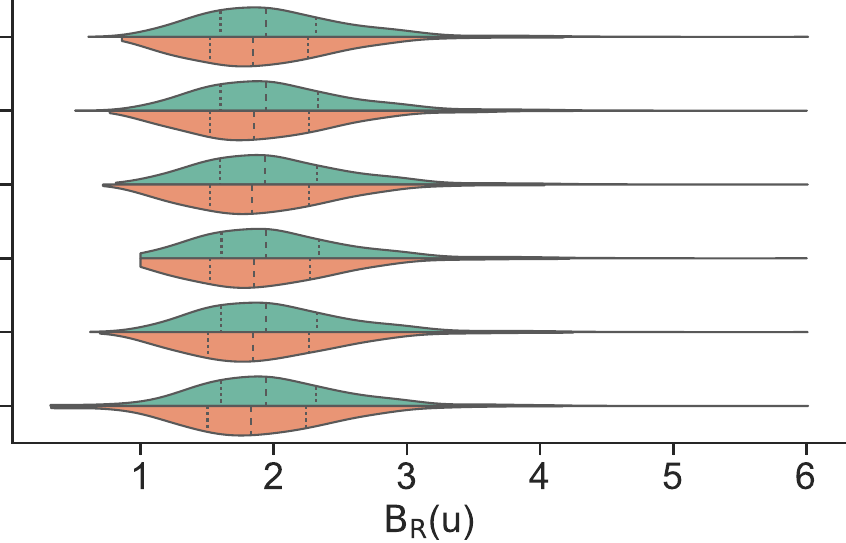}
    \end{subfigure}

     \begin{subfigure}{0.24\textwidth}
    \centering
    \includegraphics[width=\linewidth]{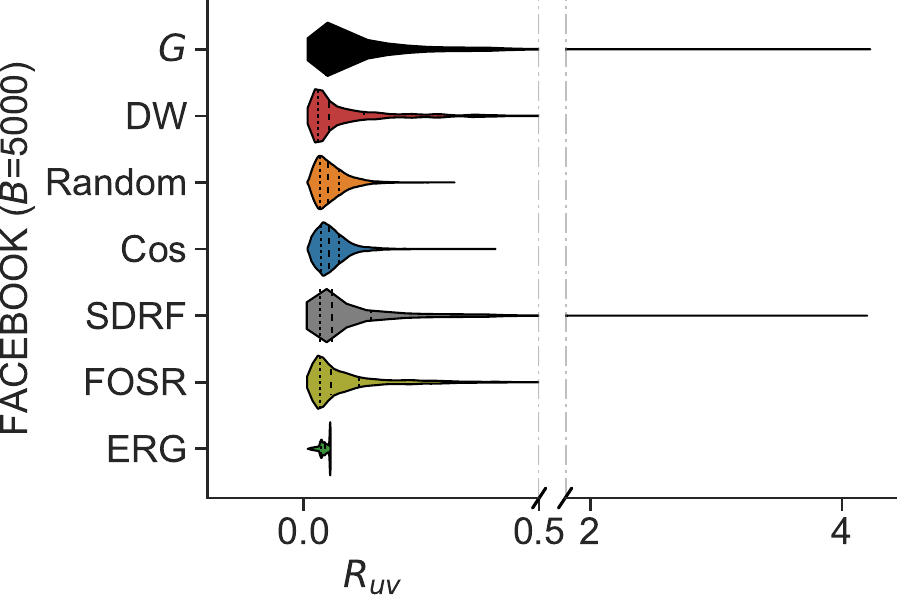}
    \end{subfigure}
    \hfill
    \centering
     \begin{subfigure}{0.24\textwidth}
    \centering
    \includegraphics[width=\linewidth]{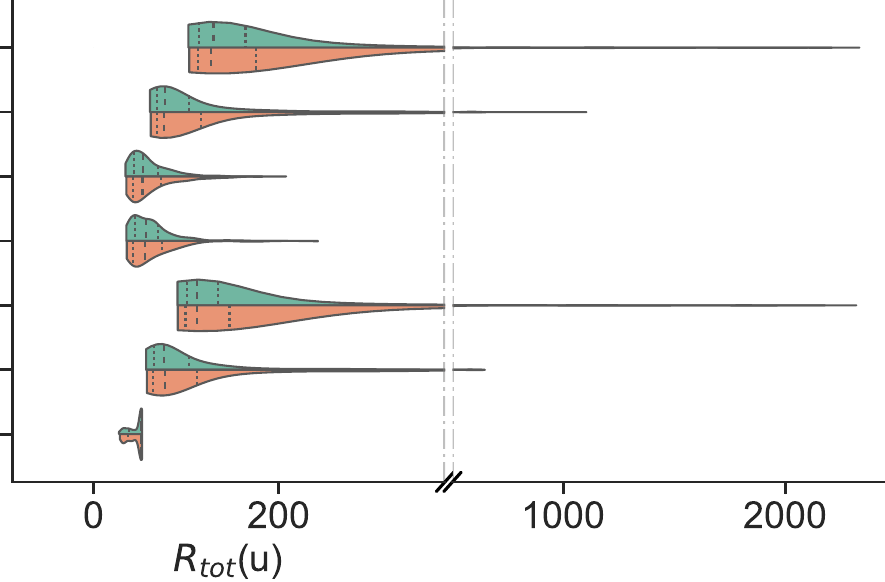}
    \end{subfigure}
    \hfill
    \centering
     \begin{subfigure}{0.24\textwidth}
    \centering
    \includegraphics[width=\linewidth]{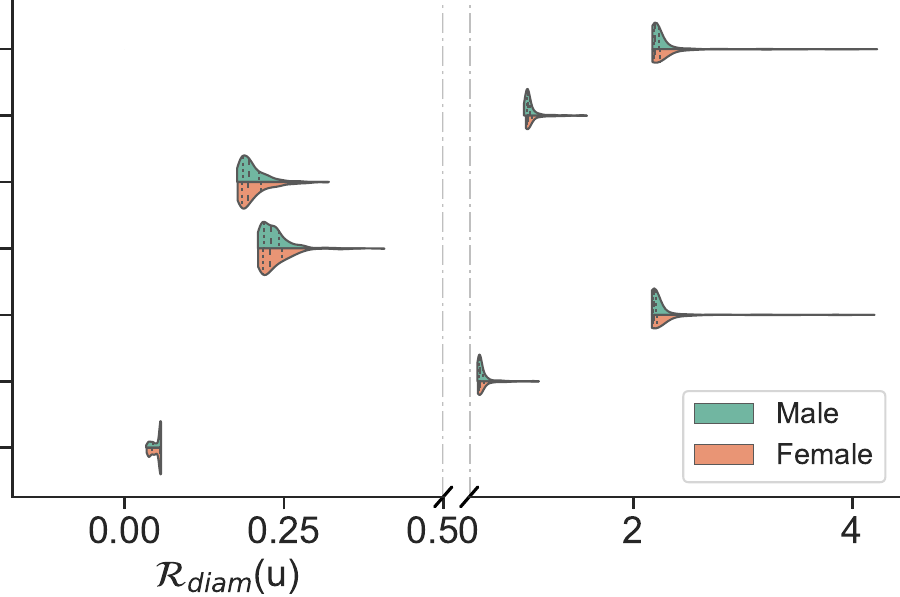}
    \end{subfigure}
    \hfill
    \centering
     \begin{subfigure}{0.24\textwidth}
    \centering
    \includegraphics[width=\linewidth]{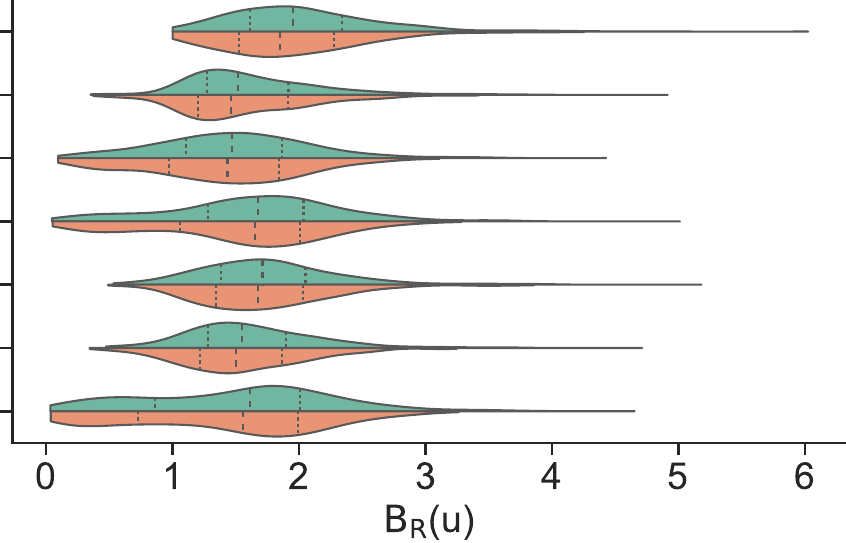}
    \end{subfigure}
    \hfill
    \caption{Distribution of $R_{uv}$ and proposed social capital metrics on the Facebook dataset and after different graph interventions, 50 links in the top row and 5,000 in the bottom row. Columns show (left) distribution of all $R_uv$ distances, (center-left) distribution of all $\Rg(u)$, (center-right) distribution of all $\Rd(u)$, (right) distribution of all $\Bt(u)$. The distributions for the node metrics are shown for the two groups according to the protected attribute (gender).}
    \label{fig:violinfb5000}
\end{figure*}

\subsection{Distribution of Social Capital by Group}

For completeness, Figure~\ref{fig:violinfb5000} illustrates the distributions of all effective resistances $R_{uv}$ and the node's social capital metrics --$\Rg(u)$, $\Rd(u)$ and $\Bt(u)$-- in the graph before and after the edge augmentation interventions for each of the groups.

The plots correspond to two different edge augmentation experiments on the Facebook dataset, with budgets of B=50 and B=5,000 new edges. Edge augmentation via \ERG{} is able to drastically reduce all effective resistances of the graph, unlike the other methods and even for the small budget.
Note how there is still a long tail in the distribution of effective resistances after edge augmentation with the baseline methods, which illustrates that there are still nodes that struggle to exchange information even after the intervention.

Regarding the node social capital metrics, edge augmentation via \ERG{} reduces all $\Rg(u)$ and $\Rd(u)$, which explains why $\Rg(S_i)$ and $\Rg(S_i)$ is improved for all groups in the graph, and particularly for the disadvantaged group. Thus, $\Delta\Rg$ and $\Delta\Rd$ are significantly reduced. %

\subsection{Evolution of Group Social Capital During the Interventions}
\label{app:evolutionfigs}
We provide additional results regarding the evolution of group social capital metrics during the graph's intervention, as initially shown in Figure~\ref{fig:googleresevolution}. The goal is to analyze the effectiveness of the edge augmentation methods on both small ($B = 50$ edges) and large ($B= 5,000$ edges) budget scenarios. Figures~\ref{fig:fb50evolution}and~\ref{fig:fb5000evolution} show the evolution of both the group social capital metrics for each group and the structural group unfairness metrics on the Facebook dataset.

We observe in Figure~\ref{fig:fb50evolution} how edge augmentation via \ERG{} is able to significantly improve the group social capital for all groups and reduce all structural unfairness metrics. In contrast, the baselines fail to do so. %

\begin{figure*}[ht]
    \begin{subfigure}{0.36\textwidth}
    \centering
    \includegraphics[width=\linewidth]{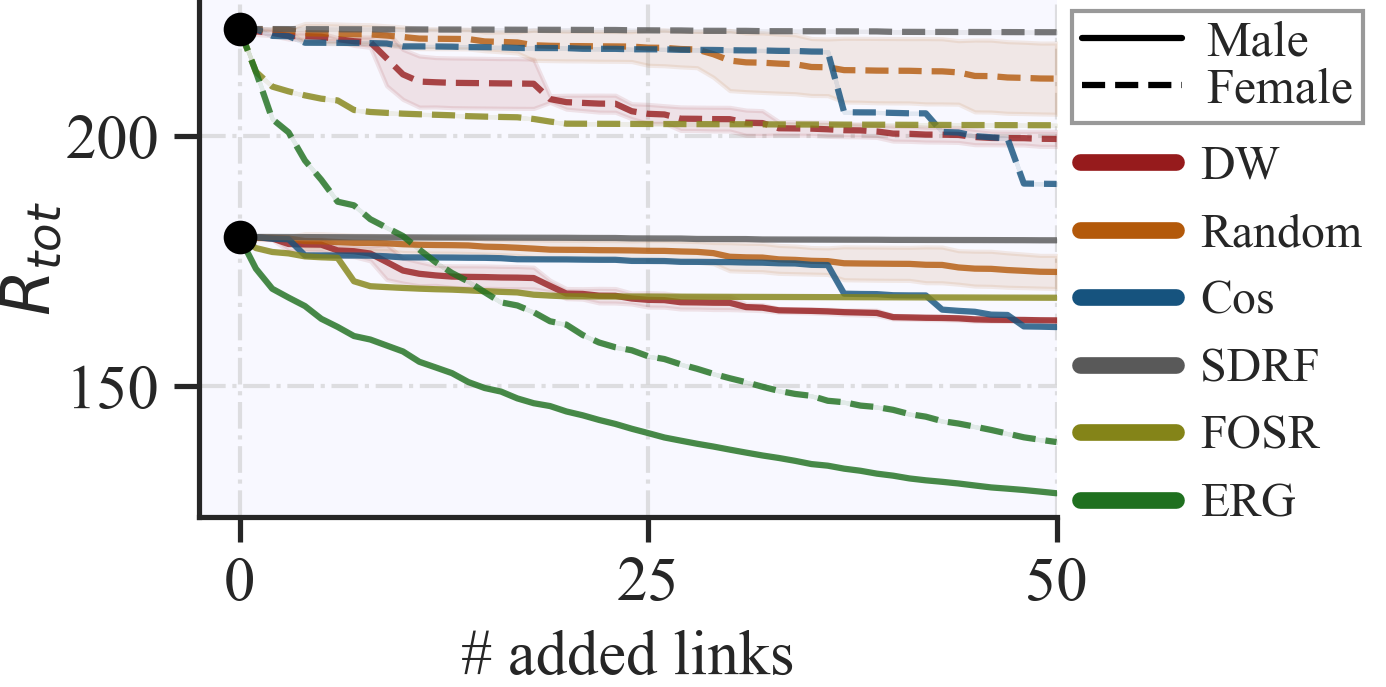}
    \end{subfigure}
    \hspace{.1in}
    \begin{subfigure}{0.27\textwidth}
    \centering
    \includegraphics[width=\linewidth]{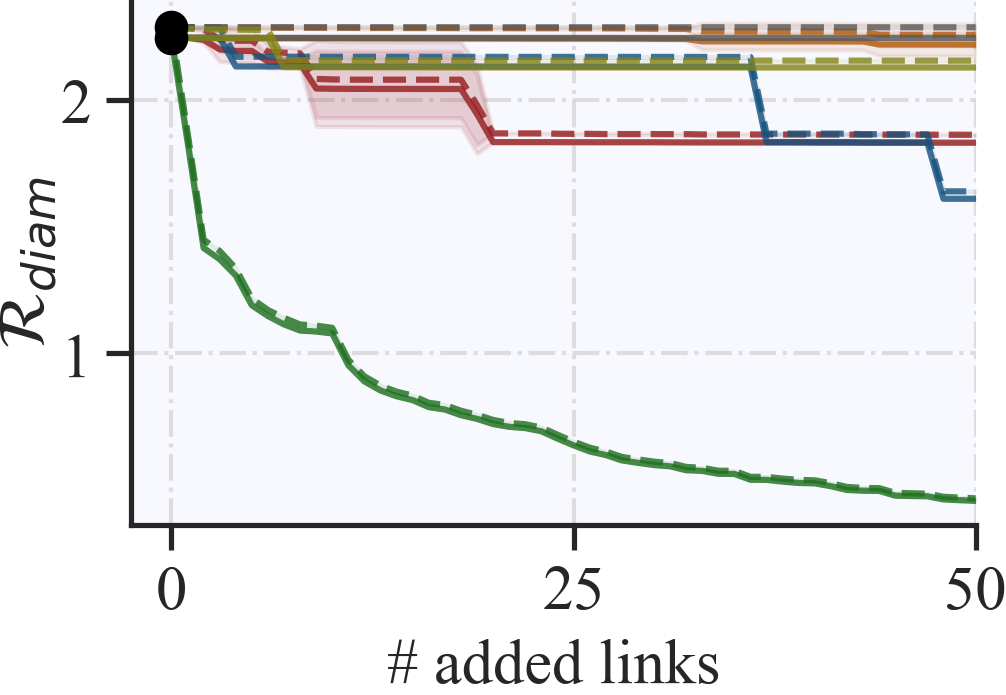}
    \end{subfigure}
    \hfill
    \begin{subfigure}{0.27\textwidth}
    \centering
    \includegraphics[width=\linewidth]{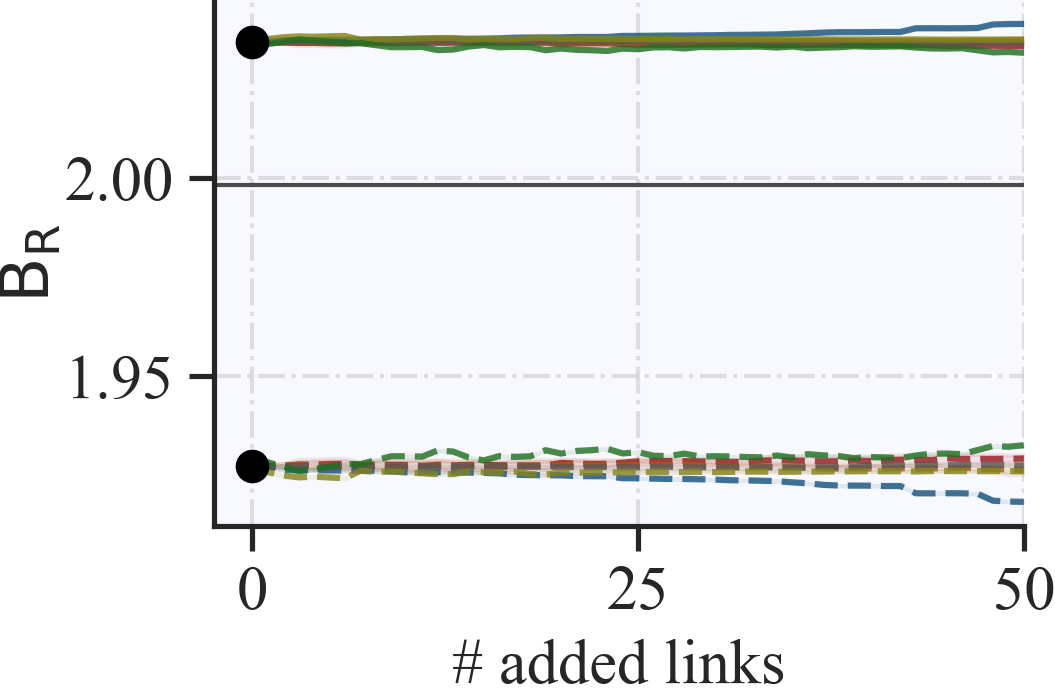}
    \end{subfigure}
    \begin{subfigure}{0.36\textwidth}
    \centering
    \includegraphics[width=\linewidth]{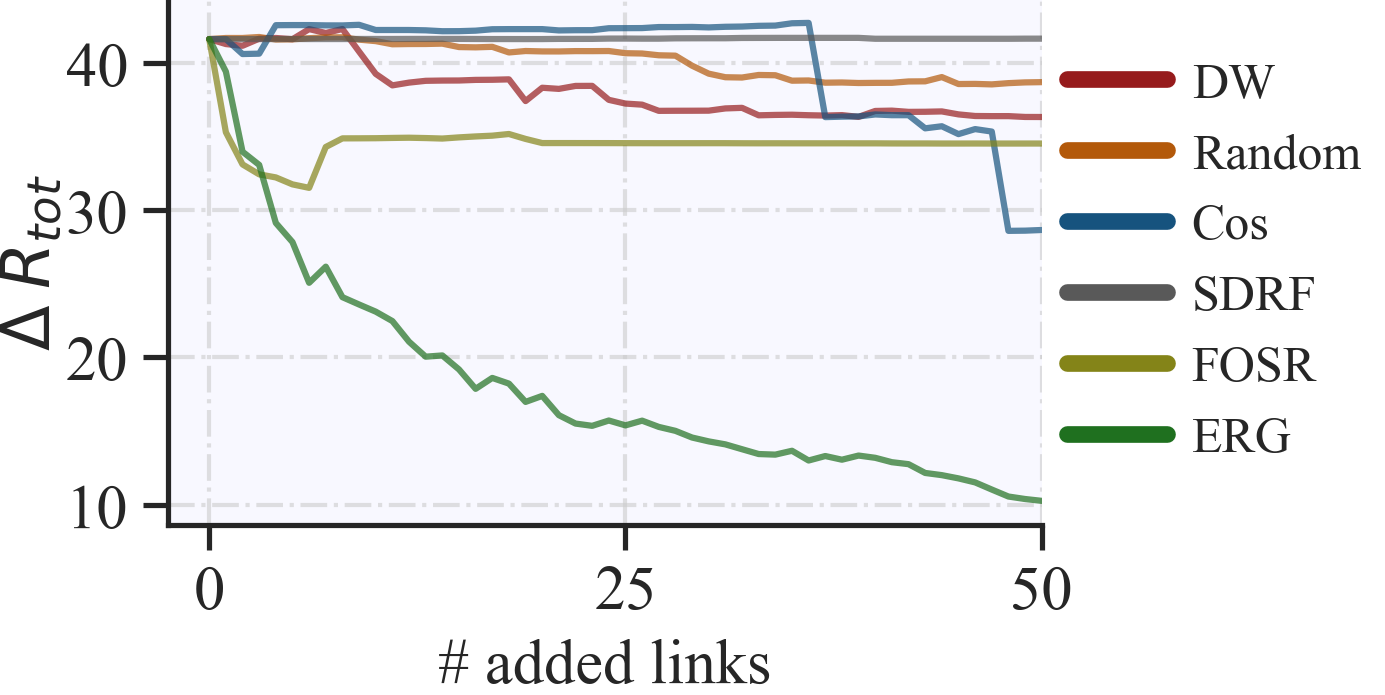}
    \end{subfigure}
    \hspace{.1in}
    \begin{subfigure}{0.27\textwidth}
    \centering
    \includegraphics[width=\linewidth]{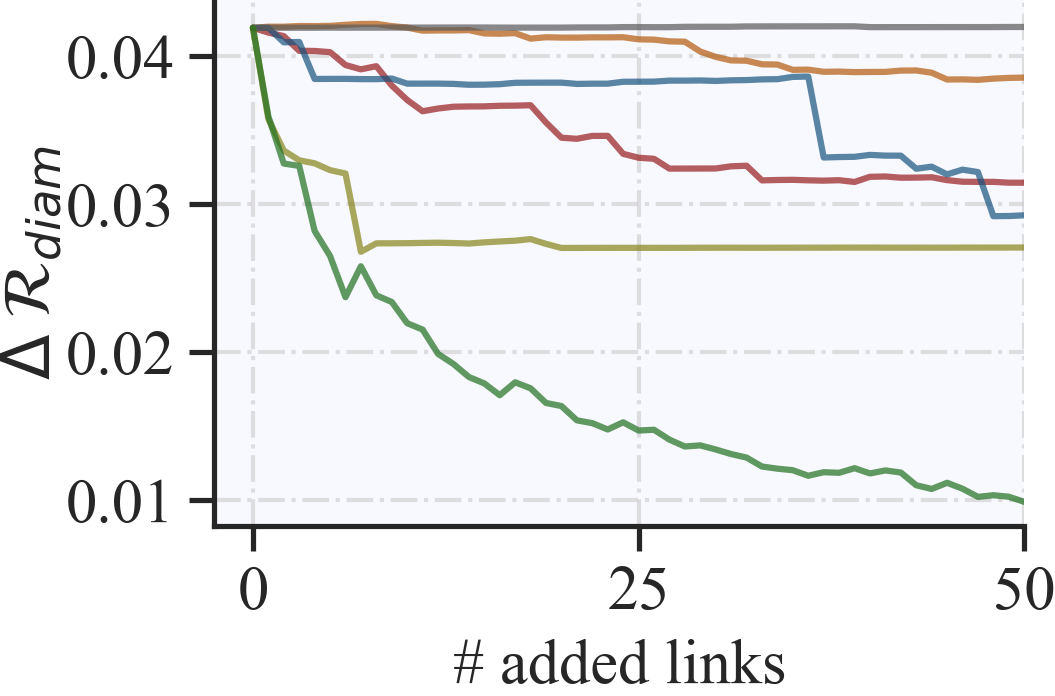}
    \end{subfigure}
    \hfill
    \begin{subfigure}{0.27\textwidth}
    \centering
    \includegraphics[width=\linewidth]{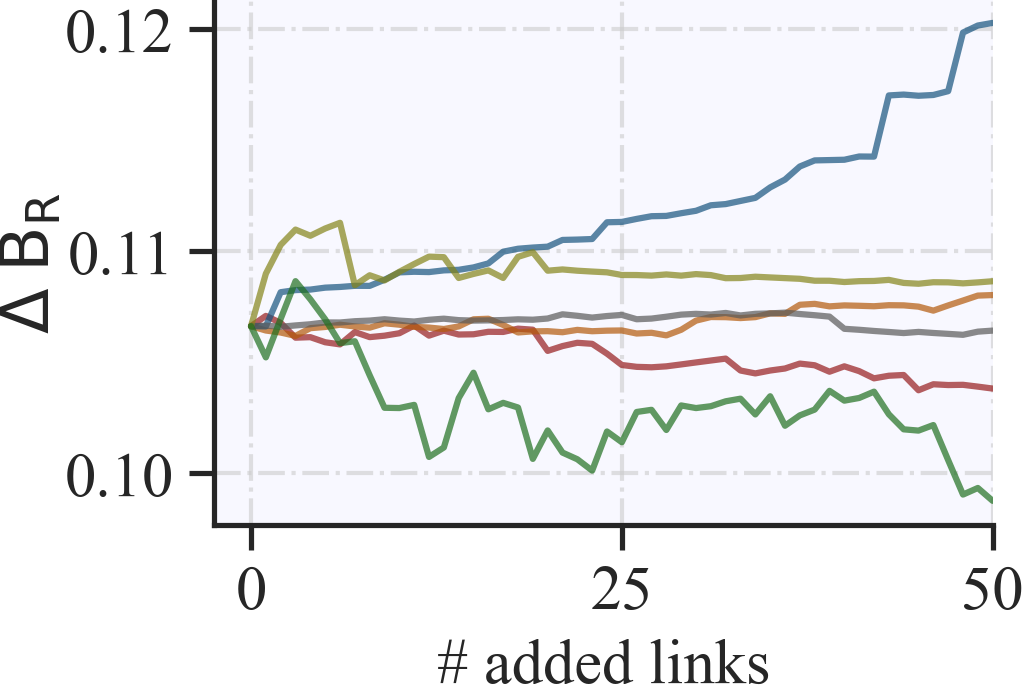}
    \end{subfigure}
    \caption{Evolution of group social capital and disparity metrics for both groups as the number of added edges increases on the Facebook dataset with a budget of 50 new edges.}
    \label{fig:fb50evolution}
\end{figure*}

Figure~\ref{fig:fb5000evolution} depicts the evolution of group social capital and disparity metrics when adding 5,000 edges to the Facebook dataset. The more edges we add to a graph, the denser the graph becomes and therefore the better the information flow. Thus, one can expect that after a large number of added edges, all methods behave similarly. 

However, we observe some differences. First,  the convergence to minimal isolation and diameter disparities is significantly faster when adding edges via \ERG{} than any other method. Second, the decrease in $\Rg(S_i)$ is very significant for both groups (males and females) even after just adding a small number of edges by means of \ERG{}. Third, regarding group control and control disparity, edge augmentation via \ERG{} systematically reduces $\Delta\Bt$ while converging each group control to the optimal $\Bt(S_i)=2-2/|\mathcal{V}|$.

We also show in Figure~\ref{fig:unc1000evolution} the evolution of $\Bt(S_i)$ and $\Delta\Bt$ on the UNC28 dataset with 1,000 edge additions along with the distribution of node's control $\Bt(u)$ after the intervention to showcase an scenario where \ERG{} is able to reach the optimal control disparity in the graph, allocating the same amount of control for each group, $\Bt(S_i)\approx 2-2/|\mathcal{V}| \:\forall\: i\in SA\rightarrow \Delta\Bt\approx 0$.

Last but not least, the different baselines do not reach a better behavior than the random method, neither on the group social capital metrics nor in terms of structural unfairness. After 5,000 edge additions, the random method significantly reduces the unfairness metrics $\Delta\Rg$ and $\Delta\Rd$ while also achieving decent rated of $\Rg(S_i)$ and $\Rd(S_i)$ since the budget is high enough to improve the information flow with no strategy. However, all baselines struggle to optimize $\Delta\Bt(S_i)$ and $\Bt(S_i)$. None of them is able to improve, and the cosine similarity approach even increased the Control Disparity.

\begin{figure*}[ht]
    \begin{subfigure}{0.36\textwidth}
    \centering
    \includegraphics[width=\linewidth]{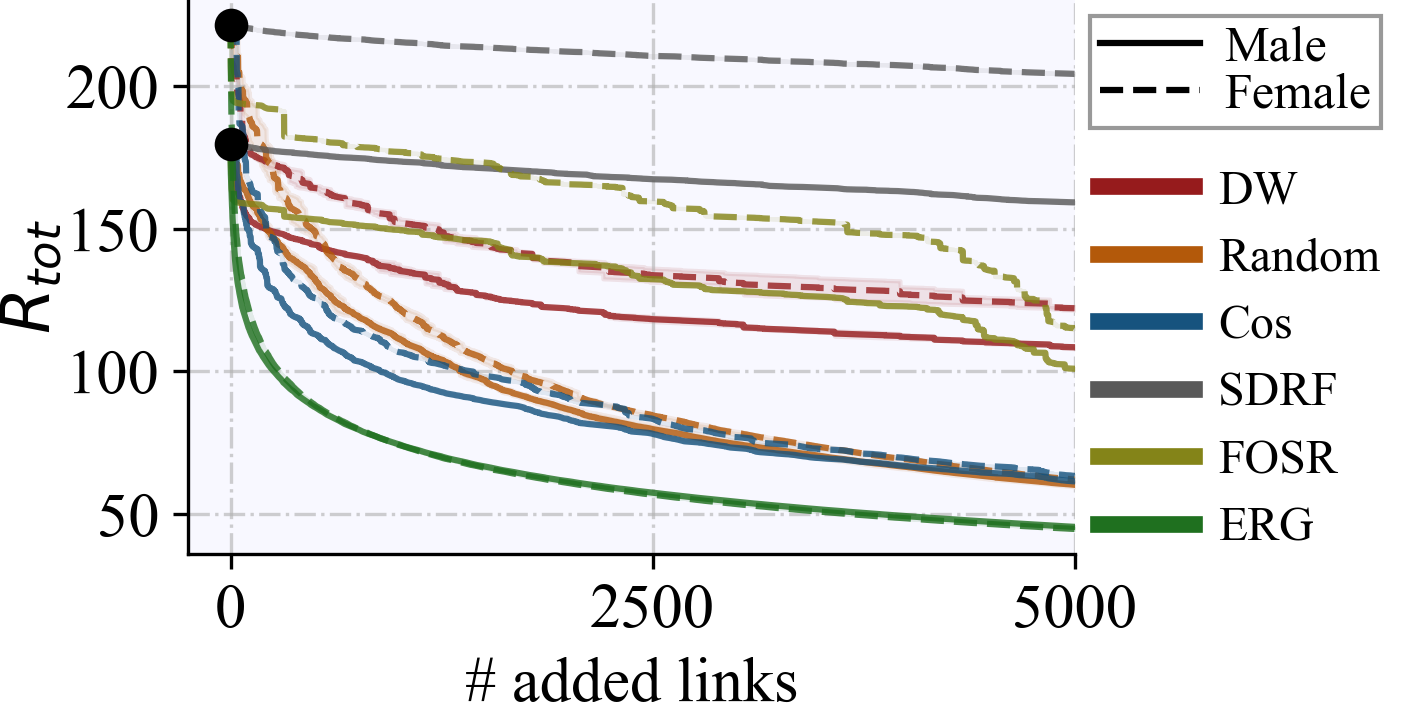}
    \end{subfigure}
    \hspace{.2in}
    \begin{subfigure}{0.27\textwidth}
    \centering
    \includegraphics[width=\linewidth]{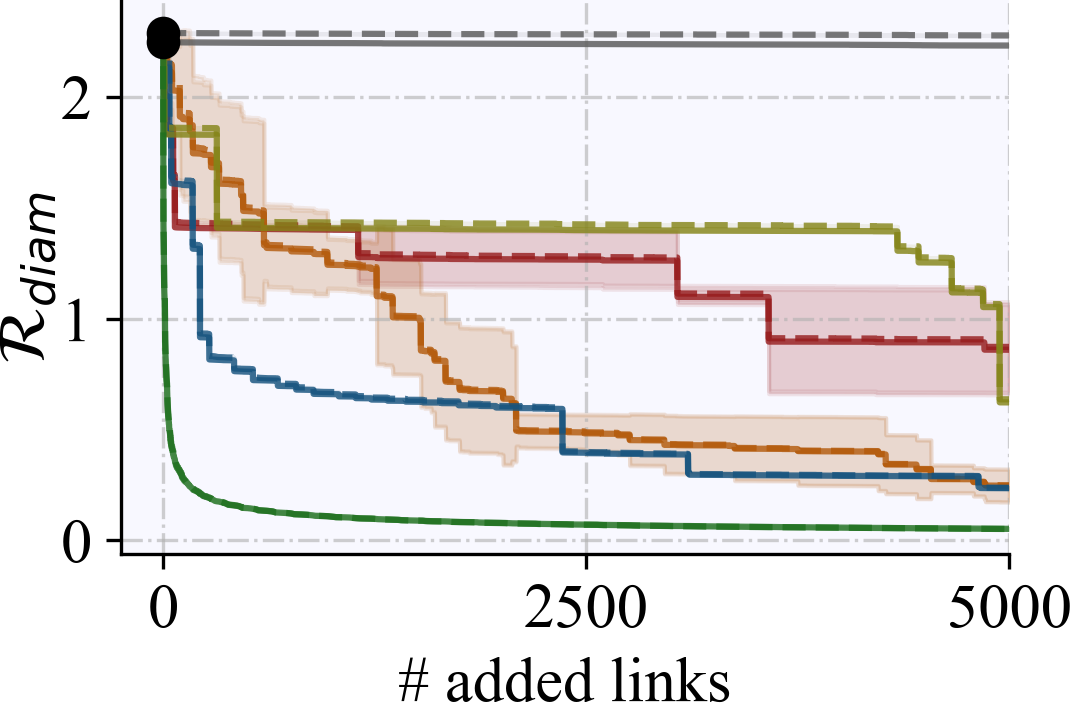}
    \end{subfigure}
    \hfill
    \begin{subfigure}{0.27\textwidth}
    \centering
    \includegraphics[width=\linewidth]{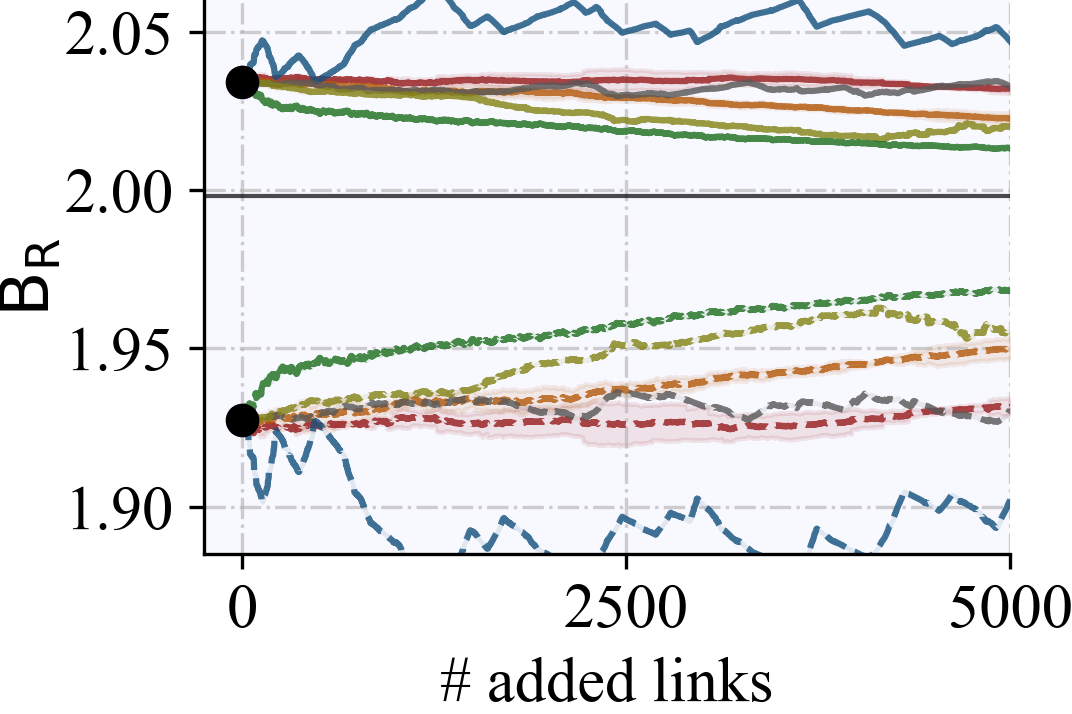}
    \end{subfigure}
    \begin{subfigure}{0.36\textwidth}
    \centering
    \includegraphics[width=\linewidth]{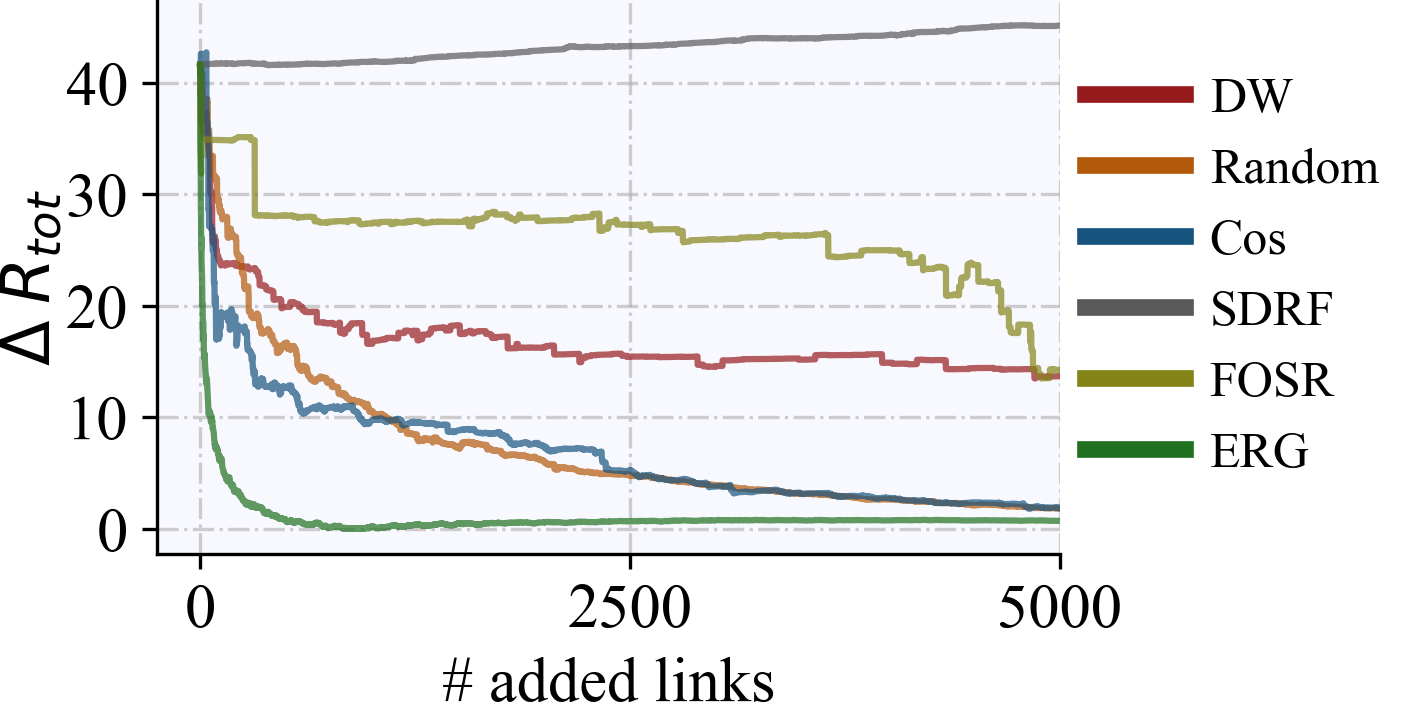}
    \end{subfigure}
    \hspace{.2in}
    \begin{subfigure}{0.27\textwidth}
    \centering
    \includegraphics[width=\linewidth]{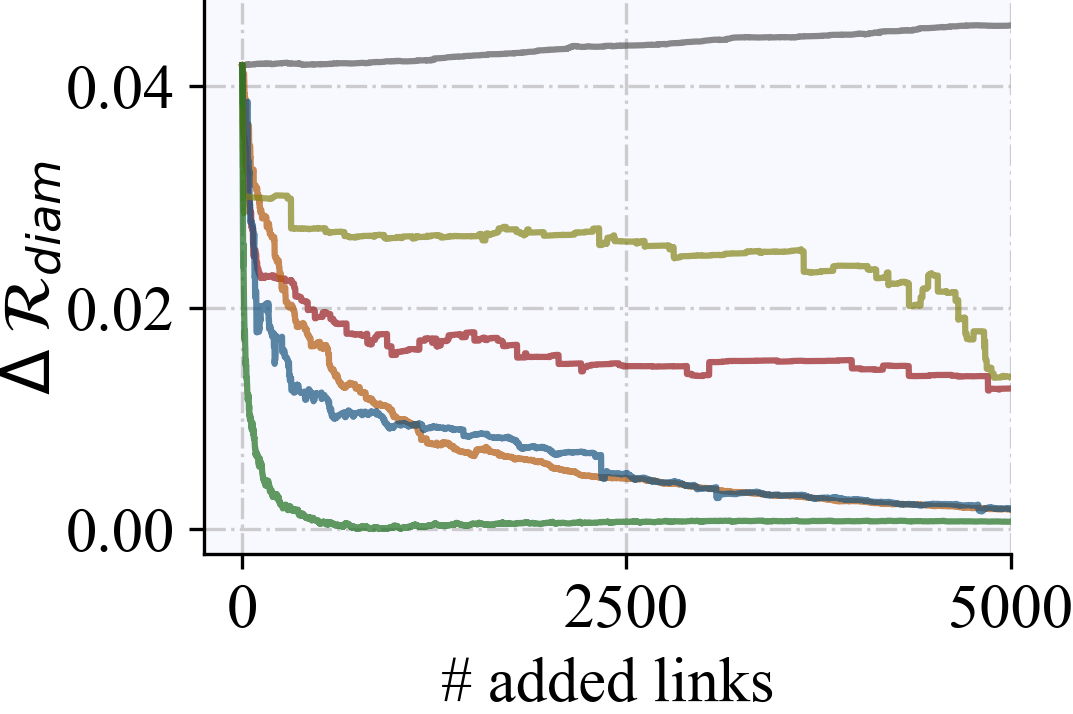}
    \end{subfigure}
    \hfill
    \begin{subfigure}{0.27\textwidth}
    \centering
    \includegraphics[width=\linewidth]{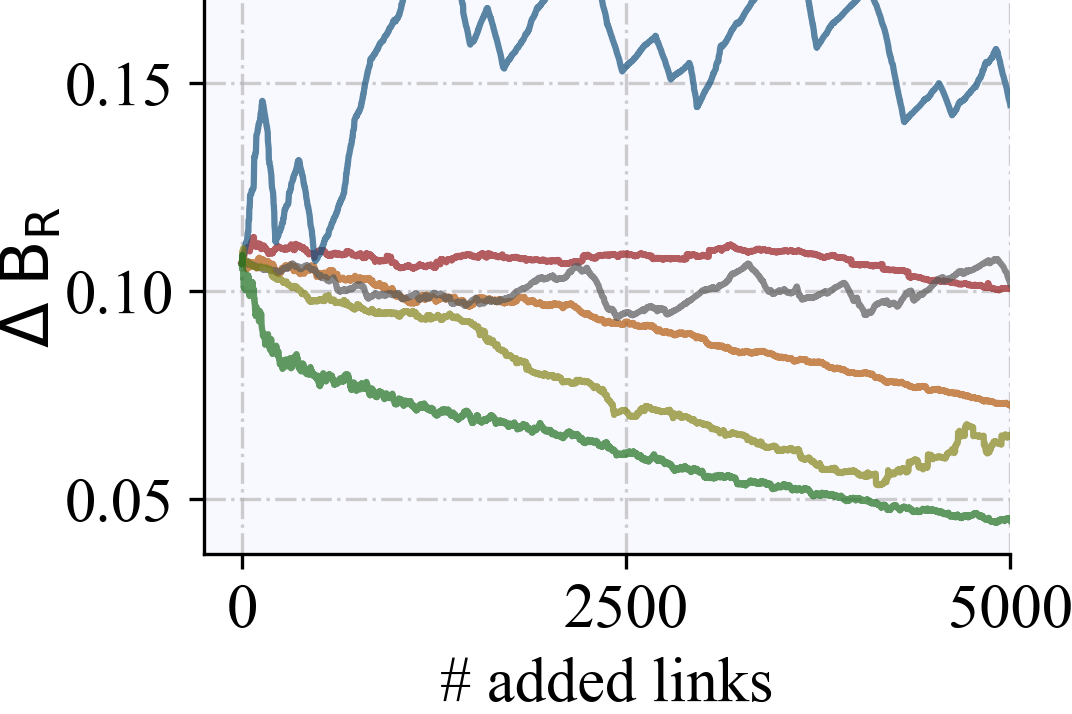}
    \end{subfigure}
    \caption{Evolution of group social capital metrics for both groups and fairness metrics as the number of added links increases, on Facebook dataset after adding 5,000 edges. Edge augmentation via \ERG{} exhibits a faster rate of convergence to the optimal scenario than the baselines.}
    \label{fig:fb5000evolution}
\end{figure*}

\begin{figure*}[ht]
    \begin{subfigure}{0.37\textwidth}
    \centering
    \includegraphics[width=\linewidth]{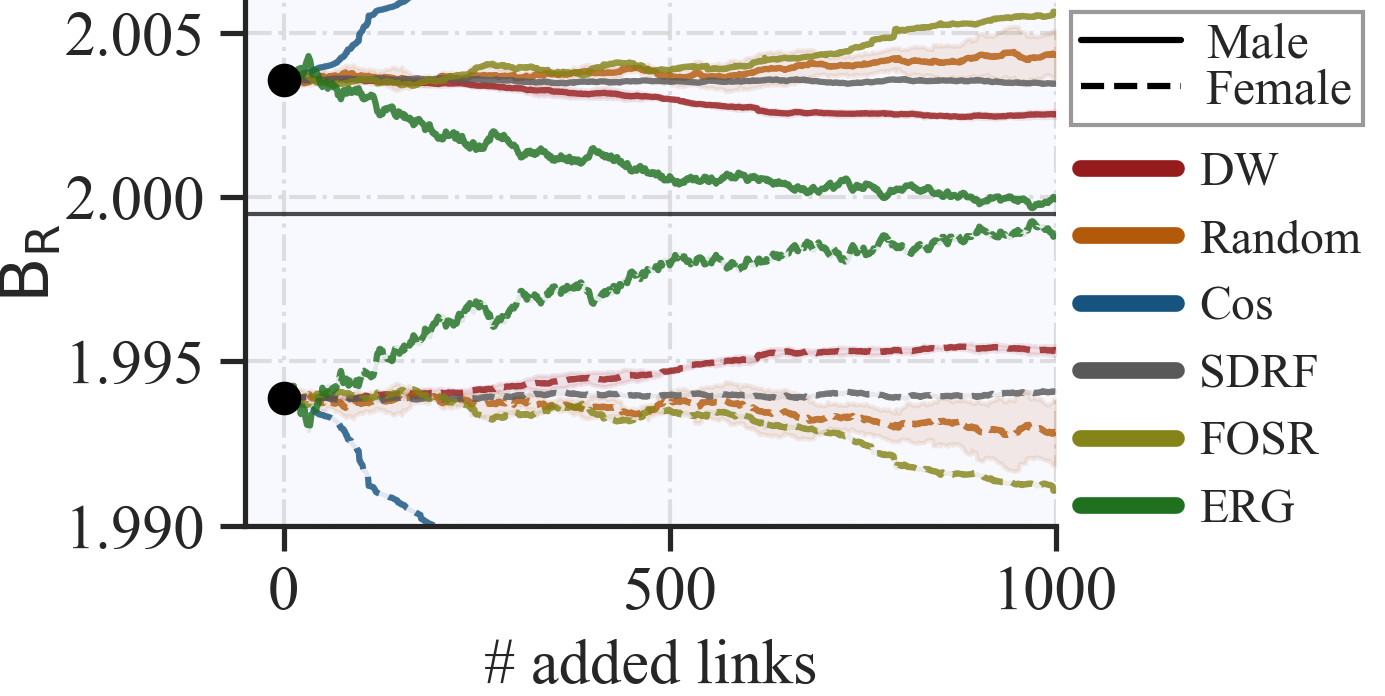}
    \end{subfigure}
    \hfill
    \begin{subfigure}{0.28\textwidth}
    \centering
    \includegraphics[width=\linewidth]{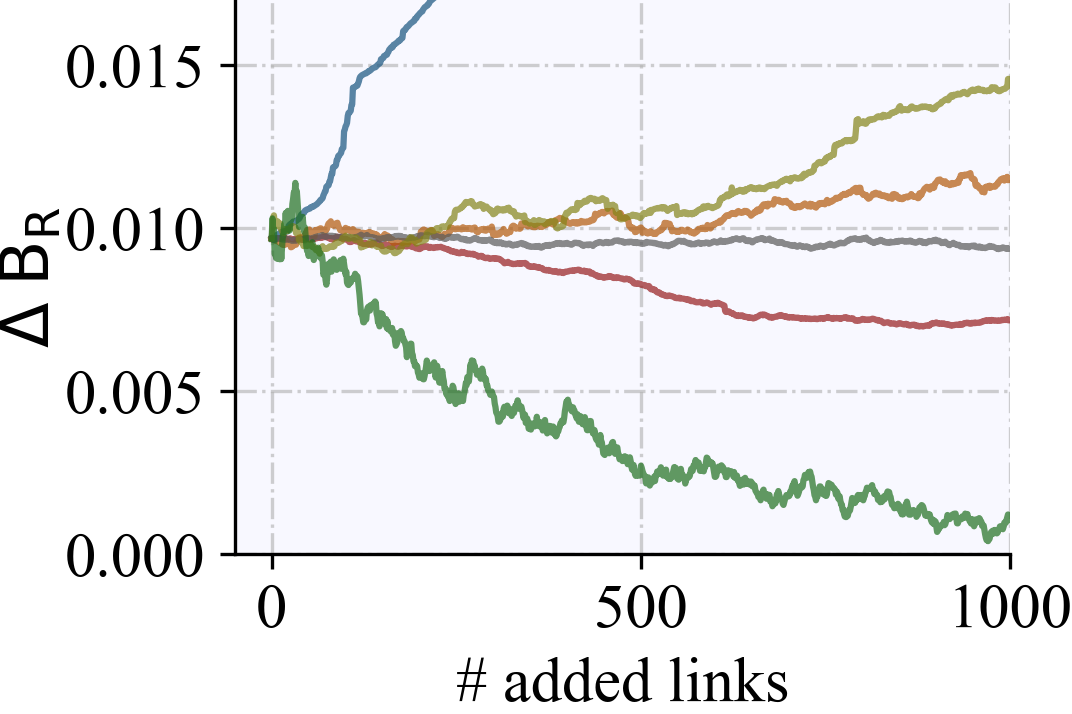}
    \end{subfigure}
    \hfill
    \begin{subfigure}{0.32\textwidth}
    \centering
    \includegraphics[width=\linewidth]{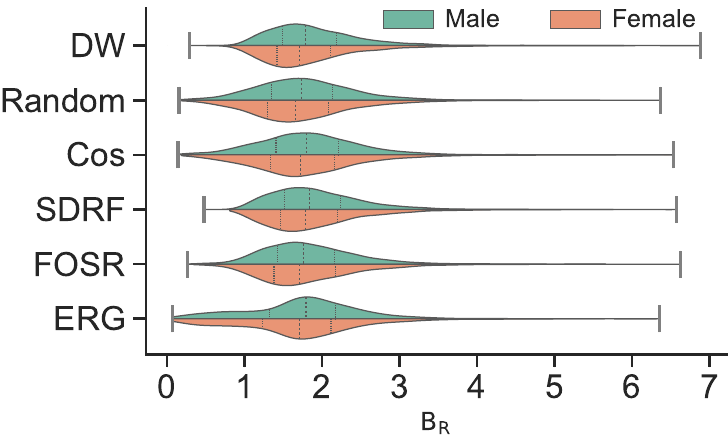}
    \end{subfigure}
    \caption{Behavior of node, group control and control disparity ($\Bt(u), \Bt(S_i), \Delta\Bt$). Illustration of the evolution of the group control (left, $\Bt(S_i)$) and control disparity (middle, $\Delta\Bt$) through an edge augmentation with $B = 1,000$ edges on the UNC dataset. Right-most figure: distribution of the nodes' control ($\Bt(u)$) for each group. Note how edge augmentation via \ERG{} yields a control for all groups approaching their optimal value of $2-2/|\mathcal{V}|$ by reducing the control of the privileged group (males) while increasing the control of the disadvantaged group (females).}
    \label{fig:unc1000evolution}
\end{figure*}

\subsection{Additional Edge Augmentation Examples}
\label{app:sec:rewiredgraphs}

\begin{figure*}[ht]
\centering
    \begin{subfigure}{\textwidth}
    \centering
    \includegraphics[width=\linewidth]{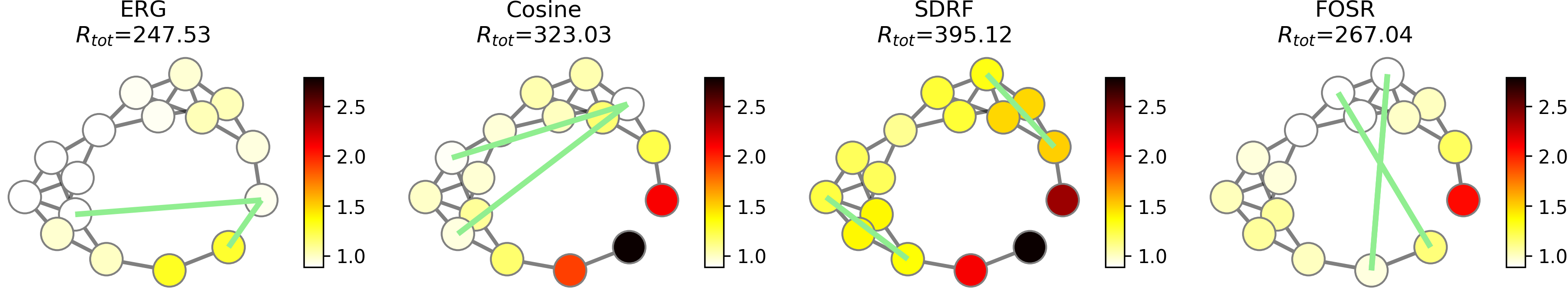}
    \caption{Connected Barbell Graph}
    \end{subfigure}
    
    \begin{subfigure}{\textwidth}
    \centering
    \includegraphics[width=\linewidth]{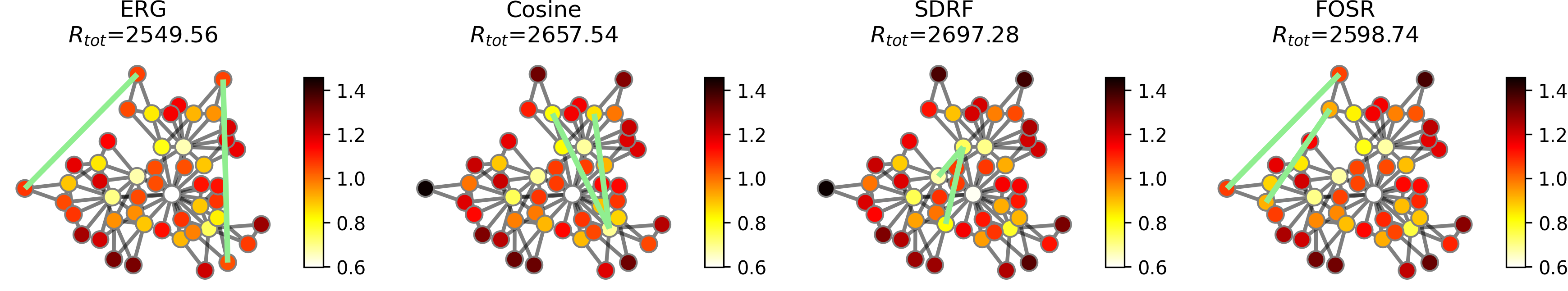}
    \caption{PCG(n=50, m=2, p=0.95)}
    \end{subfigure}
    
    \begin{subfigure}{\textwidth}
    \centering
    \includegraphics[width=\linewidth]{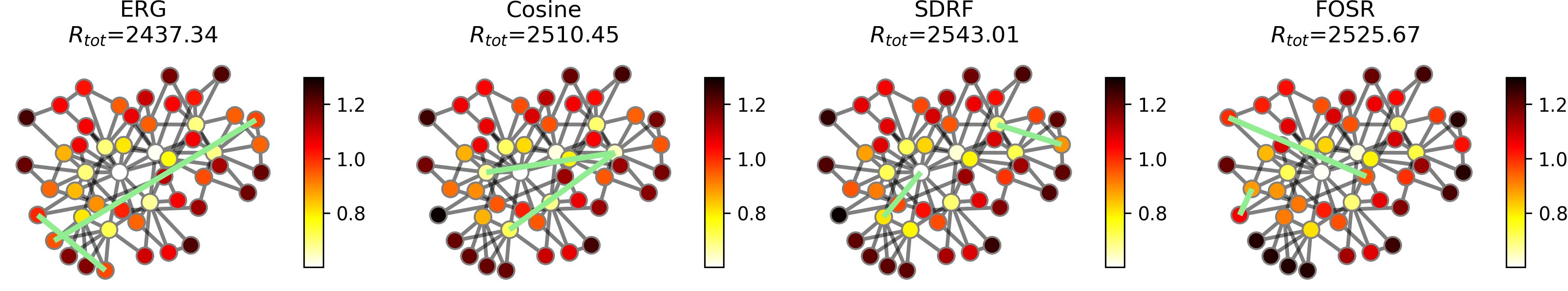}
    \caption{PCG(n=50, m=2, p=0.8)}
    \end{subfigure}

    \begin{subfigure}{\textwidth}
    \centering
    \includegraphics[width=\linewidth]{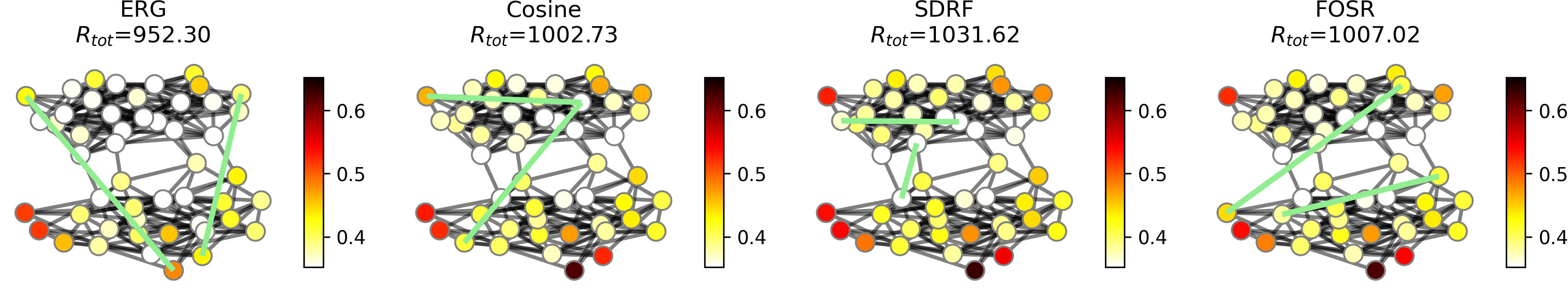}
    \caption{SBM(intra=0.5, inter=0.1)}
    \end{subfigure}
    \caption{Examples of the added links by the \texttt{ERG-Link}, Cosine, SDRF and FOSR algorithms on synthetic graphs. Nodes are colored by $\Rg(u)$. Edge augmentation via \texttt{ERG-Link} consistently yields the largest improvement in information flow in the graph and hence the smallest values of $R_{tot}$.}
    \label{fig:rew2}
\end{figure*}

For completeness, we qualitatively illustrate on Figures~\ref{fig:rew} and~\ref{fig:rew2} the behavior of four edge augmentation methods on five synthetic graphs. The Figures depict the synthetic graphs and the first two edges added by each algorithm. They also show quantitatively how the methods are able to improve the overall flow of information in the graph, defined as $\Rg=\sum_{(u,v)\in \mathcal{E} R_{uv}}$. This analysis not only provides insight into the performance of the edge augmentation methods, but also demonstrates their adaptability to different graph topologies.

The Figures correspond to graphs from five synthetic datasets, each designed to represent different graph characteristics: (1) a Barbell graph with Barbell size of 4 and path size of 5; (2) a power-law clustered graph with 50 nodes, 2 random new edges per each new node, and a 0.95 probability of closing a triangle; (3) the same graph, but with a 0.8 probability of closing a triangle, meaning that the bottleneck and clustered structures are weaker; (4) an SBM model with intra-cluster probability of 0.4 and inter-cluster probability of 0.01; and (5) a hand-crafted path graph with communities in the middle to simulate a different type of bottleneck. The Barbell graph is shown in Figure~\ref{fig:rew} and the rest of the graphs in Figure~\ref{fig:rew2}.

The Figures illustrate the results for all algorithms described in Section~\ref{sec:experiments} except for DeepWalk because its behavior closely resembles randomness. This similarity arises from the impracticality of performing extensive training after each edge addition, given the computational limitations of the proposed task. Performing extensive training after each edge addition would be computationally infeasible within the scope of our study.

Observe how \ERG{} always reduces $\Rg$ the most by adding the edges with the maximum distance. Other approaches, such as FOSR, do not exhibit a consistent performance across different graph topologies. In fact, graphs with complex structures and/or less defined bottlenecks exhibit a higher degree of difficulty for this method. Edge augmentation via Cosine similarity suffers from the same disadvantage. SDRF adds edges around the graph's bottleneck to avoid over-squeezing, but this strategy does not necessarily translate into an improvement of the graph's information flow.

\section{Computation Time}
\label{app:computation}

Table~\ref{tab:computationtime} depicts the number of edge additions per second that each algorithm is able to perform. The larger the number, the faster the method. For instance, in the conducted experiment on the Facebook dataset with $B = 5,000$ edges, edge augmentation via \ERG{} required approximately 3 minutes and 30 seconds; via Cosine Similarity (Cos), approximately 6 minutes; and by means of DeepWalk (DW), over 3 hours. In the case of the Google+ dataset, edge augmentation via \ERG{} required 36 minutes; via Cosine Similarity (Cos), 72 minutes; and by means of DeepWalk (DW), over 17 hours.

\begin{table}[ht]
    \centering
    \begin{tabular}{r|ccc}
    \toprule
    Dataset          & Facebook & Google+    & UNC28 \\
    $|V|$            & 1,034    & 3,508     & 3,985 \\
    $|\mathcal{E}|$ & 26,749 & 253,930 & 65,287 \\
    Density & 0.05 & 0.04 & 0.008 \\
    \midrule
     SDRF            & 1.00 & 0.25  & 0.46\\
     FOSR            & 166.6 & 125.0 & 83.3 \\
     DeepWalk        & 0.39 &  0.08 & 0.07\\
     Cosine          & 13.9 &  1.15 & 1.39\\
     ERG             & 23.8 &  2.31 & 1.66 \\
    \bottomrule
    \end{tabular}
    \caption{Added edges per second on the three datasets}
    \label{tab:computationtime}
\end{table}

\begin{table*}[!ht]
\centering
\begin{tabularx}{\textwidth}{p{0.1\textwidth}Xp{0.3\textwidth}}%
\toprule
\textbf{Symbol} & \textbf{Description} & \textbf{Definition}\\ 
\midrule
$G=\mathcal{(V,E)}$     & Graph = (Nodes, Edges) \\
$n=|\mathcal{V}|$                 & Number of nodes \\
$\mathbf{A}$            & Adjacency matrix: $\mathbf{A} \in \mathbb{R}^{n\times n}$ & $A_{u,v}=1$ if $(u,v)\in \mathcal{E}$ and $0$ otherwise\\
$v$ or $u$              & Node $v \in \mathcal{V}$ or $u \in\mathcal{V}$\\
$e$                     & Edge $e \in \mathcal{E}$ \\
$d_v$                   & Degree of node, i.e., number of neighbors $v$ & $d_v = \sum_{u\in V} A_{v,u}$\\
$\mathbf{D}$            & Degree diagonal matrix where $d_v$ in $D_{vv}$ & $\mathbf{D} = \text{diag}(d_0,\ldots, d_{\mathcal{V}})$ \\
$\text{vol}(G)$         & Sum of the degrees of the graph & $\text{vol}(G) = \sum_{u\in V} d_u = 2|\mathcal{E}|= \text{Tr}[\textbf{D}]$ \\
$\mathcal{N}(u)$        & Neighbors of $u$ & $\mathcal{N}(u) = \{v: (u,v)\in \mathcal{E}\}$\\
$S_i$                     & Subset of nodes & $S\subseteq V$\\
\midrule
$SA$ & Set of sensitive attributes &  $SA=\{sa_1, sa_2, \dots, sa_{|SA|}\}$\\ 
$SA(v)$ & Value of the sensitive attribute of the node $v$\\
$sa_i$ & Specific value of a sensitive attribute, e.g., $sa_i$=\textit{female}. & \\
$S_i$ & Set of nodes defined by their sensitive attribute & $S_i =  \{v \in V | SA(v) = sa_i \}$ \\
$S_d$ & Set of nodes defined by their sensitive attribute with the highest level of isolation $\Rg(S_i)$ & \\
\midrule
$\mathbf{L}$            & Graph Laplacian & $\mathbf{L=D-A=\Phi \Lambda \Phi}^\top$ \\
$\mathbf{\Lambda}$               & Eigenvalue matrix of $\mathbf{L}$ \\
$\Phi$                  & Matrix of eigenvectors of $\mathbf{L}$\\
$\lambda_i$             & The $i$-th smallest eigenvalue of $\mathbf{L}$ \\
$\mathbf{f}_i$          & Eigenvector associated with the $i$-th smallest eigenvalue of $\mathbf{L}$\\
$\mathbf{L}^+$          & The pseudo-inverse of $\mathbf{L}$ & $\mathbf{L}^+= \sum_{i>1}  \lambda_i^{-1}\phi \phi^\top$\\
$h_G$                   & Cheeger constant & Eq.~\ref{eq:cheeger}\\
\midrule
$\mathbf{e}_u$          & Unit vector with unit value at $u$ and 0 elsewhere \\
$R_{uv}$                & Effective resistance between nodes $u$ and $v$ & $R_{uv} = (\mathbf{e}_u-\mathbf{e}_v)\mathbf{L}^+(\mathbf{e}_u-\mathbf{e}_v)$\\
$\mathbf{R}$ & Effective resistance matrix where the $i,j$ entry corresponds to $R_{ij}$& $\mathbf{R}=\mathbf{1} \text{diag}(\mathbf{L}^+)^\top + \text{diag}(\mathbf{L}^+)\mathbf{1}^\top- 2\mathbf{L}^+$\\
$\mathbf{Z}$ & Commute Time Embedding matrix & $\mathbf{Z} = \sqrt{vol(G)}\mathbf{\Lambda^{-1/2}}\mathbf{\Phi}^\top$\\
$\mathbf{z}_u$          & Commute times embedding of node $Z_{u,:}$  \\
$\CT(u,v)$             & Commute time & $\CT(u,v)=\text{vol}(G)R_{u.v}$ \\
\midrule
$\Rg$               & Total Effective Resistance of $G$ & $\Rg =\frac{1}{2} \mathbf{1}^\top \mathbf{R1}$\\
$\Rd$               & Resistance Diameter of $G$ & $\Rd = \max_{u,v\in \mathcal{V}} R_{u,v}$\\
$\Rg(u)$            & Node Isolation or Total Effective Resistance & $\Rg(u)=\sum_{v\in V} R_{uv}$\\
$\Rd(u)$            & Node Resistance Diameter & $\Rd(u) = \max_{v\in V} R_{uv}$\\
$\Bt(u)$            & Node Control or Resistance Betweenness & $\Bt(u) = \sum_{v \in \mathcal{N}(u)} R_{uv}$\\
$\Rg(S_i)$            & Group Isolation or Total Effective Resistance & $\Rg(S_i) = |S|^{-1} \sum_{u\in S} \Rg(u)$\\
$\Rd(S_i)$            & Group Resistance Diameter & $\Rd(S) = |S|^{-1}\sum_{u \in S} \Rd(u)$\\
$\Bt(S_i)$            & Group Control or average Betweenness & $\Bt(S) = |S|^{-1}\sum_{u \in S} \Bt(u)$\\
\bottomrule
\end{tabularx}
\caption{Table of Notation}
\label{tab:notation}
\end{table*}

\end{document}